\documentclass[preprint, prX]{revtex4}

\usepackage{ifthen}
\usepackage{graphicx,psfrag}
\usepackage{mathrsfs}

\usepackage{epsfig}
\usepackage{hyperref}
\usepackage{amsmath}
\usepackage{amsthm}
\usepackage{amssymb}
\usepackage{graphicx}
\usepackage{verbatim}

\def\p{\partial}

\newlength\rringshift 

\newcommand*\doublemathring[1]{%
\setbox0=\hbox{$#1$}%
\dimen0=\wd0
\advance\dimen0 - \rringshift
\wd0=\dimen0
\mathring{\copy0}%
\kern-\wd0
\kern + \rringshift
\mathring{\phantom{\copy0}}%
\kern-0.0pt
}

\newcommand{\lessthansimilarto}{\lower3pt\hbox{$\buildrel{<}\over{\sim}$}}
\newcommand{\greaterthansimilarto}{\lower3pt\hbox{$\buildrel{>}\over{\sim}$}}

\newcommand{\RR}{\hbox{$I$\kern-3.8pt $R$}}

\newcommand{\norm}[1]{\left|\left|{#1}\right|\right|}
\newcommand{\pb}[2]{\left[{#1},{#2}\right]}
\newcommand{\varderiv}[2]{\frac{\delta {#1}}{\delta {#2}}}

\newcommand{\define}[1]{\emph{#1}}

\newtheoremstyle{jThm}
  {9pt}
  {9pt}
  {\itshape}
  {}
  {\bfseries}
  {:}
  {.5em}
  {}

\theoremstyle{jThm}
\newtheorem{thm}{Theorem}[section]

\theoremstyle{jThm}

\theoremstyle{jThm}
\newtheorem{cor}[thm]{Corollary}

\begin{document}

\title{Gauge Fixing and Constrained Dynamics}
\author{Jon Allen}
\author{Richard A. Matzner}
\affiliation{Theory Group, University of Texas at Austin, Austin, Texas, 78712}
\date{\today}
%

\begin{abstract}
We review the Dirac formalism for dealing with constraints in a canonical Hamiltonian formulation and discuss gauge freedom and display constraints for gauge theories in a general context. We introduce the Dirac bracket and show that it provides a consistent method to remove any gauge freedom present. We discuss stability in evolution of gauge theories and show that fixing all gauge freedom is sufficient to ensure well-posedness for a large class of gauge theories. 
Electrodynamics provides examples of the methods outlined for general gauge theories. Future work will apply the formalism, and results derived here, to General Relativity.
\end{abstract}


\maketitle
\tableofcontents

\section{Introduction}
\label{sec:introduction}
In the mid twentieth century,  Dirac \cite{dirac1}, Arnowitt, Deser, and Misner \cite{adm0}, and many others, formulated General Relativity (``G.R.") as a dynamical system. Computational evolution of the dynamical equations was almost immediately attempted \cite{hahn_lindquist}. However, though a dynamical formulation of G.R. had been constructed, once even modest computational resources were available, it became clear that multi-dimensional computation (e.g. $(2+1)d$) in strong gravitational fields was consistently hampered by instabilities -- exponential growth in errors. Eventually a number of formulations were constructed within the last two decades that exhibit enough long term stability to evolve binary black hole systems. The goals of this work are to understand how instabilities arise, how the stable computational formulations are related, and how new stable formulations can be constructed in general, using the framework of constrained Hamiltonian dynamics and gauge theories. 

With recent observations of gravitational waves \cite{grligo}, and with particular examples of stable evolution equations \cite{bssn2},\cite{nor}, a generalized form for assessing, and generating, long-term computational solutions becomes important as more complex systems are considered which incorporate additional fields, initial configurations, physical extremes, and different background geometries.

Stability, as considered here, is defined with respect to the growth of deviations between solutions as the system evolves. Although this work focuses exclusively on physical theories in the classical regime, many of the techniques developed for quantization readily lend themselves to similar study of stability.

This work begins with a review of the Hamiltonian formulation for constrained systems, as originally introduced by Dirac.  Well-posedness, initial data,  stability, and error in  evolution are reviewed within the Hamiltonian framework. Within this framework, we introduce a novel general procedure for constructing formulations of a broad class of physical theories, including G.R. We give a new proof for sufficient criteria to ensure stability and well-posedness, involving not just the evolution but also the initial data. Using this framework and geometrical motivations, a novel method is then introduced for the removal of numerical error from computations for a broad class of physical theories including gauge theories. Finally, electrodynamics is examined within the framework developed to provide a concrete example in a familiar, simple, setting. 

In concluding remarks we give retails relating to G.R. in particular, and the relationship between quantization and numerical stability, though a more thorough treatment left for subsequent work.

This work attempts a logically consistent development, and though a certain amount of formalism is inevitable, our approach is essentially that of a heuristic physicist, with no attempt at rigorous proof.

Section~(\ref{sec:GaugeHamiltonian}) below presents a self-contained review of gauge theories discussed in the framework of constraint Hamiltonian systems. Section~(\ref{sec:stability}) covers topics dealing with stability and numerical error for Hamiltonian formulations of gauge theories. Section~(\ref{subsec:examples}) applies the concepts previously discussed for general gauge theories to Electrodynamics. As a cautionary tale, we find that even this extremely well-studied linear theory has some subtleties in its Dirac formulation.

\subsection{Notation}
\label{sec:introduction:notation}
The following notation will be used throughout, unless otherwise specifically noted. 

Einstein's summation convention, $V_\alpha W^\alpha \equiv \sum_{\alpha} V_{\alpha}W^{\alpha}$, will apply to all indexed terms. 

When working with a spatial manifold, the components of a vector in a three dimensions will be labeled by lowercase Roman indexes from the middle of the alphabet, $i,j,k,l,m = \left\{1,\dots,3\right\}$, with each label running over the three spatial coordinates. The components of vectors in a full four dimensional spacetime will be labeled by the Greek indexes from the beginning of the alphabet, $\alpha,\beta,\gamma,\delta,\epsilon = \left\{0,\dots,3\right\}$, with each label running over all four space-time coordinates. We often use a comma to indicate a partial derivative: $q,_\alpha \equiv \partial q / \partial x^\alpha $ for any labeled variable $x^\alpha $ on which $q$ depends. 

When working with bundles, lowercase Roman indexes from the beginning of the alphabet will run over each bundle coordinate pair, $a,b,c = \left\{1,\dots,N\right\}$, where $N$ is the dimension of the base space, thereby labeling the bundle coordinates. These labels will be used for various bundles throughout, with the bundle space being clear from context.

Although  the $2N$ dimensional phase space is itself a cotangent bundle, the computational coordinates on phase space will be distinguished from general bundle coordinates by labeling all bundle coordinates individually using uppercase Roman indexes from the middle of the alphabet, $I,J,K,L,M = \left\{1,\dots,2N\right\}$,  running over the $N$ dimensions of configuration space and $N$ dimensions of momentum space. In abstract notation, contravariant vectors will be denoted in bold, $\mathbf{X} = X^a \otimes \mathbf{e}_{a}$, and covariant vectors will be denoted in bold with a tilde, $\tilde{\mathbf{X}} = X_a \otimes \sigma^{a}$, for basis vectors $\mathbf{e}_{a}$ and dual basis $1-$forms $\sigma^{a}$. To distinguish phase space coordinates labeling configuration space components from momenta, configuration space vectors will be designated by the Roman character $q$ while momenta will be designated by the Roman character $p$.

In the space of constraints, introduced in subsection~(\ref{subsec:constraints}), vector components will be labeled by uppercase Roman indexes from the beginning of the alphabet, $A,B,C,D = \left\{1,\dots,M\right\}$, running over all $M$ constraints. When discussing multiple constraint spaces, no confusion should arise as each constraint space being labeled should always be clear from context.

%
\section{Review of Gauge Theories and Hamiltonian Systems}
\label{sec:GaugeHamiltonian}
We begin by providing a relatively self-contained introduction to the Hamiltonian formulation of gauge theories in general. The reader is assumed to have a basic familiarity with the Lagrangian and Hamiltonian formulations along with the methods of variational calculus. Familiarity with exterior differential calculus will be occasionally assumed but is not necessary in order for the casual reader to follow along. For a more thorough treatment of the topics reviewed in this section, see \cite{frankel1},\cite{goldstein},\cite{nakahara},\cite{morrison1},\cite{rovelli1}, and \cite{jorge1}. For a much more thorough treatment of general Hamiltonian formulations with gauge freedom, see \cite{teitel1}. To simplify the presentation throughout this section, unless otherwise noted,  all theories will be expressed as finite dimensional systems and all boundary contributions will be neglected. This section is meant to serve only as a basic introduction to the material, leaving until later sections a more thorough review of advanced topics, such as those necessary for the treatment of G.R.

Beginning with a review of Lagrangian dynamics for systems yielding a bijective Legendre transform, we derive the relation between the Lagrangian and Hamiltonian formulations, and  examine the canonical Hamiltonian formulation. Then we address the more difficult case of Lagrangian formulations which generate a singular Legendre transform, leading to a discussion of constrained Hamiltonian formulations. We define and review first class and second class constraints in such systems, as well as  gauge freedom and physical observables in Hamiltonian formulations. 

We show that second class constraints can be imposed which uniquely fix a choice of gauge, thereby removing all gauge freedom present in the system, without affecting the physical observables of the system. Finally, we review the Dirac bracket and discuss gauge fixed Hamiltonian dynamics.
%
\subsection{Lagrangian and Hamiltonian Dynamics}
\label{subsec:lagHamDyn}
For a given physical system modeled by a set of variables and their derivatives, $\left\{~\mathbf{q},\mathbf{q}_{,\alpha},\mathbf{q}_{,\alpha\beta},\dots~\right\}$, the \define{action}, $S\left[\mathbf{q}\right]$, is defined to be the \define{functional}, 
\begin{align}
\label{subsec:prelim:action}
S\left[\mathbf{q}\right] \equiv \int dt ~L\left[\mathbf{q},\dot{\mathbf{q}}\right]
\end{align}
where the overdot means total derivarive w.r.t. the parameter, here $t$. When extremized, equation~(\ref{subsec:prelim:action}) yields the equations of motion for each of the variables $\mathbf{q}$.  The space of variables is called the \define{configuration space}, $M$, and the velocities, $\dot{\mathbf{q}}$, at the location $\mathbf{q} \in M$ reside in the \define{tangent space} of $M$, $T_q M$. The differentiable space of all velocities at all points over $M$ is known as the \define{tangent bundle}, $TM = \left\{T_q M | q \in M\right\}$, and has coordinates $\left(\mathbf{q},\dot{\mathbf{q}}\right) \in TM$. The functional in the integrand of equation~(\ref{subsec:prelim:action}) , $L\left[\mathbf{q},\dot{\mathbf{q}}\right]$, is called the \define{Lagrangian} and is a real valued functional of the tangent bundle coordinates, $L: TM\rightarrow \mathbb{R}$. The value returned by the Lagrangian is a scalar and therefore will be independent of the chosen coordinate system on the tangent bundle. Extremization of the action yields 
\begin{align}
\label{subsec:prelim:actionExtremum}
\delta S\left[\mathbf{q},\delta\mathbf{q}\right] = \int dt~ \left\{\left[\frac{\p L}{\p q^a} - \frac{d}{dt} \frac{\p L}{\p \dot{q}^a}\right]\delta q^a +\frac{d}{dt}\left[\frac{\p L}{\p \dot{q}^a}\delta q^a\right]\right\}= 0
\end{align}
Unless otherwise noted, assume that the variation, $\delta\mathbf{q}$, at the boundary takes the form
\begin{align}
\label{subsec:prelim:actionBoundary}
\frac{\p L}{\p \dot{q}^a}\delta q^a = C
\end{align}
for some constant $C \in \mathbb{R}$ so that the last term of equation~(\ref{subsec:prelim:actionExtremum}), being a total derivative, vanishes. The resulting extremized path yields the \define{Euler-Lagrange equations} 
\begin{align}
\label{subsec:prelim:euler}
\frac{\p L}{\p q^a} - \frac{d}{dt} \frac{\p L}{\p \dot{q}^a} = 0
\end{align}
The Euler-Lagrange equations form a second order differential system which govern the evolution of the physical system modeled by the configuration space variables, $\mathbf{q}$. For field theories, the Lagrangian, $L$, is integrated over all of spacetime to yield the action, $S$, so the total derivative in equation~(\ref{subsec:prelim:actionExtremum}) for the finite dimensional system becomes an integral over the spacetime boundary in the continuum.

The configuration space manifold, $M$, and tangent bundle, $TM$, are manifolds with definitions which are independent of the Lagrangian, $L$, or coordinate system, $\left(\mathbf{q},\dot{\mathbf{q}}\right)$. Because these spaces are defined without respect to the dynamics, two physical systems, modeled by two different Lagrangians, $L$ and $L^\prime$, can be defined on the same configuration space, $M$, and tangent bundle, $TM$. Since these manifolds, $M$ and $TM$, are independent of the dynamics, there is no natural way to define an intrinsic meaning for the velocities at a given location, $\dot{\mathbf{q}} \in T_qM$, and therefore no natural way to compare these values for two distinct locations, $\mathbf{q},\mathbf{q}^\prime \in M$ with $\mathbf{q}\neq \mathbf{q}^\prime$. Formally then, velocities have no intrinsic meaning because there is no canonical inner product structure on the tangent bundle, a necessary requirement in order to be able compare elements of two distinct tangent spaces, $T_qM$ and $T_{q^\prime}M$, in a coordinate independent manner. Although there is no natural way to compare elements of $TM$ in general, when a particular Lagrangian, $L$, is considered, that Lagrangian itself can be used to define a map from the tangent bundle, $TM$, to the dual space, $T^\star M$, of all $1-$forms over $M$, known as the \define{cotangent bundle}. The map from the tangent bundle, $TM$, to the cotangent bundle, $T^\star M$, is the \define{Legendre transform}, defined as
\begin{align}
\label{subsec:prelim:momenta}
\tilde{\mathbf{p}}\left(\mathbf{q},\dot{\mathbf{q}}\right) \equiv \varderiv{L\left[\mathbf{q},\dot{\mathbf{q}}\right]}{\dot{\mathbf{q}}}
\end{align}
taking the tangent bundle coordinates, $\left(\mathbf{q},\dot{\mathbf{q}}\right)$, into coordinates on the cotangent bundle, $\left(\mathbf{q},\tilde{\mathbf{p}}\right)$, with $\tilde{\mathbf{p}}$ being $1-$forms which are dual to the velocities, $\dot{\mathbf{q}}$. In mechanics, the coordinates of $T^\star M$ defined by equation~(\ref{subsec:prelim:momenta}), $\tilde{\mathbf{p}}$, are known as \define{canonical momenta}, and the collection of all  cotangent bundle coordinates, $\left(\mathbf{q},\tilde{\mathbf{p}}\right) \in T^\star M$, defines the \define{phase space}. For an $N$ dimensional configuration manifold, the phase space will be a $2N$ dimensional manifold with coordinates, $\left\{q^a,p_a\right\}$ for $a \in 1\dots N$, defined by the $N$ \define{conjugate pairs}. Canonical momenta, $p_a$, are often referred to as \define{conjugate momenta} with respect to the position variable, $q^a$, with which it forms the conjugate pair, $\left(q^a,p_a\right)$. Each conjugate pair present defines a single \define{degree of freedom} for the physical system, so that a $2N$ dimensional phase space has $N$ degrees of freedom. 

Once elements of the tangent bundle, $TM$, can be identified with elements of the cotangent bundle, $T^\star M$, an inner product can be constructed over $M$ by defining the norm, $\norm{\cdot}$, as
\begin{align}
\label{subsec:prelim:norm}
\norm{\dot{\mathbf{q}}} \equiv \varderiv{L\left[\mathbf{q},\dot{\mathbf{q}}\right]}{\dot{\mathbf{q}}} \cdot \dot{\mathbf{q}} = p_a\dot{q}^a \equiv \tilde{\mathbf{p}}\cdot \dot{\mathbf{q}}
\end{align}
which sends elements of the tangent bundle to scalar values, $TM \rightarrow \mathbb{R}$. Because the Lagrangian, $L$, is invariant under changes of the configuration space coordinates, and subsequent changes in the tangent bundle coordinates, the norm, equation~(\ref{subsec:prelim:norm}), will also be invariant under coordinate changes. This coordinate invariant value is well defined, for a given Lagrangian, and can therefore be used to make meaningful comparisons of velocities, $\dot{\mathbf{q}}$, in a coordinate independent way. In mechanics, the norm, $\norm{\dot{\mathbf{q}}}$, is equal to twice the \define{kinetic energy}, $T$
\begin{align}
\label{subsec:prelim:kineticTerm}
T \equiv \frac{1}{2}\norm{\dot{\mathbf{q}}} = \frac{1}{2}\tilde{\mathbf{p}}\left(\mathbf{q},\dot{\mathbf{q}}\right) \cdot \dot{\mathbf{q}} 
\end{align}
which should be a familiar physical quantity, invariant under changes of the configuration space coordinates.  Assume for the remainder of this subsection that the Legendre transform, equation~(\ref{subsec:prelim:momenta}), is a \define{bijection}, mapping unique elements of $TM$ to unique elements of $T^\star M$. (The case in which the Legendre transformation is not a bijection will be examined in subsection~(\ref{subsec:gauge}).) When the Legendre transform is a bijection, an inverse map exists allowing unique elements of the tangent bundle, $\left(\mathbf{q},\dot{\mathbf{q}}\right) \in TM$, to be written as unique expressions of the cotangent bundle coordinates, $\left(\mathbf{q},\tilde{\mathbf{p}}\right) \in T^\star M$. Expressing all velocities, $\dot{\mathbf{q}}$, as functions of the phase space coordinates, $\left(\mathbf{q},\tilde{\mathbf{p}}\right)$, allows the dynamics to be expressed entirely in phase space. 

Define the \define{canonical Hamiltonian} as
\begin{align}
\label{subsec:prelim:hamiltonian}
H \equiv \tilde{\mathbf{p}}\cdot \dot{\mathbf{q}} - L\left[\mathbf{q},\dot{\mathbf{q}}\right]
\end{align}
Treating the coordinate components of $\mathbf{q}$, $\dot{\mathbf{q}}$, and $\tilde{\mathbf{p}}$ independently, the total variation of $H$ yields
\begin{align}
\label{subsec:prelim:totalVarHam}
\delta H = \dot{q}^a\delta p_a - \frac{\delta L}{\delta q^a}\delta q^a + \left(p_a - \frac{\delta L}{\delta \dot{q}^{a}}\right)\delta \dot{q}^{a}
\end{align}
Using the definition of the momenta, equation~(\ref{subsec:prelim:momenta}), the coefficient of $\delta \dot{q}^{a}$ vanishes identically, showing that the canonical Hamiltonian, $H$, is independent of the velocities $\dot{\mathbf{q}}$. Using the canonical Hamiltonian, equation~(\ref{subsec:prelim:hamiltonian}), define the \define{canonical action} as the functional of the phase space variables, $\left(\mathbf{q},\tilde{\mathbf{p}}\right)$, given by
\begin{align}
\label{subsec:prelim:canonicalAction}
S\left[\mathbf{q},\tilde{\mathbf{p}}\right] \equiv 
\int \left(p_a dq^a - H\left[\mathbf{q},\tilde{\mathbf{p}}\right]dt\right) = \int  dt ~ L\left[\mathbf{q},\dot{\mathbf{q}}\left(\mathbf{q},\tilde{\mathbf{p}}\right)\right]
\end{align}
Extremizing the canonical action yields 
\begin{align}
\label{subsec:prelim:canonicalExtremum}
\delta S\left[\mathbf{q},\tilde{\mathbf{p}},\delta \mathbf{q},\delta \tilde{\mathbf{p}} \right] \equiv 
\int dt~\left\{\left[\dot{q}^a - \frac{\delta H}{\delta p_a}\right]\delta p_a - \left[\dot{p}_a + \frac{\delta H}{\delta q^a}\right]\delta q^a + \frac{d}{dt}\left[p_a\delta q^a\right]\right\} = 0
\end{align}
Using the definition of the momenta, equation~(\ref{subsec:prelim:momenta}), and the boundary condition placed on the variation $\delta\mathbf{q}$, equation~(\ref{subsec:prelim:actionBoundary}), the last term in equation~(\ref{subsec:prelim:canonicalExtremum}) vanishes. The resulting extremized path in phase space yields \define{Hamilton's equations}
\begin{align}
\label{subsec:prelim:momentaEOM}
\frac{d\tilde{\mathbf{p}}}{dt} &= - \frac{\delta H\left[\mathbf{q},\tilde{\mathbf{p}}\right]}{\delta \mathbf{q}}\\
\label{subsec:prelim:positionEOM}
\frac{d\mathbf{q}}{dt} &= \frac{\delta H\left[\mathbf{q},\tilde{\mathbf{p}}\right]}{\delta \tilde{\mathbf{p}}}
\end{align}
Note that no restriction on the variation of the momenta, $\delta \tilde{\mathbf{p}}$, at the boundary is necessary to extremize the canonical action. Using the Legendre transform, equation~(\ref{subsec:prelim:momenta}), along with canonical Hamiltonian, equation~(\ref{subsec:prelim:hamiltonian}),  and the canonical action, equation~(\ref{subsec:prelim:canonicalAction}), to express the extremized path in phase space coordinates, $\left(\mathbf{q},\tilde{\mathbf{p}}\right)$, as an extremized path in the tangent bundle coordinates, $\left(\mathbf{q},\dot{\mathbf{q}}\right)$, shows that the Lagrangian and Hamiltonian formulations are equivalent. Numerically, it is often more convenient to evolve Hamilton's equations, which form a first order differential system of $2N$ equations, than the $N$ second order differential system given by the Euler-Lagrange equations. 

Treating the $2N$ coordinates of the phase space independently allows an exterior calculus to be introduced on the cotangent bundle. In the space of  $1-$forms with coefficients taking values in phase space, define the \define{Poincar\'e $1-$form} as
\begin{align}
\label{subsec:prelim:poincareOneForm}
\tilde{\lambda} \equiv p_a dq^a
\end{align} 
Treating coordinate time, $t$, as a configuration space variable, the \define{Hamiltonian $1$-form} is defined as
\begin{align}
\label{subsec:prelim:hamiltonianOneForm}
\tilde{\Lambda} \equiv \tilde{\lambda} - H dt
\end{align} 
allowing the canonical action to be written as
\begin{align}
\label{subsec:prelim:actionOneForm}
S\left[\mathbf{q},\tilde{\mathbf{p}}\right] = \int dt~\tilde{\Lambda}
\end{align} 
In the space of $2-$forms with coefficients taking values in phase space, the \define{Poincar\'e $2-$form} is defined as the exterior derivative, in phase space, of the Poincar\'e $1-$form 
\begin{align}
\label{subsec:prelim:poincareTwoFormNonCompact}
\omega^2 \equiv d\tilde{\lambda} = dp_a \wedge dq^a
\end{align}
(Note that equations~(\ref{subsec:prelim:momentaEOM})~and~(\ref{subsec:prelim:positionEOM}) imply that $H$ is a function only of $t$ so $dH \wedge dt=0$.)   The Poincar\'e $2-$form, $\omega^2$, defines a \define{symplectic} structure in the phase space, $T^\star M$. A compact notation frequently used when dealing with Hamiltonian systems is given by writing the $2N$ phase space coordinates, $\left(\mathbf{q},\tilde{\mathbf{p}}\right)$, as
\begin{align}
\label{subsec:prelim:hamCondensedNotation}
z^{1} & \equiv q^{1}, \dots, z^{N} \equiv q^{N}\\\nonumber
z^{N+1} & \equiv p_{1}, \dots, z^{2N} \equiv p_{N}
\end{align}
The elements of the phase space coordinates, $\mathbf{z}$, will be denoted as $z^{K}$, with index, $K$, which runs over the $2N$ dimensions of $T^\star M$. In the compact notation, the Poincar\'e $2-$form of  equation~(\ref{subsec:prelim:poincareTwoFormNonCompact}) becomes
\begin{align}
\label{subsec:prelim:poincareTwoForm}
\omega^2 = J_{KL}~dz^{K}\wedge dz^{L}
\end{align}
In canonical phase space coordinates, $J_{KL}$, defines the \define{canonical form} given by
\begin{equation}
\label{subsec:prelim:canonicalForm}
\mathbf{J} \equiv J_{KL} = \frac{1}{2}\left(J_{KL} - J_{KL}\right) = \left[\begin{array}{cc}
\mathbf{0} & -\mathbb{I} \\
\\
\mathbb{I} & \mathbf{0} \\
\end{array}\right]
\end{equation}
with $\mathbb{I}$ being the $N\times N$ identity matrix. Transformations of the phase space coordinates which preserve the canonical form are known as \define{canonical transformations}. Any two canonical transformations can be combined to yield a third, and each canonical transformation is invertible, whence canonical transformations form a group.

Dual to the space of $1-$forms, resides the space of vectors with coefficients which take values in the phase space. Define a basis for this vector space, dual to the basis $1-$forms $d\mathbf{z}$, as the vectors $\p_a$ satisfying
\begin{align}
\label{subsec:prelim:vectBasis}
\p_K\left(dz^{L}\right) = dz^{L}\left(\p_K\right) = \delta^{L}_{K}
\end{align}
In general, the tangent bundle over $T^\star M$ will be the vector space defined as
\begin{align}
\label{subsec:prelim:phaseVectSpace}
\mathcal{V} \equiv \left\{\mathbf{X} = X^{K}~\p_{K}~ \vert~ X^{K} \in T^\star M\right\}
\end{align}
with the cotangent bundle over $T^\star M$ given by the vector space dual to $\mathcal{V}$, defined as
\begin{align}
\label{subsec:prelim:phaseDualVectSpace}
\mathcal{V}^{\star} \equiv \left\{\tilde{\mathbf{W}} = W_{K}~dz^{K} ~\vert~ W_{K} \in T^\star M\right\}
\end{align}
such that the basis vectors for $\mathcal{V}$ and dual basis for $\mathcal{V}^{\star}$ satisfy equation~(\ref{subsec:prelim:vectBasis}).

For any function, $G$, of the phase space coordinates which is differentiable at least once, the symplectic form, equation~(\ref{subsec:prelim:poincareTwoForm}), defines a vector field, $\mathbf{V}_{G} \in \mathcal{V}$, dual to the $1-$form, $dG \in \mathcal{V}^{\star}$, which satisfies
\begin{align}
\label{subsec:prelim:vecField}
dG \equiv \frac{d G}{d q^a}dq^a + \frac{dG}{dp_a}dp_a = \frac{d G}{d z^L}dz^{L} \equiv \omega^{2}\left(\mathbf{V}_{G},\cdot\right)
\end{align}
The vector field, $\mathbf{V}_{G}$, defines a flow in phase space, parameterized by $\tau$, satisfying
\begin{align}
\label{subsec:prelim:flow}
\frac{d\mathbf{z}}{d\tau} = \mathbf{V}_{G}\left[\mathbf{z}\left(\tau\right)\right]
\end{align}
which defines the components of $\mathbf{V}_{G}$, given by
\begin{align}
\label{subsec:prelim:vecComponents}
\mathbf{V}_{G} \equiv V^{K}_{G}~\p_K
\end{align}
with $K \in \left\{1,\dots,2N\right\}$. This flow, parameterized by $\tau$, defines \define{integral curves} in the phase space along which $G$ remains constant satisfying
\begin{align}
\label{subsec:prelim:flow}
\frac{dG}{d\tau} = 0 = \frac{\p G}{\p \tau} + \frac{\p G}{\p z^K}\frac{d z^K}{d\tau}
\end{align}
Here, $\tau$ parameterizes the integral curve along the vector flow generated by $G$; only when the conserved quantity $G$ is taken to define the Hamiltonian $H$ will the parameter $\tau$ be interpreted as time. Since the only restriction placed on $G$ is that it be differentiable at least once, equation~(\ref{subsec:prelim:flow}) shows that every differentiable phase space function will have an associated flow in phase space. Consider now two differentiable functions, $G$, and $F$, of the phase space variables. Associated with $G$ and $F$ are the respective vector fields $\mathbf{V}_{G}$ and $\mathbf{V}_{F}$ generating flows parameterized by $\tau_{G}$ and $\tau_{F}$. Since $G$ and $F$ are functions only of the phase space coordinates, $\frac{\p G}{\p \tau} = \frac{\p F}{\p \tau} = 0$ for all $\tau$. The \define{Poisson bracket} of $G$ and $F$ is defined to be  
\begin{align}
\label{subsec:prelim:pb}
\pb{G}{F}  & \equiv dG\left(\mathbf{V}_{F}\right) - dF\left(\mathbf{V}_{G}\right) = \omega^{2}\left(\mathbf{V}_{G},\mathbf{V}_{F}\right)
\end{align} 
and is often denoted
\begin{align}
\label{subsec:prelim:poissonBracketStructure}
\pb{G}{F} = J^{LK}~\p_{L}\left(G\right)~\p_{K}\left(F\right) 
\end{align}
with $J^{LK}$ defining the \define{cosymplectic form}. In canonical coordinates, the cosymplectic form is given by
\begin{equation}
\label{subsec:prelim:cosymplecticFormCanonical}
J^{LK} = \frac{1}{2}\left(J^{LK} - J^{KL}\right) = 
\left[\begin{array}{cc}
\mathbf{0} & \mathbb{I} \\
\\
-\mathbb{I} & \mathbf{0} \\
\end{array}\right]
\end{equation}
and is the inverse of the canonical form defined by equation~(\ref{subsec:prelim:canonicalForm}). The Poisson bracket, $\pb{G}{F}$, calculates the difference of a given phase space function, $F$, along the flow generated by $G$. In general, given two phase space functions, $G$ and $F$, the Poisson bracket will generate a third phase space function, $\pb{G}{F} = C$. The resulting phase space function, $C$, is referred to as the \define{commutation relation}. When the phase space function, $F$, is constant along the flow generated by $G$, $\pb{G}{F} = 0$, the functions $F$ and $G$ \define{commute}. The Hamiltonian, $H$, generates a \define{Hamiltonian vector field}, with an associated flow which is parameterized by coordinate time, $t$. Using the Poisson bracket, the evolution equations for the phase space coordinates, equations~(\ref{subsec:prelim:momentaEOM})~and~(\ref{subsec:prelim:positionEOM}), become
\begin{align}
\label{subsec:prelim:poissonEOM}
\frac{d \mathbf{z}}{d t}  = \frac{\p \mathbf{z}}{\p t} +  \pb{\mathbf{z}}{H} = \pb{\mathbf{z}}{H}
\end{align}
In general, for some function $F$ of the phase space coordinates, $\mathbf{z}$, which may have a dependence on the coordinate time, $t$, the total time derivative of $F$ will take the form
\begin{align}
\label{subsec:prelim:totalDerivF}
\frac{d F}{d t}  = \frac{\p F}{\p t} + \pb{F}{H}
\end{align}
Any function $F$ which is constant as the system evolves must satisfy, $\frac{d F}{dt} = 0$. In the case where the phase space function $F$ and the Hamiltonian, $H$, are time-independent, any $F$ which commutes with the Hamiltonian, $\pb{F}{H} = 0$, will remain constant as the system evolves. The vector field generated by any phase space function, $F$, which commutes with the Hamiltonian, $H$, will also be known as a Hamiltonian vector field, and will commute with the Hamiltonian vector field generated by $H$.

For canonical phase space coordinates, $\left(\mathbf{q},\tilde{\mathbf{p}}\right)$, the commutation relations amongst the phase space coordinates are
\begin{align}
\label{subsec:prelim:canonicalCommutationRelations}
\pb{q^a}{p_b} = \delta^a_b
\end{align}
All other commutation relations amongst the phase space coordinates vanish. Consider the case in which the phase space coordinates include $\tau$, the parameterization of the flow associated to the phase space function $G$.  The Poisson bracket $\pb{\tau}{G}$ yields
\begin{align}
\label{subsec:prelim:pbExp}
\pb{\tau}{G} = \frac{d}{d\tau}\tau = 1
\end{align} 
showing that $G$ is the canonical momenta conjugate to $\tau$. In general, when dealing with either canonical or non-canonical phase space coordinates, the cosymplectic form, $J^{KL}$, is defined by the commutation relations amongst the phase space coordinates, $\mathbf{z}$, 
\begin{align}
\label{subsec:prelim:cosymplecticForm}
J^{LK} \equiv \pb{z^L}{z^K}
\end{align}
In non-canonical coordinates, the elements of the cosymplectic form, $J^{LK}$, can be functions of the phase space coordinates, $J^{LK}\left(\mathbf{z}\right)$. 
Whenever the cosymplectic form, $J^{LK}$, is invertible, the symplectic form, $J_{IK}$, can be defined as the inverse of $J^{LK}$ so that
\begin{align}
\label{subsec:prelim:generalSymplecticForm}
J_{IL}J^{LK} = \delta_{I}~^{K}
\end{align}
When a distinction is necessary, the canonical cosymplectic form will be denoted $J^{LK}_{C}$. The equations of motion for the phase space coordinates take the compact form
\begin{align}
\label{subsec:prelim:hamCondensedNotation}
\dot{z}^{L} = J^{LK}\frac{\p H}{\p z^K}
\end{align}
for both canonical or non-canonical phase space coordinates. In any phase space coordinates, given two phase space functions, $G$ and $F$, the symplectic form, $\omega^2$, must map vectors over phase space to the dual space, equation~(\ref{subsec:prelim:vecField}), and must be closed, $d\omega^{2} = 0$, equation~(\ref{subsec:prelim:poincareTwoFormNonCompact}). Using the definition of the Poisson bracket in terms of the symplectic form, equation~(\ref{subsec:prelim:pb}), and insisting that partial derivatives commute, so that $dd = 0$, the Poisson bracket must satisfy the \define{Jacobi identity}
\begin{align}
\label{subsec:prelim:jacobiIdentities}
\pb{A}{\pb{B}{C}} + \pb{C}{\pb{A}{B}} + \pb{B}{\pb{C}{A}} = 0
\end{align}
for any phase space functions $A$, $B$, and $C$. 

The dynamics generated by the Lagrangian, $L$, and Hamiltonian, $H$, will yield a unique extremal for the action, but the Lagrangian and Hamiltonian are not unique themselves. Consider the addition of a total derivative, $\frac{dF}{dt}$, to the Lagrangian. The modified Lagrangian, $L^\prime \equiv L + \frac{dF}{dt}$, can change the value of the action, $S\left[\mathbf{q}\right]$, but will not change the extremal path as long as the total derivative which is added does not violate the boundary conditions of equation~(\ref{subsec:prelim:actionBoundary}). Similarly, in the Hamiltonian formulation, the addition of a total derivative, $-\frac{dF}{dt}$, to the Hamiltonian will change the canonical action, $S\left[\mathbf{q},\tilde{\mathbf{p}}\right]$, by a boundary term, $\Delta F \equiv F(t_1) - F(t_0)$, but will leave the extremal path in phase space invariant as long as equation~(\ref{subsec:prelim:actionBoundary}), expressed in phase space coordinates, remains satisfied. Using the Legendre transformation, equation~(\ref{subsec:prelim:momenta}), and the boundary term generated by the variation yielding the extremal path, equation~(\ref{subsec:prelim:actionBoundary}), any total derivative added to the Hamiltonian, $-\frac{dF}{dt}$, leaving the equations of motion invariant can be written as a canonical transformation. The function $F$ is called the \define{generating function} of the canonical transformation and, including the canonical pair $\left(t,H\right)$ as phase space coordinates, satisfies
\begin{align}
\label{subsec:prelim:canonicalGeneratingFunction}
\tilde{\Lambda}\left(\bar{\mathbf{z}}\right) = \tilde{\Lambda}\left(\mathbf{z}\right) - dF\left(\mathbf{z}\right)
\end{align}
with the new canonical coordinates, $\bar{\mathbf{z}}$, defined as functions of the initial canonical coordinates, $\mathbf{z}$. Since the Poincar\'e $1-$form $\tilde{\Lambda}\left(\bar{\mathbf{z}}\right)$ differs from the original Poincar\'e $1-$form $\tilde{\Lambda}\left(\bar{\mathbf{z}}\right)$ by an exact derivative, $dF\left(\mathbf{z})\right)$, the Poincar\'e $2-$form remains unchanged
\begin{align}
\label{subsec:prelim:generatingCanonicalForm}
\omega^{2} \equiv d\tilde{\Lambda}\left(\bar{\mathbf{z}}\right) = d\tilde{\Lambda}\left(\bar{\mathbf{z}}\right) + ddF\left(\mathbf{z}\right) \equiv d\tilde{\Lambda}\left(\bar{\mathbf{z}}\right)
\end{align}
Since the canonical form, $\omega^2$, is preserved, by definition, the transformation from $\mathbf{z}$ to $\bar{\mathbf{z}}$ is canonical, whence $F$ generates a canonical transformation. The new canonical coordinates, $\bar{\mathbf{z}}$, are defined as functions of the initial canonical coordinates, $\mathbf{z}$, by
\begin{align}
\label{subsec:prelim:canonicalGeneratingFunctionCoordinate}
\pb{z^{I}}{F} = \frac{1}{2}\left[z^{I} - \bar{z}^{M}J_{ML}~\left(\frac{d\bar{z}^{L}}{dz^{K}}\right) J^{IK}\right]
\end{align}
so that $F$ must take the form
\begin{align}
\label{subsec:prelim:canonicalGeneratingFunctionForm}
F &= \frac{1}{2}\int~\left\{z^{I}J_{IK}~dz^{K} - \bar{z}^{M}J_{ML}~d\bar{z}^{L} \right\}
\end{align}
showing that the coordinate transformation generated by $F$ must be an invertible transformation between the canonical phase space coordinates $\mathbf{z}$ and $\bar{\mathbf{z}}$. Although the transformation generated by $F$ must be invertible, equation~(\ref{subsec:prelim:canonicalGeneratingFunctionCoordinate}) shows that the generating function $F$ is only uniquely defined up to the addition of a constant multiple of any phase space function $C$ satisfying $dC = 0$, since the Poincar\'e $1-$form, defining the phase space coordinates, will only be altered by a term $d\left(F + C\right) = dF$, showing that $F+C$ and $F$ yield the same canonical transformation.  

Using the canonical Hamiltonian, $H$, to generate canonical transformations yields
\begin{align}
\label{subsec:prelim:canonicalGeneratingHamiltonian}
dH = \p_{K}H dz^{K} = \dot{z}^{L}\left(\mathbf{z}\right)J_{LK}~dz^{K}
\end{align}
which is exact. As a result, when the canonical Hamiltonian is independent of time, $\p_{t}H = 0$, adding any constant multiple of $dH$ to the Hamiltonian $1-$form, equation~(\ref{subsec:prelim:hamiltonianOneForm}), will leave the canonical form invariant. Additionally, canonical transformations of this form will also leave the extremal path invariant since the canonical Hamiltonian $H$ itself will remain unchanged. Interestingly, using equation~(\ref{subsec:prelim:canonicalGeneratingHamiltonian}), the evolution in phase space, equation~(\ref{subsec:prelim:hamCondensedNotation}), can be interpreted as a continuous \define{infinitesimal canonical transformation} generated by the canonical Hamiltonian, $H$, multiplied by the constant infinitesimal $dt$. In general, any time-independent phase space function, $G_{C}$, which commutes with the canonical Hamiltonian, $H$, for all time $t$ defines a \define{constant of motion} for the physical system. Since the constants of motion, $G_{C}$, always commute with the canonical Hamiltonian, $H$, the physical content of the theory will remain invariant under continuous infinitesimal canonical transforms generated each $G_{C}$. For example, time-independent Hamiltonians satisfy, $\pb{H}{H} = 0$, and so $H$ will be a constant of motion with the value of the Hamiltonian, $H$, corresponding to the \define{total energy} of the system. In phase space coordinates, $\mathbf{z}$, the infinitesimal transformations generated by the constant of motion $G_{C}$ and infinitesimal constants, $\epsilon$, will be 
\begin{align}
\label{subsec:prelim:ICT}
\bar{\delta}_{C}z^{L} \equiv \epsilon\pb{z^{L}}{G_{C}}
\end{align}
Using the time-independent canonical Hamiltonian, $H$, as an example, $H$ generates the familiar infinitesimal canonical transformation
\begin{align}
\label{subsec:prelim:ICTHam}
\bar{\delta}_{H}z^{L} \equiv \dot{z}^{L} dt =  dt~\pb{z^{L}}{H}
\end{align}
In addition to the constants of motion, the one dimensional groups of canonical transformations generated by the constants of motion, $G_{C}$, will also be invariant under all canonical transformations, and so correspond to physical values which are called \define{global symmetries} of the physical system. Returning to the example of systems with a time-independent canonical Hamiltonian, $H$ yields the total energy, $E$, and generates the group associated with a global symmetry under constant-time translations corresponding to \define{conservation of energy}. This is the Hamiltonian form of \define{N\"other's first theorem}, which states that a general differential system will have one conserved quantity corresponding to each continuous symmetry, with a continuous symmetry of a differential system defined by a continuous group of transformations mapping the space of solutions to the differential system into itself \cite{noether1}. 
%
\subsection{Singular Legendre Transformations}
\label{subsec:gauge}
Often physical systems will be described by a Lagrangian, $L$, which generates a Legendre transform, defined by the map from $TM\rightarrow T^\star M$ in equation (\ref{subsec:prelim:momenta}), which is not a bijection. In this case, the map from the tangent to cotangent bundle will be a \define{singular Legendre transform}, generated by a \define{singular Lagrangian}.  When the Legendre transform is singular, the square symmetric $N\times N$ matrix 
\begin{align}
\label{subsec:gauge:noninvert}
\mathbf{T} \equiv T_{ab} \equiv \frac{\delta}{\delta \dot{q}^a}\left(\frac{\delta L}{\delta \dot{q}^b}\right) \equiv \frac{\delta p_a}{\delta \dot{q}^b}
\end{align}
will not have an inverse. As a result, some of the velocities, $\dot{\mathbf{q}}$, will not be expressible as functions of the phase space coordinates, $\left(\mathbf{q},\tilde{\mathbf{p}}\right)$. The \emph{rank} of $\mathbf{T}$ is given by the dimension of the maximal square symmetric submatrix of $\mathbf{T}$ which is invertible, and is equal to the number of linearly independent columns of $\mathbf{T}$. The rank of $\mathbf{T}$, given by the integer $M$ with $M < N$, is assumed to be constant throughout phase space allowing $M$ of the momenta to be inverted in terms of $M$ velocities. The remaining $N - M$ momenta, which are not invertible, will take the form
\begin{align}
\label{subsec:gauge:const}
p_c\left(\mathbf{q},\tilde{\mathbf{p}}\right) = \phi_c\left(\mathbf{q},\tilde{\mathbf{p}}\left(\mathbf{q},\dot{\mathbf{q}}\right)\right)
\end{align}
for phase space functions $\phi_c\left(\mathbf{q},\tilde{\mathbf{p}}\left(\mathbf{q},\dot{\mathbf{q}}\right)\right)$ which are independent of the $N-M$ non-invertible velocities $\dot{q}^{c}$. If the functions $\phi_c$ were to depend on the non-invertible velocities, $\dot{q}^{c}$, then equation~(\ref{subsec:gauge:const}) would yield an invertible relation, in contradiction with the assumption that the rank of $\mathbf{T}$ is $M$. Under an appropriate change of coordinates on the tangent bundle, the matrix $\mathbf{T}$ can be brought into block form with the maximal invertible subblock given by the $M\times M$ square symmetric matrix $\mathbf{O}$. In these coordinates, the Lagrangian will take the form
\begin{align}
\label{subsec:gauge:lagrangian}
L = \dot{q}^{a}O_{ab}\dot{q}^{b} + A_{a}\dot{q}^{a} + \phi_{c}\dot{q}^{c} + B\left(\mathbf{q}\right)
\end{align}
where $a,b \in \left\{1,\dots, M\right\}$, $c \in \left\{1,\dots, (N-M)\right\}$, and $B\left(\mathbf{q}\right)$ is independent of any velocities, $\dot{\mathbf{q}}$. The terms in equation~(\ref{subsec:gauge:lagrangian}) which are linear or independent of the velocities, $A_{a}\dot{q}^{a}$ and $B\left(\mathbf{q}\right)$ respectively, will not affect the rank of $\mathbf{T}$, and therefore will not affect the maximal invertible subblock, $\mathbf{O}$. From equation~(\ref{subsec:gauge:lagrangian}), the velocities $\dot{q}^{c}$ will appear at most linearly in the Lagrangian, suggesting that there is a transformation to coordinates in which the $N-M$ phase space functions, $\phi_c$, vanish. In these coordinates $\mathbf{T}$ will take the form
\begin{equation}
\label{subsec:gauge:noninvertBlock}
T_{ab} = \frac{\delta}{\delta \dot{q}^b}\left(\frac{\delta L}{\delta \dot{q}^a}\right) = \left[\begin{array}{cc}
\mathbf{0} & \mathbf{0} \\
\\
\mathbf{0} & O_{cd} \\
\end{array}\right]
\end{equation}
with $\mathbf{O}$ the $M\times M$ maximal invertible subblock of $\mathbf{T}$. The form of equation~(\ref{subsec:gauge:lagrangian}) suggest that such a coordinate transformation can be accomplished by adding a total derivative, $\frac{dF}{dt}$, to the action which satisfies
\begin{align}
\label{subsec:gauge:totalDerivative}
\frac{d F}{d t} & = \frac{\p F}{\p t} + \frac{\p F}{\p q^{a}}\dot{q}^{a} + \frac{\p F}{\dot{q}^{b}}\ddot{q}^{b}   + \dots = -\phi_c\dot{q}^{c}
\end{align} 
for some function $F$ of the tangent bundle coordinates, $\left(\mathbf{q},\dot{\mathbf{q}}\right)$. In phase space coordinates  
\begin{align}
\label{subsec:gauge:generatingFunction}
\frac{d F}{d t} & = \frac{\p F}{\p t} + \frac{\p F}{\p q^{a}}\dot{q}^{a} + \frac{\p F}{\p p_{b}}\dot{p}_{b} = -\phi_c\dot{q}^{c}
\end{align} 
The addition of a total derivative satisfying equation~(\ref{subsec:gauge:totalDerivative}) will yield a Lagrangian, $L^{\prime} \equiv L + \frac{d F}{dt}$, which generates $N-M$ canonical momenta of the form 
\begin{align}
\label{subsec:gauge:vanishingPFCC}
p_c = 0
\end{align}
Since these expressions for the momenta, $p_c = 0$, have been derived using a specific coordinate system, it is not possible to drop the $N-M$ momenta from the phase space without restricting the permissible canonical transformations, and consequently fixing the value of the boundary terms present in the action. In particular, a given solution in phase space was shown to evolve under a continuous set of canonical transformations which are generated by the canonical Hamiltonian, $H$. Consequently, the form of the non-invertible momenta is not guaranteed to be invariant as the system evolves.
%
\subsection{Constraints in the Lagrangian Formulation}
\label{subsec:constraints}
\label{subsubsec:lagrangian_constraints}
Equations expressing relations amongst the solution space coordinates which must be preserved  by the dynamics are \define{constraints}. In the Lagrangian formulation, constraints are introduced through \define{equations of constraint}, taking the form $f\left(\mathbf{q},\dot{\mathbf{q}}\right) = 0$, and are imposed by modifying the Lagrangian to include multiples of the equations of constraint. These multiplying factors are known as \define{Lagrange multipliers} and take values such that the equations of constraint hold. A constrained Lagrangian, $L_C$, with Lagrange multipliers, $\lambda^A$, and constraints, $f_{A}\left(\mathbf{q},\dot{\mathbf{q}}\right) = 0$ takes the form
\begin{align}
\label{subsec:constraints:constrainedLagrangian}
L_{C} \equiv L + \lambda^{A}f_{A}
\end{align}
with $L$ denoting the unconstrained Lagrangian. Constraints which uniquely determine the Lagrange multiplers are \define{holonomic}. Constraints which are not holonomic are \define{non-holonomic} and do not uniquely determine the Lagrange multipliers. The equations of constraint establish relations amongst the $N$ configuration space coordinates, $\mathbf{q}$, and the velocities, $\dot{\mathbf{q}}$, reducing the dimension of the space of solutions. For holonomic constraints, all Lagrange multipliers are uniquely determined, whence the equations of motion can be inverted to yield the equations of constraint. Solving both the equations of motion and equations of constraint simultaneously, the dimension of the configuration manifold, $M$, can be reduced by one for each constraint present, reducing the tangent bundle, $TM$, by two dimensions. Non-holonomic constraints are not able to reduce the space of solutions since the undetermined Lagrange multipliers present do not restrict solutions to the equations of motion to a submanifold of the tangent bundle which is itself a tangent bundle to some reduced configuration space. In the Hamiltonian formulation, for each holonomic constraint present, one degree of freedom is removed from the phase space. When moving to the Hamiltonian formulation from the Lagrangian formulation when non-holonomic constraints are present, the Lagrange multipliers are not uniquely determined by the equations of motion and must be accounted for in the phase space coordinates, therefore non-holonomic constraints do not allow the phase space to be reduced.

In the Hamiltonian formulation, the canonical Hamiltonian, $H$, derived from the constrained Lagrangian, equation~(\ref{subsec:constraints:constrainedLagrangian}), will generate dynamics, consistent with solutions to the Euler-Lagrange equations, which preserve the constraints $f_{A}\left(\mathbf{q},\dot{\mathbf{q}}\left(\mathbf{q},\tilde{\mathbf{p}}\right)\right) = 0$. If the canonical Hamiltonian, $H$, is derived from a singular Lagrangian, the $N-M$ expressions of equation~(\ref{subsec:gauge:const}) can be expressed as the $N-M$ constraints
\begin{align}
\label{subsec:constraints:primaryConstraints}
p_c - \phi_c\left(\mathbf{q},\tilde{\mathbf{p}}\left(\mathbf{q},\dot{\mathbf{q}}\right)\right) = 0
\end{align}
with $\phi_c$ being a function of the invertible phase space coordinates. Constraints imposed on the phase space coordinates resulting from a Lagrangian formulation generating a singular Legendre transform are known as \define{primary constraints}. Using the definition of the canonical momenta, equation~(\ref{subsec:prelim:momenta}), the constraints of equation~(\ref{subsec:constraints:primaryConstraints}) must be generated by a Lagrangian which is at most linear in the non-invertible velocities, $\dot{q}^{c}$. Since the Lagrangian cannot involve terms which are quadratic in the non-invertible velocities, $\dot{q}^{c}$, the resulting Euler-Lagrange equations generated by extremizing the action, $S$, cannot completely determine the dynamics for $\dot{q}^c$. Transforming to the coordinates derived in subsection~(\ref{subsec:gauge}) in which all $N-M$ non-invertible momenta vanish, equation~(\ref{subsec:gauge:vanishingPFCC}), the Lagrangian, $L$, will be independent of the non-invertible velocities, $\dot{q}^{c}$, yielding the constraints $p_c = 0$. In these coordinates, the $N-M$ configuration space variables, $q^c$, will have velocities which do not appear in the Lagrangian, and so must generate $N-M$ Euler-Lagrange equations of the form
\begin{align}
\label{subsec:constraints:elEquations}
\chi_c\left(\mathbf{q},\dot{\mathbf{q}}\right) \equiv \frac{\p L}{\p q^{c}} - \frac{d}{dt}\frac{\p L}{\p \dot{q}^{c}} = \frac{\p L}{\p q^{c}} = \frac{\delta L}{\delta q^{c}} = 0
\end{align}
revealing that the Lagrangian will be independent of the variables $q^{c}$ as well as the velocities $\dot{q}^{c}$. Since the Lagrangian is independent of the variables $q^c$ and velocities $\dot{q}^c$, the dynamics leaves these values undetermined. The dynamics generated must satisfy equation~(\ref{subsec:constraints:elEquations}), thereby defining $N-M$ more constraints, in addition to the $N-M$ primary constraints, which are inherent in the system. Continuing to work in the coordinate system in which the Lagrangian, $L$, is independent of the coordinates $\left(q^{c},\dot{q}^{c}\right)$, define a new Lagrangian, $L^\prime$, as the value of the Lagrangian $L$ evaluated with all undetermined terms set equal to zero, $q^c = \dot{q}^c = 0$. The Lagrangian $L$ can then be expressed as 
\begin{align}
\label{subsec:constraints:invertLagrangian}
L\left[\mathbf{q},\dot{\mathbf{q}}\right] = L^\prime\left[\mathbf{q},\dot{\mathbf{q}}\right] + q^{c}\chi_c\left[\mathbf{q},\dot{\mathbf{q}}\right] + \dot{q}^{c}p_c\left[\mathbf{q},\dot{\mathbf{q}}\right] 
\end{align}
modulo terms which do not affect the Legendre transform or the dynamics. The form of $L$ in equation~(\ref{subsec:constraints:invertLagrangian}) shows that the $2(N-M)$ coordinates $\left(q^{c},\dot{q}^{c}\right)$ take the same form as undetermined Lagrange multipliers for the $2(N-M)$ constraints. Although these results were derived in a coordinate system in which the Lagrangian is independent of the configuration space variables $q^c$ and velocities $\dot{q}^{c}$, no canonical transformation can remove the $2(N-M)$ undetermined functions present in the formulation. Consequently, all primary constraints in the Hamiltonian formulation will be the direct result of non-holonomic constraints present in the Lagrangian formulation, precisely because these primary constraints arise from the presence of velocities, $\dot{q}^{c}$, which cannot be determined by the canonical dynamics.

\subsection{Constraints in the Hamiltonian Formalism}
\label{subsubsec:hamiltonian_constraints}
In order for the primary constraints to be satisfied under the evolution generated by the canonical Hamiltonian, all of the constraints must commute with $H$, since only then will they continue to vanish as the system evolves. Furthermore, as shown in subsection~(\ref{subsec:gauge}), the form of the primary constraints can change as the result of a canonical transformation and so should be designated distinctly from statements which retain their form under all acceptable canonical transformations. In order to avoid confusion with statements which remain true throughout phase space, the symbol $\approx$ is used for \define{weak equalities}, which are equations that will be true only when all constraints are satisfied. Weak equalities are not valid throughout phase space since they will be true only while the evolution satisfies the constraints, whence they cannot be used to reduce the dimension of the phase space directly. The primary constraints are only weakly equal to zero and so will be expressed as
\begin{align}
\label{subsec:constraints:weaklyVanishing}
p_c \approx 0
\end{align}
The requirement that the primary constraints commute with the canonical Hamiltonian leads to \define{consistency constraints} which take the form
\begin{align}
\label{subsec:constraints:secondary}
\pb{p_{c}}{H}= \chi_c\left(\mathbf{q},\tilde{\mathbf{p}}\right) \approx 0
\end{align}
and must weakly vanish in order for the primary constraints to be satisfied as the system evolves. The constraints, $\chi_c \approx 0$, generated by the primary constraints are often referred to as \define{secondary constraints}.  Constraints which are weakly equal to zero are said to be \define{weakly vanishing}, and two phase space functions which generate a weakly vanishing commutation relation \define{weakly commute}. A constraint or commutation relation which is identically zero is called \define{strongly vanishing}. Strong equalities are valid throughout phase space, whether or not weak equalities are satisfied, and will be denoted with the standard equal sign.

The process of finding consistency constraints must be continued until the set of all consistency constraints, along with the primary constraints, vanish. That is, if $\chi_c ~\displaystyle{\not}{\approx}~ 0$ after all weak equalities are evaluated then $\chi_c$ must generate a further constraint on the system
\begin{align}
\label{subsec:constraints:consistency}
\chi^{\prime}_c \equiv \pb{\chi_c}{H} = \dot{\chi}_c\left(\mathbf{q},\tilde{\mathbf{p}}\right) \approx 0
\end{align}
When a complete set of constraints is found such that all constraints generated by the $N-M$ primary constraints, $p_c \approx 0$, along the flow in phase space generated by the canonical Hamiltonian weakly vanish, no further constraints are present. The complete set of all constraints, primary and all consistency constraints, imposed on the system will be denoted
\begin{align}
\label{subsec:constraints:FCC}
\mathcal{C}_{A} = \left\{p_c,\chi_a,\chi^\prime_b,\dots\right\}
\end{align}
with the label $A$ running over all constraints. The submanifold of phase space on which all constraints vanish defines the \define{constraint manifold}. If all constraints are necessary to define the constraint manifold uniquely, then the set of constraints is \define{irreducible}, otherwise the set of all constraints will be \define{reducible}. Reducible sets of constraints will not all be independent, allowing some constraints to be written as vanishing functions of the remaining constraints. For all examples considered here, the set of all constraints will be irreducible and will contain each of the primary first class constraints, $p_c \approx 0$, which will generate a single consistency constraint, $\chi_c \approx 0$, that weakly commutes with the canonical Hamiltonian. These physical systems will have a total of $2(N-M)$ weakly vanishing constraints defining a constraint manifold with $2N - 2(N-M) = 2M$ dimensions. 

The distinction between primary and consistency constraints, as pointed out by Dirac \cite{dirac2}, is relatively unimportant compared to the distinction made between constraints which have a weakly vanishing commutation relation with all other constraints and those which have a non-vanishing commutation relation with at least one other constraint. Constraints which commute with all other constraints are \define{first class}, while those which have a non-vanishing commutation relation with at least one other constraint are \define{second class}. For all examples considered here, all primary and consistency constraints generated by the Lagrangian will be first class. The first class constraints, $\mathcal{C}_{A}$, of a theory will be \define{closed} under the Poisson bracket, satisfying 
\begin{align}
\label{subsec:constraints:FCCalgebra}
\pb{\mathcal{C}_{A}}{\mathcal{C}_{B}} = \Gamma^{C}_{AB}\mathcal{C}_{C} \approx 0
\end{align}
with $\Gamma^{C}_{AB}$ defining the structure coefficients. The commutation relations amongst the first class constraints is known as the \define{first class constraint algebra}, often shortened to just \define{constraint algebra} when no second class constraints are present. When the constraint algebra is defined by structure coefficients which are constant matrixes, the $\Gamma^{C}_{AB}$ are known as the \define{structure constants}. First class constraints which generate structure coefficients that are not constant matrixes but rather functions of the phase space variables are sometimes also referred to as \define{business class constraints}. All properties derived here for first class constraints will also to apply to business class constraints, so no distinction will be made. Any phase space function, $G$, satisfying 
\begin{align}
\label{subsec:constraints:FCfunction}
\pb{G}{\mathcal{C}_{A}} \approx 0
\end{align}
for all first class constraints is referred to as a \define{first class function}. In particular, the canonical Hamiltonian, $H$, used to derive the consistency constraints will be a first class function, satisfying equation~(\ref{subsec:constraints:FCfunction}), and is referred to as the \define{first class Hamiltonian}, $H_{FC}$. The first class Hamiltonian, $H_{FC}$, will always be assumed to be time-independent, satisfying $\pb{H^{\prime}_{FC}}{H_{FC}} \approx 0$ for any $H^{\prime}_{FC}$ which differs from $H_{FC}$ by any transformation generated by some combination of first class constraints. As a result, all first class constraints must also be time-independent in order to commute with the first class Hamiltonian as the system evolves. The first class Hamiltonian will then be associated with the total energy of the system, and symmetry under global time translations will correspond to conservation of energy. For diffeomorphic invariant field theories having locally vanishing first class Hamiltonian density, $\mathcal{H}_{0} \approx 0$, such as in G.R., the locally vanishing form of the canonical Hamiltonian means that the theory does not permit a canonical local definition of energy. In such theories, conserved quantities associated with global symmetries will be given by quantities evaluated over domain boundaries.

It is important to note that the first class constraint algebra, equation~(\ref{subsec:constraints:FCCalgebra}), is assured to close only on the constraint manifold, where all constraints vanish, so it does not make sense to talk about the first class constraint algebra elsewhere in phase space. It is also true that first class functions are only defined on the constraint manifold, and in general will have non-vanishing commutation relations with the first class constraints elsewhere in phase space. This includes the first class Hamiltonian, $H_{FC}$, which generates the dynamics. As a result, the phase space dynamics will only be meaningfully defined for systems which remain on the first class constraint manifold.
%
\subsection{Gauge Freedom and the Extended Hamiltonian}
\label{subsec:extendedHamiltonian}
Consider the first class Hamiltonian, $H_{FC}$, derived in coordinates in which the primary constraints take the form $p_c \approx 0$. The primary constraints, $p_c \approx 0$, cannot be present in the first class Hamiltonian, $H_{FC}$, since the theory does not provide canonical evolution equations for the configuration space variables $q^{c}$ which are conjugate to the vanishing momenta. Since the first class Hamiltonian, $H_{FC}$, is a first class function, any multiple of first class constraints can be added to the Hamiltonian without modifying the constraint manifold or constraint algebra. The addition of some combination of the first class constraints to the first class Hamiltonian corresponds to a change in the undetermined multipliers of the non-holonomic constraints in the Lagrangian formulation. Therefore, the physically meaningful content of the theory will remain unchanged whether the dynamics are generated by the first class Hamiltonian, $H_{FC}$, or a Hamiltonian defined by the addition of some combination of the first class constraints, $\mathcal{C}_{A}$, to the first class Hamiltonian. These observations led Dirac to introduce the \define{total Hamiltonian}
\begin{align}
\label{subsec:hamiltonian:TH}
H_T \equiv H_{FC} + \lambda^{c}p_c
\end{align}
with the coefficients of the $N-M$ primary constraints, given by the $N-M$ undetermined Lagrange multipliers $\lambda^c$, providing dynamical equations for the configuration space variables $q^c$. Although the total Hamiltonian, $H_T$, will provide evolution equations for all phase space coordinates, it is not the most general extension to the first class Hamiltonian, $H_{FC}$, since the variables $q^{c}$, which multiply the secondary constraints, $\chi_c \approx 0$, are no longer completely arbitrary, having their velocities specified by the Lagrange multipliers $\lambda^{c}$. The most general extension to the first class Hamiltonian, $H_{FC}$, must then include contributions from all first class constraints, $\mathcal{C}_{A}$, with undetermined Lagrange multipliers, $\lambda^{A}$, yielding the \define{extended Hamiltonian}
\begin{align}
\label{subsec:hamiltonian:EXH}
H_E \equiv H_{FC} + \lambda^A\mathcal{C}_A
\end{align} 
On the first class constraint manifold $H_{E} \approx H_{T} \approx H_{FC}$, so the first class Hamiltonian, total Hamiltonian, and extended Hamiltonian will all yield the same physical results. 

Since the value of first class functions will agree for the dynamics generated by either the first class Hamiltonian, $H_{FC}$, total Hamiltonian, $H_T$, or extended Hamiltonian, $H_{E}$, physically meaningful quantities must be first class functions so that the addition of terms involving the first class constraints will not affect their dynamics. These physically meaningful quantities are called \define{observables}, which are defined to be non-vanishing first class functions of the phase space variables. As an example, when the Lagrangian is not singular, the theory has no first class constraints and so all phase space coordinates represent physically meaningful content. Transformations of the phase space coordinates which leave the observables invariant define \define{gauge transformations} with the group of all gauge transformations defining the \define{gauge group}. The ability to perform gauge transformations amongst the canonical phase space variables is known as \define{gauge freedom}. In the extended Hamiltonian, the gauge freedom of the theory is embodied in the undetermined multipliers $\lambda^{A}$ which can be any function of the phase space coordinates, $\mathbf{z}$, and coordinate time, $t$. 

Consider an infinitesimal canonical transformation generated by some sum of first class constraints, $\mathcal{C}_{A}$, multiplied by infinitesimals, $\epsilon^{A}$, defining
\begin{align}
\label{subsec:hamiltonian:FCCICTGen}
G_{0} = \epsilon^{A}\mathcal{C}_{A} \approx 0
\end{align}
The generating function $G_{0}$ will satisfy 
\begin{align}
\label{subsec:hamiltonian:gaugeGenCommute}
\pb{G_{0}}{H_{FC}} \approx \pb{G_{0}}{H_{T}}\approx \pb{G_{0}}{H_{E}} \approx 0
\end{align}
for all values of $\epsilon^{A}$, including arbitrary functions of the phase space coordinates and coordinate time, $t$, and the infinitesimal variation vector in phase space generated by $G_{0}$ will have components 
\begin{align}
\label{subsec:hamiltonian:fccICTVar}
\hat{\delta} z^{L} \equiv \pb{z^{L}}{G_{0}} 
\end{align}
Because $G_{0}$ is a weakly vanishing function of the first class constraints it will weakly commute with all first class functions, whence the variation $\hat{\delta}\mathbf{z}$ must leave the constraint manifold and all observables invariant. Since this must be true for any value of $\epsilon^{A}$, the collection of all first class constraints, $\mathcal{C}_{A}\approx 0$, define the \define{generators of gauge transformations}. 

When working with field theories, the definition of gauge freedom will be slightly different, distinguishing global symmetry transformations from local gauge transformations while ensuring that gauge freedom be defined independently of any particular physical solution. This is accomplished by defining the space of permissible Lagrange multipliers such that all $\lambda^{A}$ have compact support contained within a single coordinate patch of the computational domain. This definition ensures that permissible Lagrange multipliers will yield local transformations, independent of any particular physical solution or any features of the global topology. Gauge freedom for field theories is identified by transformations generated by Lagrange multipliers having an arbitrary dependence on coordinates, $\left(t,\mathbf{x}\right)$, vanishing at the endpoints. Lagrange multipliers which do not have an arbitrary dependence on the computational coordinates identify a redundancy in the description of a physical solution. Using these definitions, gauge transformations correspond to transformations on phase space along an extremal path within the first class constraint manifold. Coordinate dependent Lagrange multipliers that do not vanish at the boundaries define transformations on phase space between extremal solutions that preserve gauge transformations, thus defining an evolution. This distinction is of little importance for theories where coordinate time, $t$, can be identified with a globally defined parameter, in which case the Hamiltonian does not vanish and generates a connected extremal path throughout. When it is not possible to define a global physical time independent of the dynamical fields canonically, the distinction between the roles of Lagrange multipliers becomes paramount, as will be the case for General Relativity, since the canonical Hamiltonian will (locally) vanish, as we will discuss in a subsequent paper. In any field theory, dynamics are defined by fixing all gauge freedom present while removing any redundancy in the description of physical solutions. As a result, this formulation defines an initial data problem with solutions uniquely specifying values for all phase space coordinates on an initial slice of constant coordinate time.
%
\subsection{Gauge Fixing and the Dirac Bracket}
\label{subsec:diracbracket}
When two phase space functions $G$ and $F$ have a non-vanishing commutation relation throughout a neighborhood of phase space, thereby satisfying $\pb{G}{F} \neq 0$ for all $\mathbf{z}$ in some neighborhood of $\mathbf{z}_{0}$ denoted by $\mathcal{U}_{0}$, then the commutation relation can be inverted to define a surface in phase space with coordinates on the surface defined by the value of the functions $G$ and $F$ in the neighborhood $\mathcal{U}_{0}$. For example, any canonical pair $\left(q^a,p_a\right)$ will generate the commutation relation $\pb{q^{a}}{p_{a}} = 1$, which is independent of the value of the phase space coordinates themselves and therefore valid throughout phase space, and, somewhat trivially then, the commutation relation can be inverted to define a surface in phase space with coordinates on the surface given by the values of $q^{a}$ and $p_{a}$. The ability of two phase space functions, $F$ and $G$, which generate an invertible commutation relation to act as the coordinates of a surface in phase space is related directly to the non-vanishing of their \emph{Lagrange bracket} defined as
\begin{align}
\label{subsec:diracbracket:lagrangeBracket}
\left\{G,F\right\} \equiv \frac{\p q^{n}}{\p G}\frac{\p p_{n}}{\p F} - \frac{\p q^{n}}{\p F}\frac{\p p_{n}}{\p G}
\end{align}
In order for $G$ and $F$ to act as surface coordinates, at least for some neighborhood $\mathcal{U}_{0}$, the Lagrange bracket must not vanish, $\left\{G,F\right\} \neq 0$ for all $\mathbf{z} \in \mathcal{U}_{0}$. Consider then two phase space functions, $G$ and $F$, which may or may not generate a nowhere vanishing commutation relation, along with two constraints, $C = 0$ and $A = 0$, defined throughout the neighborhood $\mathcal{U}_{0}$, which have a nowhere vanishing Lagrange bracket, $\left\{C,A\right\} \neq 0$. It is then possible to construct a bracket which, in the neighborhood $\mathcal{U}_{0}$, yields the value of $\pb{F}{G}$ restricted to the surface $C=A=0$ by removing the components of the phase space flow along both $F$ and $G$ which project onto the flows generated by $C$ and $A$ in phase space, thereby projecting onto the constrained solution where $C=A=0$. This bracket, denoted $\pb{\cdot}{\cdot}_D$, of any phase space function, $F$, with either constraint, $A$ or $C$, must satisfy
\begin{align}
\label{subsec:diracbracket:diracBracketConst}
\left[F,C\right]_D=\left[F,A\right]_D=\left[G,C\right]_D=\left[G,A\right]_D=0
\end{align}
for any $F,G$ in the neighborhood $\mathcal{U}_{0}$, and must also satisfy the Jacobi identity
\begin{align}
\label{subsec:diracbracket:jacobiIdentity}
\pb{E}{\pb{F}{G}}_{D} + \pb{G}{\pb{E}{F}}_{D} + \pb{F}{\pb{G}{E}}_{D} = 0
\end{align}
for any phase space functions $E$, $F$ and $G$. The generalization of the Poisson bracket which manifestly satisfies the constraints imposed on the Hamiltonian system, satisfying equations~(\ref{subsec:diracbracket:diracBracketConst})~and~(\ref{subsec:diracbracket:jacobiIdentity}), is known as the \emph{Dirac bracket}. For a collection of $2L$ constraints, $\vec{\mathcal{S}} = \left\{\mathcal{S}_{1},\dots,\mathcal{S}_{2L}\right\}$, which are surface forming in some neighborhood $\mathcal{U}_{0}$, the Dirac bracket of any two phase space functions $G$ and $F$ will be
\begin{align}
\label{subsec:diracbracket:diracBracketLB}
\pb{F}{G}_D\equiv \pb{F}{G} - \pb{F}{\mathcal{S}_{D}}\delta^{DA}\left\{\mathcal{S}_{A},\mathcal{S}_{B}\right\}\delta^{BE}\pb{\mathcal{S}_{E}}{G}
\end{align}
It should be clear that the number of constraints, $2L$, must be even in order for the collection of constraints to be surface forming, otherwise the resulting bracket will not be symplectic, and thus will not satisfy equation~(\ref{subsec:diracbracket:jacobiIdentity}). Furthermore, because the surface is defined throughout some neighborhood of phase space by the vanishing of the constraints, $\mathcal{S}_{A} = 0$, the constraints must be strongly vanishing since weakly vanishing constraints are defined only on the constraint manifold. The requirement that the Lagrange bracket of the constraints nowhere vanish is a requirement that the \emph{constraint commutation matrix} defined by
\begin{align}
\label{subsec:diracbracket:constMatrix}
D_{AB} = \pb{\mathcal{S}_{A}}{\mathcal{S}_{B}}
\end{align}
be invertible. The relation between the constraint commutation matrix and the Lagrange bracket of the constraints satisfies
\begin{align}
\label{subsec:diracbracket:constMatrixLagrangeBracket}
\sum_{B=1}^{2L}~\left\{\mathcal{S}_{A},\mathcal{S}_{B}\right\}D_{BC} = \delta_{AC}
\end{align}
When the constraint commutation matrix, equation~(\ref{subsec:diracbracket:constMatrix}), is invertible, the Dirac bracket, equation~(\ref{subsec:diracbracket:diracBracketLB}), for any arbitrary phase space functions $F$ and $G$, will be given by
\begin{align}
\label{subsec:diracbracket:diracBracket}
\pb{F}{G}_{D}~\equiv~\pb{F}{G}-\pb{F}{\mathcal{S}_A} D^{AB}\pb{\mathcal{S}_{B}}{G}
\end{align}
where $D^{AB}$ denotes the inverse to the constraint commutation matrix of equation~(\ref{subsec:diracbracket:constMatrix}). In particular, for any phase space function $F$, the Dirac bracket yields $\pb{F}{\mathcal{S}_{A}}_{D} = 0$ for any of the $2L$ constraints $\mathcal{S}_{A} = 0$, showing that $\pb{\cdot}{\cdot}_{D}$ satisfies equation~(\ref{subsec:diracbracket:diracBracketConst}). 

When dealing with gauge theories, the first class constraints will generate a vanishing constraint commutation matrix on the constraint manifold, because the first class constraint algebra is closed, and therefore cannot be used to construct a Dirac bracket. This will be true only on the first class constraint manifold, but the theory offers no natural way to define the commutation relations amongst the first class constraints off of the constraint manifold. Consider a minimal set of second class constraints, $\mathcal{S}_{A}$, imposed upon the system in order for the constraint commutation matrix generated by the set of all first class and second class constraints to be invertible. Because the first class constraints weakly commute amongst themselves, it will be necessary to impose a minimum of one independent second class constraint for every independent first class constraint present. Assuming a minimal set of second class constraints, $\mathcal{S}_{A}$, has been found the constraint commutation matrix of all second class and first class constraints can be inverted. For a gauge theory with $L$ first class constraints, denote the set of all constraints, second class and first class, as 
\begin{align}
\label{subsec:diracbracket:allConstraints}
\vec{\mathcal{D}} \equiv \left\{\mathcal{C}_{1},\dots,\mathcal{C}_{L},\mathcal{S}_{1},\dots,\mathcal{S}_{L}\right\}
\end{align}
with components, $\mathcal{D}_{A}$, having an index, $A$, which runs over all $2L$ constraints. Once a minimum set of second class constraints has been found, the constraint commutation matrix, $D_{AB} = \pb{\mathcal{D}_{A}}{\mathcal{D}_{B}}$, will be invertible and the resulting Dirac bracket will generate evolution equations for the original canonical phase space coordinates, $\mathbf{z}$, given by
\begin{align}
\label{subsec:diracbracket:diracBracketCanHam}
\pb{\mathbf{z}}{H_{FC}}_{D}~&\equiv~\pb{\mathbf{z}}{H_{FC}}-\pb{\mathbf{z}}{\mathcal{D}_A}D^{AB}\pb{\mathcal{D}_{B}}{H_{FC}}
\end{align}
The evolution equations for the canonical phase space coordinates, $\mathbf{z}$, will be identical to those generated by the Hamiltonian
\begin{align}
\label{subsec:diracbracket:conHam}
H_D \equiv H_{FC} + \Lambda^{A}\mathcal{D}_{A}
\end{align}
with Lagrange multipliers, $\Lambda^{A}$, given by
\begin{align}
\label{subsec:diracbracket:LagrangeMultipliers}
\Lambda^A \equiv D^{AB}\pb{\mathcal{D}_{B}}{H_{FC}}
\end{align}
This result shows that imposing a minimal set of second class constraints on the system, thereby allowing the constraint commutation matrix generated by the set of all first class and second class constraints to be inverted, uniquely fixes all of the undetermined multipliers present in the extended Hamiltonian, $H_E$. Once all Lagrange multipliers have been uniquely fixed, no gauge freedom will remain, as can be seen by considering any variation, $\hat{\delta}\mathbf{z}$, generated by any first class constraint using the Dirac bracket. Such variations will satisfy
\begin{align}
\label{subsec:diracbracket:diracBracketGaugeVariation}
\hat{\delta}\mathbf{z} = \pb{\mathbf{z}}{\epsilon^{A}\mathcal{C}_{A}}_{D} \equiv 0
\end{align}
because the Dirac bracket satisfies equation~(\ref{subsec:diracbracket:diracBracketConst}). The process of removing all gauge freedom is called \define{gauge fixing}, and equation~(\ref{subsec:diracbracket:diracBracketGaugeVariation}) shows that, with an appropriate choice of second class constraints, the Dirac bracket can be used to yield a gauge fixed system. 

The commutation relations amongst the original set of canonical variables, when restricted to the constraint manifold defined by $\mathcal{D}_{A} = 0$, will necessarily change since the phase space has been reduced. The new commutation relations restricted to the constraint manifold will be generated by the Dirac bracket, yielding
\begin{align}
\label{subsec:diracbracket:diracBracketCommutationRelations}
\pb{z^{L}}{z^{K}}_{D} = \pb{z^{L}}{z^{K}} + \pb{z^{L}}{\mathcal{D}_{A}}D^{AB}\pb{\mathcal{D}_{B}}{z^{K}}
\end{align}
These commutation relations yield the cosymplectic form, $J^{LK}\left(\mathbf{z}\right)$, as defined by equation~(\ref{subsec:prelim:cosymplecticForm}), of the phase space which is restricted to the constraint manifold. Since the Dirac bracket satisfies the Jacobi identities, the constraint manifold will be a symplectic manifold with the inverse of the cosymplectic form, $J^{LK}\left(\mathbf{z}\right)$, defining the symplectic form, $\omega^{2}$, restricted to the constraint manifold. Once the gauge has been fixed, the remaining freedom in the system will correspond precisely to the physical degrees of freedom. For example, when the first class constraint algebra is defined by $N-M$ independent primary constraints, $p_c \approx 0$, which each generate a single independent secondary constraint, $\chi_c \approx 0$, yielding a total of $2(N-M)$ first class constraints, it will be necessary to impose $2(N-M)$ independent second class constraints, yielding $4(N-M)$ total constraints, in order for the set of all constraints to yield an invertible constraint matrix.  Once a surface in phase space has been constructed from all $4(N-M)$, the reduced phase space on which all $4(N-M)$ constraints are satisfied will have dimensions  $2N - 4(N-M) = 4M - 2N \equiv 2D$. A theorem by Darboux, \cite{jorge1}, proves that all symplectic manifolds are locally equivalent, therefore the constraint manifold can be given a local coordinate system at any point which can be written as $D$ canonical pairs, and so the system is said to have $D$ degrees of freedom.

When working with a field theory rather than a discrete system, the requirement on the commutation relations amongst the constraints in order for the constraint commutation matrix to be invertible becomes
\begin{align}
\label{subsec:diracbracket:fieldInverse}
\begin{split}
\int d^3x^{\prime}~\left\{D^{AB}\left(x,x^\prime\right)\left[\mathcal{D}_{B}\left(x^{\prime\prime}\right),\mathcal{D}_{C}\left(x^{\prime}\right) \right]\right\} & \equiv  \int d^3x^{\prime}~\left\{D^{AB}\left(x,x^\prime\right)D_{BC}\left(x^{\prime},x^{\prime\prime}\right)\right\} \\
& = \delta^A_C~\delta\left(x,x^{\prime\prime}\right)
\end{split}
\end{align}
where $\delta\left(x,x^{\prime}\right)$ is the Dirac delta function. In general, the constraint commutation matrix for field theories, $D_{AB}\left(x,x^\prime\right)$, will involve differential operators so that the inverse will be an integral operator. In a field theory then, for an invertible constraint commutation matrix, $D_{AB}\left(x,x^\prime\right)$, the Dirac bracket between two arbitrary phase space functions, $F$ and $G$, will be
\begin{align}
\label{subsec:diracbracket:diracBracketContinuum}
\begin{split}
\left[F\left(x\right),G\left(x^\prime\right)\right]_{D}~&\equiv~\left[F\left(x\right),G\left(x^\prime\right)\right] \\
& -\int d^3x^{\prime\prime\prime}~\int d^3x^{\prime\prime}~
	\left(
		\left[
			F\left(x\right),
			\mathcal{D}_{A}\left(x^{\prime\prime}\right)
		\right] D^{AB}\left(x^{\prime\prime},x^{\prime\prime\prime}\right)
		\left[
			\mathcal{D}_{B}\left(x^{\prime\prime\prime}\right),
			G\left(x^\prime\right)
		\right] 
		\right)
\end{split}
\end{align}
Just as in the finite dimensional case, the undetermined multipliers of the Hamiltonian $H_D$, defined in equation~(\ref{subsec:diracbracket:conHam}), satisfy
\begin{align}
\label{subsec:diracbracket:LagrangeMultipliersField}
\Lambda^{A}\left(x\right) = \int d^3x^\prime~\left(
		D^{AB}\left(x,x^{\prime}\right)
		\left[
			\mathcal{D}_{B}\left(x^{\prime}\right),
			H_{0}
		\right]~\right)
\end{align}
showing that each $\Lambda^{A}\left(x\right)$ will have a coordinate dependence. 
%
\subsection{Synopsis of Gauge Systems in General}
\label{subsec:gaugeSystems:synopsis}
The kernel of the singular Legendre transform identifies, defines the gauge freedom. Consequently, any physically relevant part of a solution must reside in the preserved non-singular part of the Legendre transformation, providing a connection for the physical solutions in both the Hamiltonian and Lagrangian formulations. In the Lagrangian formulation, gauge freedom is identified with the transformations generated by the arbitrary coordinate dependence in the canonical position variables whose variation leads to the constraints of equation~(\ref{elEquations}). In the Hamiltonian formulation, which is the focus of this work, the presence of gauge freedom within the theory is manifest by transformations generated by the first class constraints, derived from the singular Legendre transform. Through the introduction of a judicious choice of second class constraints, the gauge can be uniquely fixed, removing all gauge freedom from the theory. Once the gauge freedom has been completely removed, the resulting Dirac bracket can be used to construct expressions for the physical observables and define dynamics for the gauge fixed Hamiltonian system. The resulting gauge fixed Hamiltonian system is defined entirely by the physical observables. This in fact was the reason motivating the examination of this framework for canonical quantization procedures and the starting point for the pioneering efforts of Dirac, Bergmann and many others to quantize gravity \cite{dirac2},\cite{bergmann1}, \cite{bergmann2} and \cite{dewitt1}. For a brief history on the development of constrained Hamiltonian dynamics, see \cite{salisbury1}. 

\section{Stability}
\label{sec:stability}
Our discussion of stability in sections~(\ref{subsec:stability:integrability})~thru~(\ref{subsec:numerical:hyperbolicity}) is approached from the standpoint of differential systems. The remainder of this work does not depend directly on the material or techniques from these sections, which use methods that lie somewhat off the path maintained by the other sections of this work, so the reader may omit these sections upon first reading.

In addition to the framework introduced in section~(\ref{sec:GaugeHamiltonian}), it will be helpful to develop a framework for treating systems of differential equations more generally before embarking on an examination of stability in the context of Hamiltonian formulations of gauge theories. 
We present in section~(\ref{subsec:stability:integrability}) a review of useful material from the study of general differential systems, providing a brief introduction to Pfaffian systems, solution manifolds, and integrability. Before approaching general Hamiltonian formulations of gauge theories from the perspective of a differential system, we introduce in section~(\ref{subsec:stability:invariant}) a decomposition of the phase space tangent bundle, discussed in subsection~(\ref{subsec:lagHamDyn}), into gauge invariant vector spaces. We then express a general Hamiltonian formulation with gauge freedom in the form of a differential system in section~(\ref{subsec:stability:gaugeIntegrable}), and then show that any system in which the gauge freedom has not been completely fixed will fail to be integrable. As a corollary, we show that Hamiltonian formulations of gauge theories which are completely gauge fixed will be integrable. In section~(\ref{subsec:numerical:hyperbolicity}), we examine hyperbolicity and well-posedness for Hamiltonian formulations of gauge theories. The discussion of numerical stability concludes in section~(\ref{subsec:stability:removeErr}) with the introduction of a geometrically motivated method for removing numerical error from Hamiltonian formulations of gauge theories. To simplify the discussion and avoid the pedantry necessary for dealing with infinite dimensional systems, only finite dimensional theories are considered in the general discussion of stability in this section.
%
\subsection{Pfaffian Systems and Integrability}
\label{subsec:stability:integrability}
Consider now a general theory defined on some manifold, $\mathcal{M}$, with $N$ variables, $\left\{z^{I}\right\}$, and $M$ constraints, $\mathcal{C}^{A}\left(\mathbf{z}\right) = 0$. Although the $M$ constraints can be written as $\mathcal{C}^{A}\left(\mathbf{z}\right) = 0$, at some point $\mathbf{z} \in \mathcal{M}$, the system actually evolves along some path in the tangent bundle, $T\mathcal{M}$. As a result, it is natural to write the $M$ constraints as $M$ linearly independent $1-$forms, $\theta^{A}$, belonging to the cotangent bundle, $\theta^{A} \in T^{\star}\mathcal{M}$, and to consider the possible solutions for the constrained system to be the space of evolution vectors, tangent to $\mathcal{M}$, to be the vectors $\mathbf{V} \in T\mathcal{M}$ satisfying 
\begin{align}
\label{subsec:numerical:pfaffianOneForms}
\theta^{A}\left(\mathbf{V}\right) = 0
\end{align}
 for all $M$ $1-$forms $\theta^{A}$. The $M$ linearly independent $1-$forms $\theta^{A} \in T^{\star}\mathcal{M}$ are called \define{Pfaffians}, and the vector space defined by
\begin{align}
\label{subsec:numerical:distribution}
\Delta \equiv \left\{\mathbf{V} \in T\mathcal{M}~ \vert~  \theta^{A}\left(\mathbf{V}\right) = 0 ~\forall~ \theta^{A}\right\}
\end{align}
is called a \define{distribution} for smooth vector fields $\mathbf{V} \in T\mathcal{M}$.  The linear independence of the $M$ Pfaffians, $\theta^{A}$ means that in an open neighborhood of any $\mathbf{z} \in \mathcal{M}$, the $M$ Pfaffians must satisfy
\begin{align}
\label{subsec:numerical:pfaffiansIndependent}
\bigwedge_{A = 1}^{M} \theta^{A} \equiv \theta^{1}\wedge\dots\wedge\theta^{M} \neq 0
\end{align}
A theory with constraints defined by a collection of Pfaffians is called a \define{Pfaffian system}.  For a theory with $N$ independent coordinates, $\left\{z^{I}\right\}$, and $M$ constraints, $\mathcal{C}^{A} = 0$, the distribution $\Delta$ will have $N - M$ dimensions. An \define{integral manifold}, $\Sigma$, for a distribution $\Delta$ is defined as a submanifold of $\mathcal{M}$,  
\begin{align}
\label{subsec:numerical:integralManifold}
i: \Sigma \hookrightarrow \mathcal{M}
\end{align}
which is everywhere tangent to the distribution, allowing the integral manifold to be defined as
\begin{align}
\label{subsec:numerical:integralManifoldTangent}
\Sigma \equiv \left\{\mathbf{z}^{\prime}\left(\mathbf{z}\right) \in \mathcal{M}~\vert~\mathbf{V}\left(i\left(\mathbf{z}^\prime \right) \right) = V^{I}\p_{I}~\mathbf{z}^{\prime}\left(\mathbf{z}\right) = 0~ \forall ~\mathbf{V} \in \Delta\right\}
\end{align}
Since each Pfaffian is independent, an integral manifold $\Sigma$ can have most $N -M$ dimensions. Since $\Sigma$ is everywhere tangent to the distribution, $\Delta$, the pullback of each Pfaffian, $\theta^{A} \in T^{\star}\mathcal{M}$, must satisfy
\begin{align}
\label{subsec:numerical:vanishingPullback}
i^{\star}\left(\theta^{A}\right) = 0~\forall~\theta^{A}
\end{align}
If the Pfaffian $1-$forms, $\theta^{A}$, do not satisfy equation~(\ref{subsec:numerical:vanishingPullback}), then dual to every $i^{\star}\left(\theta^{A}\right) \neq 0  \in T^{\star}\Sigma$, would be a vector, $\mathbf{V} \in T\Sigma$, which would not belong to the distribution, $\Delta$, whence $\Sigma$ can only be an integral manifold if $i^{\star}\theta^{A} = 0$.

Consider the space of all $p$-forms over $\mathcal{M}$, written $\Omega^{p}\left(\mathcal{M}\right)$, the space of all exterior forms over $\mathcal{M}$, $\Omega^{\star}\left(\mathcal{M}\right) = \bigoplus_{k=0}^{N} \Omega^{k}\left(\mathcal{M}\right)$, and the map, 
\begin{align}
\label{subsec:numerical:exteriorDerivativeMap}
d: \Omega^{p-1} \hookrightarrow \Omega^{p}
\end{align}
satisfying $dd = d^2 = 0$, which defines the exterior derivative. Since the $M$ Pfaffians $\theta^{A} \in \Omega^{1}\left(\mathcal{M}\right)$ must satisfy equation~(\ref{subsec:numerical:vanishingPullback}), the wedge product of $p$ Pfaffian $1-$forms must form a basis in $\Omega^{p}\left(\mathcal{M}\right)$ for the space of all $p$-forms over $\mathcal{M}$ residing in the kernel of the pullback $i^{\star}: \Omega^{p}\left(\mathcal{M}\right) \rightarrow \Omega^{p}\left(\Sigma\right)$. In order for the space of all exterior forms in the kernel of the pullback $i^{\star}:\Omega^{\star}\left(\mathcal{M}\right) \rightarrow \Omega^{\star}\left(\Sigma\right)$ to be preserved under the map $d$, equation~(\ref{subsec:numerical:exteriorDerivativeMap}), the Pfaffian $1-$forms must satisfy 
\begin{align}
\label{subsec:numerical:closedPfaffians}
d\theta^{A} = -\omega^{A}~_{B}\wedge \theta^{B}
\end{align}
with connection $\omega^{A}~_{B}$ derived by the structure equations, since for any map $j:\mathcal{N} \rightarrow \mathcal{M}$, the exterior derivative, $d$, commutes with the pullback, $j^{\star}:\Omega^{\star}\left(\mathcal{M}\right) \rightarrow \Omega^{\star}\left(\mathcal{N}\right)$.  The property that the exterior derivative, $d$, commutes with the pullback of an differentiable map, $j$, yielding the relation
\begin{align}
\label{subsec:numerical:commutePullback}
j^{\star}\circ d = d \circ j^{\star}
\end{align}
is extremely useful and is a direct consequence of the exterior calculus;  it is true for any exterior derivative, $d$, and differentiable map $j$ \cite{frankel1},\cite{nakahara},\cite{jorge1},\cite{diffgeo1}. If $d\theta^{A}$ satisfies equation~(\ref{subsec:numerical:closedPfaffians}), the set of $M$ Pfaffians $1-$forms define a \define{differential ideal}
\begin{align}
\label{subsec:numerical:differentialIdeal}
\mathcal{I} \equiv \left\{\theta^{A}\right\}_{diff}
\end{align}
also written as $d\mathcal{I} \subset \mathcal{I}$. Since the distribution, $\Delta$, equation~(\ref{subsec:numerical:distribution}), is defined by the Pfaffian $1-$forms, equation~(\ref{subsec:numerical:closedPfaffians}) can also be expressed as 
\begin{align}
\label{subsec:numerical:closedPfaffiansDistrib}
d\theta^{A}\left(\mathbf{X},\mathbf{Y}\right) = \mathbf{X}\left\{\theta^{A}\left(\mathbf{Y}\right)\right\} - \mathbf{Y}\left\{\theta^{A}\left(\mathbf{X}\right)\right\} -  \theta^{A}\left(\left[\mathbf{X},\mathbf{Y}\right]\right) = -\theta^{A}\left(\left[\mathbf{X},\mathbf{Y}\right]\right) = 0~\forall~\mathbf{X},\mathbf{Y}\in \Delta
\end{align}
which is a statement about the closure of the distribution under the Lie bracket. Since $\mathbf{X},\mathbf{Y} \in \Delta$, the Pfaffian $1-$forms, $\theta^{A}$, can only form the basis for a differential ideal if the distribution, $\Delta$, is a \define{closed} vector space, meaning that the Lie bracket of any two vectors $\mathbf{X},\mathbf{Y} \in \Delta$ must satisfy
\begin{align}
\label{subsec:numerical:integrableDistribution}
\left[\mathbf{X},\mathbf{Y}\right] \equiv \mathbf{Z} \in \Delta
\end{align}
Because the distribution, $\Delta$, was defined for smooth vector fields only, the Lie bracket is well defined for all vectors $\mathbf{X},\mathbf{Y} \in \Delta$. When equation~(\ref{subsec:numerical:integrableDistribution}) is satisfied, the distribution is said to be in \define{involution}, a property often expressed as $\left[\Delta,\Delta\right] \subset \Delta$.

Comparing equation~(\ref{subsec:numerical:closedPfaffians}) to Cartan's \define{structure equations}, the $M$ independent Pfaffians, $\theta^{A}$, cannot 
form a differential ideal, $\mathcal{I}$,  if \define{torsion} is present, since non-zero torsion would yield
\begin{align}
\label{subsec:numerical:notClosedPfaffians}
d\theta^{A} = -\omega^{A}~_{B}\wedge \theta^{B} + \tau^{A}
\end{align}
with a non-zero torsion $2-$form $\tau^{A} \in \Omega^{2}\left(\mathcal{M}\right)$. The structure equations also imply that $\omega^{A}~_{B}$ should be treated as a \define{connection}, allowing a calculus in vector spaces over $\mathcal{M}$ to be defined in a coordinate independent manner. Expressing the connection, $\omega^{A}~_{B}$, as a $M\times M$ matrix of $1-$forms in $\Omega^{1}\left(\mathcal{M}\right)$, the $M$ independent Pfaffian $1-$forms, $\theta^{A} \in \Omega^{1}\left(\mathcal{M}\right)$, can be used to write $d\theta^{A}$ as a $2-$form in $\Omega^{2}\left(\mathcal{M}\right)$ as
\begin{align}
\label{subsec:numerical:apparentTorsion}
d\theta^{A} = -\omega^{A}~_{B}\wedge\theta^{B} = -\omega^{A}~_{IB} \theta^{B}~_{J} ~dz^{I}\wedge dz^{J} ~\in \Omega^{2}\left(\mathcal{M}\right)
\end{align}
Using equation~(\ref{subsec:numerical:closedPfaffians}), 
\begin{align}
\label{subsec:numerical:connectionElements}
i^{\star}d\theta^{A} = di^{\star}\theta^{A} = 0
\end{align}
showing that the $M$ $2-$forms, $d\theta^{A}$, vanish when pulled back to $\Sigma$.

A Pfaffian system is \define{integrable} whenever there exists an $N-M$ dimensional integral manifold, $\Sigma$, called a \define{maximal integral manifold}, defined in $\mathcal{M}$ by $M$ coordinates $\mathcal{C}^{A} = 0$. Pfaffian systems which are not integrable are called \define{nonintegrable} systems. 
From  equation~(\ref{subsec:numerical:connectionElements}), all $M$ Pfaffian $1-$forms, $\theta^{A}$, must be vanishing closed forms when pulled back to the integral manifold, $\Sigma$, which, from  equation~(\ref{subsec:numerical:vanishingPullback}), allows each Pfaffian $1-$form, $\theta^{A}$, to be expressed as a locally exact $1-$form
\begin{align}
\label{subsec:numerical:exactPfaffians}
i^{\star}\left(\theta^{A}\right) = i^{\star}\left(d\mathcal{C}^{A}\right) = d~i^{\star}\left(\mathcal{C}^{A}\right) = 0
\end{align}
whenever the distribution, $\Delta$, is in involution. Integrating equation~(\ref{subsec:numerical:exactPfaffians}), the $M$ constraints $\mathcal{C}^{A}$ must be constant in $\mathcal{M}$, since each belongs to the kernel of the pre-image, $i^{-1}:\mathcal{M} \rightarrow \Sigma$, a consequence of equations~(\ref{subsec:numerical:integralManifoldTangent})~and~(\ref{subsec:numerical:exactPfaffians}). Whence a Pfaffian system will be integrable whenever the distribution, $\Delta$, is in involution, or equivalently, whenever the $M$ linearly independent Pfaffians form a basis for the differential ideal, $\mathcal{I}$; a result originally proven by Frobenius \cite{frankel1},\cite{cfb}. 

As a consequence of equations~(\ref{subsec:numerical:exactPfaffians}), if the $M$ Pfaffians, $\theta^{A}$, do not form a differential ideal, $\mathcal{I}$, the $M$ constraints, $\mathcal{C}^{A} = 0 \in \Omega^{0}\left(\mathcal{M}\right)$, cannot define an $N-M$ dimensional maximal integral manifold, $\Sigma$, in $\mathcal{M}$. This is a direct  result of equation~(\ref{subsec:numerical:closedPfaffians}), since if the $M$ Pfaffians, $i^{\star}\left(\theta^{A}\right) \in T^{\star}\Sigma$, do not form a differential ideal, $\theta^{A}$ will not be locally closed and therefore cannot yield a set of $M$ vanishing exact $1-$forms in $T^{\star}\Sigma$. This means that for non-integrable systems, there can be no guarantee that the constraints $\mathcal{C}^{A}\left(\mathbf{z}\right) = 0$, expressed as functions on $\mathcal{M}$, will be preserved, even locally, as the system evolves.
%
\subsection{Gauge Invariant Vector Spaces}
\label{subsec:stability:invariant}
In order to facilitate the examination of the integrability of Hamiltonian formulations of gauge theories provided in subsection~(\ref{subsec:stability:gaugeIntegrable}), it will be useful to derive certain gauge invariant vector spaces as subspaces of the tangent bundle, $\mathcal{V}$, and cotangent bundle, $\mathcal{V}^{\star}$, of the phase space in which the Hamiltonian formulation is defined. Throughout this section, a \define{gauge invariant vector space} will be used to describe a vector space which is preserved under gauge transformations. This does not mean that elements of the vector space will be preserved under gauge transformations, only that any gauge transformation will define a bijective map from the vector space to itself. Determining the gauge invariant vector subspaces of the tangent bundle, $\mathcal{V} \equiv T\mathcal{M}$, defined for a canonical $2N$ dimensional phase space, $\mathcal{M}$, will motivate the need to introduce second class constraints, thereby fixing the gauge, allow gauge transformations in phase space to be projected onto transformations in each gauge invariant subspace, and allow general phase space transformations to be expressed as the sum of a gauge transformation and a transformation which cannot be expressed as a gauge transformation, corresponding to constraint violations. In subsection~({\ref{subsec:stability:removeErr}), a numerical method for projecting out constraint violating transformations will be introduced.

As shown in subsection~(\ref{subsec:extendedHamiltonian}), Hamiltonian formulations of gauge theories generate a set of first class constraints, $\mathcal{C}^{A} \approx 0$, which vanish on the constraint surface and commute weakly with one another as well as with the canonical Hamiltonian, $H$. As a result, there will be a Hamiltonian vector, subsection~(\ref{subsec:lagHamDyn}), associated with each first class constraint, $\mathbf{C}_{A}$, given by
\begin{align}
\label{subsec:numerical:gaugeVecField}
\mathbf{C}_{B} \equiv C^{I}~_{B}\p_{I} \equiv \delta_{BA}J^{KI}~\p_{K}\mathcal{C}^{A}~\p_{I}
\end{align}
so that
\begin{align}
\label{subsec:numerical:gaugeVecFieldEquiv}
d\mathcal{C}^{A} \equiv \frac{d\mathcal{C}^{A}}{d z^L}dz^{L} \equiv \delta^{AB}\omega^{2}\left(\mathbf{C}_{B},\cdot\right) = \delta^{AB}C^{I}~_{B}~J_{IK}~dz^{K}
\end{align}
The Hamiltonian vectors generated by first class constraints, $\mathbf{C}_{A}$, will be referred to as \define{first class Hamiltonian vectors}. For notational convenience, the vectors $\mathbf{C}_{B}$ have been expressed with the constraint index, $B$, lowered. Since the set of first class constraints is assumed to be irreducible, the first class Hamiltonian vectors must be linearly independent. In addition to being linearly independent, each first class Hamiltonian vector, $\mathbf{C}_{B}$, must be tangent to the constraint manifold since
\begin{align}
\label{subsec:numerical:gaugeVecTangent}
\mathbf{C}_{B}\left(\mathcal{C}_{A}\right) \equiv  C^{K}_{B}\p_{K}\mathcal{C}_{A} = \pb{\mathcal{C}_{B}}{\mathcal{C}_{A}} \approx 0
\end{align}
Using equations~(\ref{subsec:constraints:FCCalgebra})~and~(\ref{subsec:numerical:gaugeVecField}), the \define{Lie bracket} of any two first class Hamiltonian vectors, $\mathbf{C}_{B}$ and $\mathbf{C}_{A}$, will be
\begin{align}
\label{subsec:numerical:lieBracketGauge}
\left[\mathbf{C}_{A},\mathbf{C}_{B}\right]  & \equiv \mathbf{C}_{A}\left(\mathbf{C}_{B}\right) - \mathbf{C}_{B}\left(\mathbf{C}_{A}\right) \\\nonumber
& = \left(C^{K}_{A}\p_{K}C^{L}_{B} - C^{K}_{B}\p_{K}C^{L}_{A}\right)\p_{L} \\\nonumber
& = -J^{JL}\p_{J}\left(\Gamma^{C}_{AB}\mathcal{C}_{C}\right)~\p_{L} \approx 0
\end{align}
showing that the first class Hamiltonian vectors commute on the constraint manifold. There should be no confusion between the Lie bracket, which acts on vectors, and the Poisson bracket, which acts on functions. 

Since the first class Hamiltonian vectors are all independent and the Lie bracket of any two first class Hamiltonian vectors vanishes on the constraint manifold, the first class Hamiltonian vectors form a basis on the constraint manifold for a vector space which is a subspace of all vectors tangent to the constraint manifold and which is invariant under all gauge transformations. For a gauge theory with $2(N-M)$ first class constraints, $\left\{\mathcal{C}^{A}\right\}$, embedded into a $2N$ dimensional phase space, $\mathcal{M}$, define the vector space over the basis of $2(N-M)$ independent first class Hamiltonian vectors, $\left\{\mathbf{C}_{A}\right\}$, to be 
\begin{align}
\label{subsec:numerical:gaugeVectorSpace}
\mathcal{V}_{G} \equiv \left\{\mathbf{X} = X^{A}\mathbf{C}_{A} ~ \vert~  X^{A} \in \mathcal{M} \right\}
\end{align}
The vector space $\mathcal{V}_{G}$, restricted to the constraint manifold is a closed subspace tangent to the constraint manifold, and so must be invariant under gauge transformations because the constraint manifold itself is gauge invariant. In addition to $\mathcal{V}_{G}$ it is useful to define the dual vector space
\begin{align}
\label{subsec:numerical:gaugeDualVectorSpace}
\mathcal{V}^{\star}_{G} \equiv \left\{\tilde{\mathbf{Y}} = Y_{A}\tilde{\mathbf{W}}^{A} ~ \vert~  Y_{A} \in \mathcal{M} \right\}
\end{align}
over the basis, $\tilde{\mathbf{W}}^{B}$, satisfying
\begin{align}
\label{subsec:numerical:longitudinalOneForms}
\tilde{\mathbf{W}}^{B}\left(\mathbf{C}_{A}\right) =  \delta^{B}~_{A} 
\end{align}
The $1-$forms, $\tilde{\mathbf{W}}^{B}$, defined in equation~(\ref{subsec:numerical:longitudinalOneForms}) will be referred to as \define{first class Hamiltonian $1-$forms}. Since $\mathcal{V}^{\star}_{G}$ is dual to $\mathcal{V}_{G}$, when restricted to the constraint manifold $\mathcal{V}^{\star}_{G}$ must also be invariant under all gauge transformations. Using equation~(\ref{subsec:numerical:gaugeVecField}), the components of $\tilde{\mathbf{W}}^{B} \in \mathcal{V}^{\star}$ must satisfy
\begin{align}
\label{subsec:numerical:longitudinalOneFormsComponents}
\tilde{\mathbf{W}}^{B}\left(\mathbf{C}_{A}\right) \equiv W^{B}~_{K}~C^{L}~_{A}~ dz^{K}\left(\p_{L}\right) = W^{B}~_{K}~C^{K}~_{A} = \delta^{B}~_{A} 
\end{align}
Although the symplectic form, $\omega^2$, maps the first class Hamiltonian vectors, $\mathbf{C}_{A} \in \mathcal{V}$, to exact $1-$forms in $\mathcal{V}^{\star}_{G}$, equation~(\ref{subsec:numerical:gaugeVecField}), there is in general no canonical way to express the first class Hamiltonian $1-$forms $\tilde{\mathbf{W}}^{B} \in \mathcal{V}^{\star}$, as exact $1-$forms in $\mathcal{V}^{\star}$. Expressing the first class Hamiltonian $1-$forms uniquely as exact $1-$forms would require unique phase space functions, $w^{B}$, satisfying $\pb{w^{B}}{\mathcal{C}^{A}} = \delta^{AB}$ so that $\bar{\mathbf{W}}^{B} = dw^{B}$. Since the first class constraints are constant on the constraint manifold, the phase space functions, $w^{B}$, must also be constant, and therefore would function as second class constraints. Although there is no canonical way to express the first class Hamiltonian $1-$forms as exact forms in phase space, $2(N-M)$ independent $1-$forms satisfying equation~(\ref{subsec:numerical:longitudinalOneForms}) will always exist.

In addition to the vector space $\mathcal{V}_{G}$ and dual $\mathcal{V}^{\star}_{G}$, it will be useful to define the vector space, $\mathcal{V}_{\perp} \subset \mathcal{V}$, which is the space of all vectors $\mathbf{X} \in \mathcal{V}$ orthogonal to the constraint manifold. Using the expression for the components of the first class Hamiltonian $1-$forms, equation~(\ref{subsec:numerical:longitudinalOneFormsComponents}), along with the canonical symplectic form on $\mathcal{M}$, the space of vectors orthogonal to the first class constraint manifold, $\mathcal{V}_{\perp}$, will be defined over the basis vectors
\begin{align}
\label{subsec:numerical:orthogonalVect}
\mathbf{Y}^{B} \equiv W^{B}~_{K}~J^{KL}\p_{L}
\end{align}
so that
\begin{align}
\label{subsec:numerical:orthogonalVectorSpace}
\mathcal{V}_{\perp} \equiv \left\{\mathbf{W} = W_{A}\mathbf{Y}^{A} ~ \vert~  W_{A} \in \mathcal{M} \right\}
\end{align}
Using the definition of the first class Hamiltonian $1-$forms, equation~(\ref{subsec:numerical:longitudinalOneForms}), and first class Hamiltonian vectors, equation~(\ref{subsec:numerical:gaugeVecField}), each $\mathbf{Y}^{B}$ must satisfy
\begin{align}
\label{subsec:numerical:orthogonalProof}
\mathbf{Y}^{A}\left(\mathcal{C}^{B}\right) \equiv W^{A}~_{K}~J^{KL}~\p_{L}\mathcal{C}^{B} = \delta^{AB}
\end{align}
for one, and only one, first class constraint, so will be referred to as \define{first class orthogonal vectors}. The orthogonal vector space $\mathcal{V}_{\perp}$ will have a dual space 
\begin{align}
\label{subsec:numerical:orthogonalDualVectorSpace}
\mathcal{V}^{\star}_{\perp} \equiv \left\{\mathbf{X} = X^{A}\tilde{\mathbf{Z}}_{B} ~ \vert~  X^{A} \in \mathcal{M} \right\}
\end{align}
defined over the basis $1-$forms, $\tilde{\mathbf{Z}}_{B}$ satisfying
\begin{align}
\label{subsec:numerical:orthogonalOneForm}
\tilde{\mathbf{Z}}_{B}\left(\mathbf{Y}^{A}\right) = \delta^{A}~_{B}
\end{align}
Using equation~(\ref{subsec:numerical:orthogonalProof}), the basis $1-$forms $\tilde{\mathbf{Z}}_{B}$ can be expressed as  exact $1-$forms
\begin{align}
\label{subsec:numerical:orthogonalOneFormComponent}
\tilde{\mathbf{Z}}_{B} \equiv \delta_{AB}~d\mathcal{C}^{A}
\end{align}
The $1-$forms $\tilde{\mathbf{Z}}_{B}$ are dual to the first class orthogonal vectors and so will be called \define{first class orthogonal $1-$forms}. From equation~(\ref{subsec:numerical:orthogonalOneFormComponent}), all first class orthogonal $1-$forms must vanish as the system evolves in order for the evolution to remain on the first class constraint manifold, defined by the vanishing of the first class constraints. Since there is no canonical expression for the components of the first class Hamiltonian $1-$forms, there can be no canonical expression for the components of the basis vectors for $\mathcal{V}_{\perp}$. Since the first class constraint manifold, defined by the vanishing of the first class constraints, $\mathcal{C}^{A} \approx 0$, is gauge invariant, the vector space $\mathcal{V}_{\perp}$ and dual vector space $\mathcal{V}^{\star}_{\perp}$ must also be gauge invariant.

As a result of equation~(\ref{subsec:numerical:orthogonalProof}), the vector spaces $\mathcal{V}_{\perp} \subset \mathcal{V}$ and $\mathcal{V}_{G} \subset \mathcal{V}$ must be disjoint gauge invariant proper subspaces of $\mathcal{V}$.  Since the physical content of any gauge theory must reside on the constraint manifold, be invariant under gauge transformations, and must evolve without violating any of the first class constraints, the space of vectors tangent to the physical observables must also form a vector subspace, $\mathcal{V}_{P} \subset \mathcal{V}$, which is orthogonal to both $\mathcal{V}_{G}$ and $\mathcal{V}_{\perp}$. Define the space of vectors tangent to the physical observables, $\mathcal{V}_{P}$ so that $\mathcal{V}$ decomposes as 
\begin{align}
\label{subsec:numerical:physVectOrth}
\mathcal{V} = \mathcal{V}_{P} \oplus \mathcal{V}_{G}\oplus \mathcal{V}_{\perp} 
\end{align}
Since $\mathcal{V}$ is the sum of disjoint vector subspaces which are invariant under gauge transformations, each vector subspace will be tangent to a subspace of $\mathcal{M}$ which will be mapped to itself under gauge transformations. As a result, the canonical phase space $\mathcal{M}$ is also decomposable into a direct sum of disjoint subspaces
\begin{align}
\label{subsec:numerical:physVectOrth}
\mathcal{M} = \mathcal{M}_{P} \oplus \mathcal{M}_{G}\oplus \mathcal{M}_{\perp} 
\end{align}
which are invariant under gauge transformations. Since the tangent space to each subspace of $\mathcal{M}$ is gauge invariant, the dimensions of each the subspaces $\mathcal{M}_{P},\mathcal{M}_{G}$, and $\mathcal{M}_{\perp}$ must be equal to the dimension of the tangent subspaces $\mathcal{V}_{P},\mathcal{V}_{G}$, and $\mathcal{V}_{\perp}$ respectively. Counting the independent basis vectors for each vector subspace yields
\begin{align}
\label{subsec:numerical:physVectOrthDim}
dim\left(\mathcal{M}_{G}\right) & = dim\left(\mathcal{M}_{\perp}\right) = 2(N-M)\\\nonumber
dim\left(\mathcal{M}_{P}\right) & = dim\left(\mathcal{M}\right) - 4(N-M) = 2N - 4(N-M) = 2D
\end{align}
which is in agreement with subsection~(\ref{subsec:diracbracket}) where it was shown that any gauge theory with $2(N-M)$ first class constraints will have $2N - 4(N-M) = 2D$ gauge invariant physical variables corresponding to $D$ degrees of freedom.
%
\subsection{Integrability of Gauge Theories}
\label{subsec:stability:gaugeIntegrable}
Consider the Hamiltonian formulation of a gauge theory with $D$ degrees of physical freedom expressed in a $2N$ dimensional phase space, $\mathcal{M}$, with coordinates, $\left\{z^{I}\right\}$. As shown in subsection~(\ref{subsec:constraints}), the canonical Hamiltonian formulation of a gauge theory with $D$ degrees of freedom will have $2(N-M)$ first class constraints, $\mathcal{C}^{A} \approx 0$, and $2(N-M)$ undetermined multipliers corresponding to the available gauge freedom. In order to show integrability for this system, as defined in subsection~(\ref{subsec:stability:integrability}), it is necessary to show that there exists a map from the $2N$.dimensional phase space $\mathcal{M}$, with coordinates $\left\{z^{I}\right\}$, to a $2N$ dimensional phase space, $\bar{\mathcal{M}}$, with canonical coordinates defined by $2(N-M)$ canonical pairs $\left(\bar{g}_{A},\bar{\mathcal{C}}^{A}\right)$ and $D$ canonical pairs $\left(\bar{q}^{a},\bar{p}_{a}\right)$, in which the $2(N-M)$ canonical momenta $\bar{C}^{A}$, corresponding to the first class constraints, are constant, so that $d\bar{\mathcal{C}}^{A} = 0$ as the system evolves, ensuring that the system remains on the first class constraint manifold. Following the decomposition of the canonical phase space defined in subsection~(\ref{subsec:stability:invariant}), the $2N$ dimensional canonical phase space, $\bar{\mathcal{M}}$, with canonical coordinates $\left\{\bar{q}^{a},\bar{p}_{b},\bar{g}_{A},\bar{\mathcal{C}}^{B}\right\}$, decomposes into $\bar{\mathcal{M}} = \bar{\mathcal{M}}_{P} \oplus\bar{\mathcal{M}}_{G} \oplus\bar{\mathcal{M}}_{\perp}$ with coordinates   
\begin{align}
\label{subsec:numerical:canonicalDecomposition}
\left(\bar{q}^{a},\bar{p}_{b}\right) \in \bar{\mathcal{M}}_{P}\\
\left(\bar{\mathcal{C}}^{A}\right) \in \bar{\mathcal{M}}_{G}\\
\left(\bar{g}_{B}\right) \in \bar{\mathcal{M}}_{\perp}
\end{align}
Using this decomposition, the tangent bundle, $\bar{\mathcal{V}} \equiv T\bar{\mathcal{M}}$, also decomposes, $\bar{\mathcal{V}} =  \bar{\mathcal{V}}_{P} \oplus\bar{\mathcal{V}}_{G} \oplus\bar{\mathcal{V}}_{\perp}$. Using the canonical basis for $\bar{\mathcal{V}}$ defined by $\left\{\p_{\bar{q}^{a}},\p_{\bar{p}_{b}},\p_{\mathcal{\bar{C}}^{A}},\p_{\bar{g}_{B}}\right\}$ the vector subspaces can be expressed as
\begin{align}
\label{subsec:numerical:canonicalVectorDecomposition}
\bar{\mathcal{V}}_{P} & \equiv \left\{\bar{\mathbf{P}} = P^{a}\p_{\bar{q}^{a}} + P^{b}\p_{\bar{p}_{b}} \in \bar{\mathcal{V}} ~\vert~ \bar{\mathbf{P}}\left(\bar{\mathbf{z}}\right) = 0~ \forall~ \bar{\mathbf{z}} \in \bar{\mathcal{M}}_{P}\right\} \\
\bar{\mathcal{V}}_{G} & \equiv \left\{\bar{\mathbf{C}} = C^{A}\p_{\bar{g}_{A}} \in \bar{\mathcal{V}}~\vert~ \bar{\mathbf{C}}\left(\bar{\mathbf{z}}\right) = 0~ \forall~ \bar{\mathbf{z}} \in \bar{\mathcal{M}}_{G}\right\} \\
\bar{\mathcal{V}}_{\perp} & \equiv \left\{\bar{\mathbf{Y}} = Y^{B}\p_{\bar{\mathcal{C}}^{B}} \in \bar{\mathcal{V}}~\vert~ \bar{\mathbf{Y}}\left(\bar{\mathbf{z}}\right) = 0~ \forall~ \bar{\mathbf{z}} \in \bar{\mathcal{M}}_{\perp}\right\}
\end{align}
Note that these definitions do not require that the vector subspaces defined above be closed, meaning that the Lie bracket of any two vector fields in a given subspace can yield a vector field which does not belong to the subspace. 

In order to examine the integrability of the Hamiltonian formulation of a gauge theory, it is necessary to express all constraints as Pfaffian $1-$forms. Begin by using each independent first class constraint, $\mathcal{C}^{A} \approx 0$, to define a $1-$form
\begin{align}
\label{subsec:numerical:FCCPfaffian}
\theta^{A} \equiv d\bar{\mathcal{C}}^{A} - \p_{K}\mathcal{C}^{A~}dz^{K}
\end{align}
which is an element of the cotangent bundle over the product space $\mathcal{M} \times \bar{\mathcal{M}}$. From subsection~(\ref{subsec:stability:invariant}), the $2(N-M)$ $1-$forms defined by equation~(\ref{subsec:numerical:FCCPfaffian}), each satisfy
\begin{align}
\label{subsec:numerical:FCCPfaffianVanish}
\theta^{A}\left(\mathbf{C}_{B}\right) \approx 0
\end{align}
for all first class Hamiltonian vectors, $\mathbf{C}_{A}$, as defined in equation~(\ref{subsec:numerical:gaugeVecField}). As a result of the independence of the first class constraints, the $2(N-M)$ $1-$forms defined in equation~(\ref{subsec:numerical:FCCPfaffian}) will be linearly independent, thereby satisfying equation~(\ref{subsec:numerical:pfaffiansIndependent}). Since the first class Hamiltonian vectors, $\mathbf{C}_{A}$, are all independent and all tangent to the constraint manifold, $\mathbf{C}_{B}\left(\mathcal{C}^{A}\right) \approx 0$, equation~(\ref{subsec:numerical:gaugeVecTangent}), the $2(N-M)$ $1-$forms, $\theta^{A}$, will define a distribution, $\Delta$, which includes all vectors tangent to the constraint manifold. Since the distribution defined by the Pfaffian $1-$forms $\theta^{A}$ includes all vectors tangent to the constraint manifold, thereby including all infinitesimal gauge transformations, and since the Pfaffian $1-$forms $\theta^{A}$ are linearly independent on the constraint manifold, the Pfaffian system for the canonical Hamiltonian formulation will be defined by the $2(N-M)$ independent $1-$forms $\theta^{A}$ given by equation~(\ref{subsec:numerical:FCCPfaffian}). 

The canonical Pfaffian system generated by the $2(N-M)$ independent first class constraints yields a distribution in $\mathcal{V}$ defined as
\begin{align}
\label{subsec:numerical:FCCDistribution}
\Delta_{c} \equiv \left\{\mathbf{V} \in \mathcal{V}~\vert~\theta^{A}\left(\mathbf{V}\right) \approx 0~\forall~\theta^{A}\right\}
\end{align}
From equation~(\ref{subsec:numerical:FCCPfaffian}), the distribution $\Delta_{c}$ will include all vectors in $\mathcal{V}$ which are tangent to the constraint manifold, and so must include the gauge invariant vector spaces, $\mathcal{V}_{P}$ and $\mathcal{V}_{G}$, defined in subsection~(\ref{subsec:stability:invariant}), In addition to the space of vectors tangent to the constraint manifold, the distribution $\Delta_{c}$ must include all first class orthogonal vectors, $\mathbf{Y}^{D}$, defined by equation~(\ref{subsec:numerical:orthogonalVect}), with coefficients which vanish on the constraint manifold. As a result, the distribution, $\Delta_{c}$, will be given by
\begin{align}
\label{subsec:numerical:FCCDistributionCompose}
\Delta_{c} = \mathcal{V}_{P}\oplus\mathcal{V}_{G} \oplus \mathcal{V}^{0}_{\perp}
\end{align}
with
\begin{align}
\label{subsec:numerical:FCOVFzeroSection}
\mathcal{V}^{0}_{\perp} \equiv \left\{\mathbf{W} = W_{A}\mathbf{Y}^{A} \in \mathcal{V}_{\perp}~\vert~ W_{A} = 0\right\}
\end{align}
which is the space of all vectors tangent to the first class orthogonal vectors, $\mathbf{Y}^{A}$, having coefficients which vanish as a result of a particular choice of gauge, $W_{A}\left(\bar{\mathbf{g}}\right) = 0$. The vanishing coefficients, $W_{A}\left(\bar{\mathbf{g}}\right) = 0$, correspond to possible second class constraints. This vector space can be defined throughout phase space since any vector with vanishing coefficients will belong to the distribution, and so for a given constraint manifold defines a vector field which extends off of the constraint manifold. Although $\mathcal{V}_{\perp}$ is a gauge invariant vector space, it cannot be expected that $\mathcal{V}^{0}_{\perp}$ will be gauge invariant since the vanishing coefficients, $W_{A} = 0$, will generically be dependent upon the gauge freedom present in the theory. The canonical Pfaffian system defined by the $2(N-M)$ $1-$forms of equation~(\ref{subsec:numerical:FCCPfaffian}) defines a similar distribution in $\bar{\mathcal{V}}$.

In subsection~(\ref{subsec:stability:integrability}) it was shown that a Pfaffian system will only be integrable if the distribution, $\Delta$, generated by the Pfaffian $1-$forms is in involution, $\left[\Delta,\Delta\right] \subset \Delta$. Using equation~(\ref{subsec:numerical:FCCDistributionCompose}), the Pfaffian system defined for a canonical Hamiltonian formulation yields the following theorem
\begin{thm}
\label{subsec:numerical:theorem:gaugeIntegrability}
Hamiltonian systems with gauge freedom cannot be integrable.
\end{thm} 
In order to prove theorem~\ref{subsec:numerical:theorem:gaugeIntegrability}, it is sufficient to show that any Pfaffian system which contains a Pfaffian $1-$form, $\theta^{A}$, generated by any first class constraint, $\mathcal{C}^{A} \approx 0$, will result in a distribution, $\Delta_{g}$, which cannot be in involution.
\begin{proof}
For any independent Pfaffian, $\theta^{A}$, corresponding to a first class constraint, $\mathcal{C}^{A}$, through equation~(\ref{subsec:numerical:FCCPfaffian}), define vector fields $\mathbf{X} \equiv \mathbf{C}_{A}$ and $\mathbf{Y} \equiv W_{A}\mathbf{Y}^{A} \in \mathcal{V}^{0}_{\perp}$ such that
\begin{align}
\label{subsec:numerical:SCCBracketVect}
\pb{\mathcal{C}^{A}}{W_{A}} \neq 0
\end{align}
in some open region of phase space. The vectors $\mathbf{X}$ and $\mathbf{Y}$ both satisfy, $\theta^{A}\left(\mathbf{X}\right) = \theta^{A}\left(\mathbf{Y}\right) = 0$, on the constraint manifold, and so $\mathbf{X},\mathbf{Y} \in \Delta_{c}$. Using the definitions for the first class Hamiltonian vectors, equations~(\ref{subsec:numerical:gaugeVecField}), and first class orthogonal vectors, equation~(\ref{subsec:numerical:orthogonalVect}), the Lie bracket of the vector fields $\mathbf{X}$ and $\mathbf{Y}$ generates a vector field defined by
\begin{align}
\label{subsec:numerical:proofLieBracket}
\mathbf{Z}\equiv \left[\mathbf{X},\mathbf{Y}\right] \approx C^{L}~_{A}~W^{A}~_{K}~\p_{L}\left(s_{A}\right)~J^{KI}\p_{I} = \pb{\mathcal{C}^{A}}{W_{A}}\mathbf{Y}^{A}
\end{align}
Using the property that first class orthogonal vectors are orthogonal to the constraint manifold, equation~(\ref{subsec:numerical:orthogonalProof}), along with equation~(\ref{subsec:numerical:SCCBracketVect}), the Pfaffian $1-$form, $\theta^{A}$, equation~(\ref{subsec:numerical:FCCPfaffian}), and vector field $\mathbf{Z}$, equation~(\ref{subsec:numerical:proofLieBracket}), will satisfy
\begin{align}
\label{subsec:numerical:proofNonintegrable}
\theta^{A}\left(\mathbf{Z}\right) \neq 0
\end{align}
everywhere in some open region of phase space. As a result, $\mathbf{Z} ~\displaystyle{\not}{\in} \Delta_{c}$,  showing that the distribution, $\Delta_{c}$, cannot be in involution, $\left[\Delta_{c},\Delta_{c}\right] \displaystyle{\not}{\subset} \Delta_{c}$, anywhere in this open region of phase space. Whence the Pfaffian system is nonintegrable. Since this has been shown for any choice of first class constraint, $\mathcal{C}^{A}$, it will be true for all first class constraints, proving that any Hamiltonian formulation with gauge freedom cannot be integrable.
\end{proof}

As shown is subsection~(\ref{subsec:diracbracket}), Hamiltonian formulations of gauge theories with $2(N-M)$ first class constraints, corresponding to the generators of gauge transformations, can be gauge fixed by the introduction of $2(N-M)$ constraints which yield an invertible matrix of commutation relations with the $2(N-M)$ first class constraints. As a result, once the gauge has been completely fixed the $2(N-M)$ original first class constraints are converted to second class constraints, yielding a total of $4(N-M)$ second class constraints given by the $4(N-M)$ independent phase space functions $\mathcal{C}^{A} = \mathcal{S}_{B} = 0$ which satisfy the $2(N-M)$ equations $\pb{\mathcal{C}^{A}}{\mathcal{S}_{A}} \neq 0$ everywhere in some neighborhood of the original first class constraint manifold. The remaining independent components of the original $2N$ dimensional canonical phase space are given by the $4M - 2N \equiv 2D$ physical observables which, at each point on the constraint manifold, locally form a symplectic manifold with $D$ degrees of physical freedom. As a consequence, theorem~\ref{subsec:numerical:theorem:gaugeIntegrability} has the following corollary
\begin{cor}
\label{subsec:numerical:theorem:gaugeFixedIntegrability}
Hamiltonian formulations of gauge theories with all gauge freedom uniquely fixed through the Dirac bracket will be integrable.
\end{cor}
\begin{proof}
Using the $4(N-M)$ second class constraints, $\mathcal{C}^{A} = \mathcal{S}_{B} = 0$, define $4(N-M)$ Pfaffian $1-$forms 
\begin{align}
\label{subsec:numerical:integrablePfaffians}
\theta^{A} \equiv d\mathcal{C}^{A}\\\nonumber
\theta_{B} \equiv d\mathcal{S}_{B}
\end{align}
From equation~(\ref{subsec:numerical:FCCDistributionCompose}), the distribution generated by this Pfaffian system must be $\mathcal{V}_{P}$, the vector space tangent to the space of physical observables. Since the $4(N-M)$ Pfaffians are exact, equation~(\ref{subsec:numerical:integrablePfaffians}), they must form the basis for a differential ideal, $\mathcal{I}$. Therefore the manifold defined by the $2D$ physical observables will form a $2N - 4(N-M) = 2D$ dimensional integral manifold, $\Sigma_{P}$, which is of maximal rank. Whence, Hamiltonian formulations of gauge theories in which all gauge freedom has been uniquely fixed through the Dirac bracket will be integrable.
\end{proof}
It is important to note that theorem~(\ref{subsec:numerical:theorem:gaugeFixedIntegrability}), along with the corollary~(\ref{subsec:numerical:theorem:gaugeFixedIntegrability}), were proven without reference to a particular  set of constraints, only the existence of first class constraints when gauge freedom is present and the ability to define phase space functions which uniquely fix the gauge freedom, and thereby fail to commute with the first class constraints. 

Once the gauge has been fixed, a unique invertible map from $\mathcal{M}$ to $\bar{\mathcal{M}}$ can be locally defined everywhere in some neighborhood of the constraint manifold. As a result, elements of $\mathcal{V}$ can be decomposed into contributions from the subspaces $\mathcal{V}_{P}, \mathcal{V}_{G}$ and $\mathcal{V}_{\perp}$, which are defined in terms of a given constraint manifold. A particularly useful consequence of this is the ability to construct a path in phase space between any two solutions residing in the same open region on which the Dirac bracket is invertible, thus allowing any phase space variation violating the constraints to be identified uniquely and removed. This feature will be explored further in subsection~(\ref{subsec:stability:removeErr}). Without fixing the gauge, there would be no way to uniquely specify a map defining the phase space components $\bar{q}^{a},\bar{p}_{b},\bar{g}_{B} \in \bar{\mathcal{M}}$ as functions of the phase space coordinates $z^{I} \in \mathcal{M}$ and therefore no way to restrict the evolution to the first class constraint manifold. If it were possible to completely restrict the evolution to the constraint manifold, any Hamiltonian formulation, with or without gauge freedom, would be integrable, but manifestly restricting to the constraint manifold would require the first class constraints to strongly vanish, in contradiction with the definition of the first class constraints. 
%
\subsection{Hyperbolicity}
\label{subsec:numerical:hyperbolicity}
This subsection assumes that the reader is familiar with the notions of hyperbolicity and well-posedness for differential systems, topics which are thoroughly covered elsewhere \cite{kgo}. Additionally, in section~(\ref{subsec:examples}) the reader will be assumed to be familiar with basic pseudo-differential methods which are necessary to define hyperbolicity and well-posedness for second order partial differential systems \cite{nor},\cite{brown3}. Throughout these notes, only hyperbolic gauge theories will be considered.

Suppose a complete set of time independent second class constraints has been imposed, uniquely fixing all gauge freedom. The resulting gauge fixed extended Hamiltonian will generate the following equations of motion
\begin{align}
\label{subsec:numerical:eomExtended}
\dot{\mathbf{z}} = \pb{\mathbf{z}}{H} + \pb{\mathbf{z}}{\lambda^{A}\mathcal{C}_{A}} + \pb{\mathbf{z}}{\gamma_{B}\mathcal{S}^{B}}
\end{align}
where $\mathcal{C}_{A}, \mathcal{S}^{B}$ denote the set of first class and second class constraints respectively and $\lambda^{A},\gamma_{B}$ denote their Lagrange multipliers. As shown in subsection~(\ref{subsec:diracbracket}), all $4(N-M)$ Lagrange multipliers are defined as phase space functions through the Dirac bracket, manifestly satisfying
\begin{align}
\label{subsec:numerical:eomExtendedDBEquiv}
\frac{d \mathcal{C}_{A}}{dt} & = \pb{\mathcal{C}_{A}}{H_{E}} = \pb{\mathcal{C}_{A}}{H}_{D} = 0\\\nonumber
\frac{d \mathcal{S}^{B}}{dt} & = \pb{\mathcal{S}^{B}}{H_{E}} = \pb{\mathcal{S}^{B}}{H}_{D} = 0
\end{align}
Since the gauge fixed extended Hamiltonian, $H_{E}$, preserves all $4(N-M)$ constraints, the Hamiltonian vector field generated by $H_{E}$ must simultaneously satisfy all $4(N-M)$ independent Pfaffian $1-$forms generated by the constraints, subsection~(\ref{subsec:stability:gaugeIntegrable}), and therefore must be an element of the distribution, $\Delta$. Using the results of subsection~(\ref{subsec:stability:gaugeIntegrable}), because the gauge has been fixed through the introduction of a complete set of second class constraints, any vector belonging to the distribution, $\mathbf{V} \in \Delta$, must have vanishing coefficients for all components tangent to any first class Hamiltonian vector, $\mathbf{C}_{A}$,  or any first class orthogonal vector, $\mathbf{Y}^{B}$. As a result, given a solution $\mathbf{z}\left(t_{0}\right)$ at time $t_{0}$ simultaneously satisfying all first class constraints and imposed second class constraints, along the phase space flow generated by the gauge fixed extended Hamiltonian, equation~(\ref{subsec:numerical:eomExtended}), the $2(N-M)$ variational vectors
\begin{align}
\label{subsec:numerical:fccVariationalVectors}
\delta_{A}\mathbf{z} & \equiv \pb{\mathbf{z}}{\mathcal{C}_{A}} \equiv \delta_{AB}\mathbf{C}^{B}\left(\mathbf{z}\right) \\\nonumber
& \mathbf{C}^{A} \in \mathcal{V}_{G}
\end{align}
generated by the original first class constraints as well as the $2(N-M)$ variational vectors
\begin{align}
\label{subsec:numerical:sccVariationalVectors}
\delta^{B}\mathbf{z} & \equiv \pb{\mathbf{z}}{\mathcal{S}^{B}} \equiv \mathbf{S}^{B}\left(\mathbf{z}\right)\\\nonumber
& \mathbf{S}^{B} \in \mathcal{V}_{\perp}
\end{align}
generated by the imposed second class constraints must be each be preserved. The invertibility of the Dirac bracket ensures that these variational vectors are preserved since each of the $4(N-M)$ constraints must simultaneously commute with the gauge fixed extended Hamiltonian and vanish strongly, $\mathcal{C}_{A} = \mathcal{S}^{B} = 0$, throughout some open neighborhood of the initial constraint manifold. Since the $4(N-M)$ variational vectors of equations~(\ref{subsec:numerical:fccVariationalVectors})~and~(\ref{subsec:numerical:sccVariationalVectors}) are generated by $4(N-M)$ independent tangent vectors, forming a basis for the vector subspace $\mathcal{V}_{G}\oplus \mathcal{V}_{\perp} \subset \mathcal{V}$, on any given second class constraint manifold all variations tangent to any of the $4(N-M)$ variational vectors must vanish otherwise, as the system evolves, some of the $4(N-M)$ constraints defining the second class constraint manifold would be violated. The evolution equations generated by the gauge fixed extended Hamiltonian, which keep all $4(N-M)$ constraints constant, therefore ensure that no variations tangent to those of equations~(\ref{subsec:numerical:fccVariationalVectors})~or~(\ref{subsec:numerical:sccVariationalVectors}) enter the system.

Since the $4(N-M)$ independent variational vectors are preserved under the flow generated by the gauge fixed extended Hamiltonian, each variational vector must correspond to an independent eigenvector for the system of evolution equations generated by the gauge fixed extended Hamiltonian. Because all $4(N-M)$ constraints strongly commute with the gauge fixed extended Hamiltonian, the phase space flow generated by $H_{E}$ will remain on the initial second class constraint manifold. Since each of the $4(N-M)$  preserved variational vectors must have vanishing coefficients on the initial second class constrain manifold, each of the corresponding eigenvectors must each yield a zero eigenvalue. This should be an anticipated result since each of the $4(N-M)$ constraints defines a constant of motion, corresponding to a value in phase space which propagates with zero speed. By assumption the Hamiltonian formulation is hyperbolic, therefore the remaining $2N - 4(N-M) = 2D$ independent eigenvectors will be given by the $2D$ complex variational vectors
\begin{align}
\label{subsec:numerical:evoVariationalVectors}
\delta \mathbf{z}^{\pm}_{a} \equiv \p_{t}{\mathbf{z}}^{\pm}_{a} \equiv \pb{\mathbf{z}^{\pm}_{a}}{H_{E}}
\end{align}
with $\mathbf{z}^{\pm}_{a} \equiv q^{a} \pm ip_{a}$ for $a = 1,\dots,D$ corresponding to the $D$ canonical conjugate pairs $\left(q^{a},p_{a}\right)$. The explicit form of the $D$ canonical conjugate pairs can be found using the Dirac bracket, as described in subsection~(\ref{subsec:diracbracket}). Although the exact expression for the $D$ canonical conjugate pairs, found using the Dirac bracket, will be dependent upon the choice of second class constraints used to fix the gauge, the $D$ physical degrees of freedom are themselves gauge independent. Therefore, if equation~(\ref{subsec:numerical:evoVariationalVectors}) yields $D$ real pairs of eigenvalues for any choice of second class constraints, which when imposed uniquely fix the gauge, it must yield $D$ real pairs of eigenvalues for any choice, independent of how the gauge is fixed. This means that the hyperbolicity of any Hamiltonian formulation will be independent of the gauge freedom present. 
Using this result, equations~(\ref{subsec:numerical:fccVariationalVectors})~and~(\ref{subsec:numerical:sccVariationalVectors}) provide $4(N-M)$ eigenvectors each having eigenvalue equal to zero, and equation~(\ref{subsec:numerical:evoVariationalVectors}) provides the remaining $2D$ eigenvectors each having the real pair of non-zero eigenvalues $\pm \omega_{a} \in \mathbb{R},~ \left|\omega_{a}\right| > 0$. Whence gauge fixed Hamiltonian formulations will posses a complete set of independent eigenvectors resulting in a strongly hyperbolic, and therefore well-posed, system. 

Since an independent second class constraint must be imposed for each independent first class constraint in order to completely fix the gauge, any Hamiltonian formulation in which the gauge is not fixed will not have a full set of conserved second class constraints necessary to generate a full set of  conserved variational vectors, equation~(\ref{subsec:numerical:sccVariationalVectors}). Because of this, any Hamiltonian formulation containing gauge freedom will generate a system of evolution equations which cannot posses a complete set of eigenvectors. Whence, Hamiltonian formulations containing gauge freedom can form only weakly hyperbolic systems at best. 

Evolution equations derived using the Dirac bracket will not, by design, propagate any non-physical quantities, hence the constraints will each correspond to a zero eigenvalue for the system of evolution equations. While zero-modes are normally cause for concern, the Dirac bracket ensures that constraints come in matched pairs, so there will be no degeneracy in the eigenvectors for the resulting system, the physical quantities will have physical propagation speeds, and the system will be stable. In this system, error does not propagate, so the difference between two solutions, which differ at an initial time by some perturbation of the fields, will remain exponentially bounded in time as the system evolves. Numerical error will, naturally, still enter into the system but any such error can be controlled directly using the methods we present in section~(\ref{subsec:stability:removeErr}). A concrete example is provided in section~(\ref{subsubsec:em:shGFsys}).
%
\subsection{Removing Numerical Error}
\label{subsec:stability:removeErr}
The role of the extended Hamiltonian is to fix the gauge completely, yet it is inevitable that error will be introduced in any numerical simulation. When this occurs, finite numerical error will map the system on to an alternate extremal path within some neighborhood of the original solution. The difference in solutions will correspond to a different neighboring gauge choice, which violates the analytic values of second class constraints, or a different neighboring set of solutions for the first class constraints. In either case, the result will be that any error introduced into the system will alter the Lagrange multipliers of the constraints which are found in the gauge fixed extended Hamiltonian. Generically, the Lagrange multipliers in the extended Hamiltonian of the first class constraints will depend on some combination of the time derivatives of the second class constraints projected onto the first class constraint manifold, while the Lagrange multipliers of the second class constraints will involve time derivatives of the first class constraints. 

From equation~(\ref{subsec:diracbracket:LagrangeMultipliers}), the Lagrange multipliers in the extended Hamiltonian are proportional to constraints and so vanish on the constraint surface, leaving only the first class Hamiltonian, $H_{FC}$. Using $H_{FC}$ to derive error correction coefficients along the evolution path yields 
\begin{align}
\label{subsec:numerical:multipliersNumericalError}
\epsilon^{\left(1\right)}_{\Gamma}\left(x\right) = \int d^3x^\prime~\left\{
		D_{\Gamma\Phi}\left(x,x^{\prime}\right)
		\left[
			\mathcal{C}^{\Phi}\left(x^{\prime}\right),
			H_{FC}
		\right]~\right\}
\end{align} 
which is just the definition for the Lagrange multipliers in the gauge fixed extended Hamiltonian. The role of these Lagrange multipliers is to freeze the non-physical degrees of freedom, removing error from propagating. 

As mentioned in subsection~(\ref{subsec:stability:gaugeIntegrable}), once the gauge is fixed, there will exist a canonical method for projecting out all constraint violations. The canonical method is provided by the additional terms generated by equation~(\ref{subsec:numerical:multipliersNumericalError}), which have the exact form of the original Lagrange multipliers. Comparing equation~(\ref{subsec:numerical:multipliersNumericalError}) to the form of the gauge fixed Lagrange multipliers, equation~(\ref{subsec:diracbracket:LagrangeMultipliers}), reveals that the modified Lagrange multiplier terms, calculated using the first class Hamiltonian, serve to alter the phase space path in which non-physical terms do not propagate. Augmenting the gauge fixed extended Hamiltonian by the terms of equation~(\ref{subsec:numerical:multipliersNumericalError}) will generate evolution equations for the field variables in which the constraint violations are corrected up to first order in the error. When the constraint algebra does not strongly close, higher order corrections to the propagation of error involve variations in the Dirac bracket itself, with second order terms
\begin{align}
\label{subsec:numerical:multipliersNumericalErrorSecondOrder}
\epsilon^{\left(2\right)}_{\Gamma}\left(x\right) = \int d^3x^\prime~\left\{
		D_{\Gamma\Phi}\left(x,x^{\prime}\right)
		\pb{\pb{\mathcal{C}^{\Phi}}{H_{FC}}}{H_{FC}} + \frac{1}{2}D_{\Gamma\Sigma}D_{\Phi\Theta}\pb{\pb{\mathcal{C}^{\Sigma}}{\mathcal{C}^{\Theta}}}{H_{FC}}\pb{\mathcal{C}^{\Phi}}{H_{FC}}
	~\right\}
\end{align}
Using the extended gauge fixed Hamiltonian, all order corrections about the current solution can be generated, yielding the total correction term 
\begin{align}
\label{subsec:numerical:multipliersNumericalErrorAllOrder}
\epsilon_{\Gamma} \equiv \sum_{n=1}^{\infty} \epsilon_{\Gamma}^{\left(n\right)}
\end{align}
Each $\epsilon_{\Gamma}^{\left(n\right)}$ is of order $N$ in the time derivative generated by $H_{FC}$. Using these results, the modified gauge fixed extended Hamiltonian takes the form
\begin{align}
\label{subsec:numerical:numExtHam}
H_{N} & = H_{FC} - \epsilon_{\Gamma}\mathcal{C}^{\Gamma}
\end{align}
which is always weakly equal to $H_{E}$ on the constraint surface and strongly equal to $H_{E}$ when the constraint commutation relations remain independent of the evolved fields, and converge as the numerical error goes to zero. They differ only off-shell, with $H_{N}$ serving to control error propagation to higher orders.

For all gauge theories considered here, the multipliers of second class constraints will depend on linear combinations of the evolution of the first class constraints as generated by the canonical Hamiltonian. This means that all multipliers of second class constraints must vanish weakly on the first class constraint manifold. This is not necessarily true for the remaining Lagrange multipliers which are proportional to the evolution of second class constraints and may be non-zero and contribute to the evolution. Higher order terms capture changes to the constraint algebra, and first class Hamiltonian, as the system evolves. The canonical method presented here, applicable to all gauge theories, is a generalization motivated by the methods originally implemented for General Relativity by Brown and Lowe \cite{brown2}. 

The gauge fixing process fixes the propagation of non-physical fields along the flow generated by the extended Hamiltonian, but does not address numerical error already within a system. In general, given a set of values for the canonical fields at some fixed time, $t=t_0$, the generator of transformations which remove numerical error on the fixed time slice will have no terms proportional to the evolution generated by $H_{FC}$, taking the form
\begin{align}
\label{subsec:numerical:numErrTransGenerating}
E^{\left(0\right)}_{N} & = - \epsilon^{\left(0\right)}_{\Gamma}\mathcal{C}^{\Gamma}
\end{align}
at some $t_0$, with the error terms $-\epsilon^{\left(0\right)}_{\Gamma}$ proportional to the constraints themselves and calculated using the current off-shell values for all canonical variables. For reference, equation~(\ref{subsec:numerical:numErrTransGenerating}) will be referred to as the \define{error correction generating function}. 

From equation~(\ref{subsec:hamiltonian:fccICTVar}) and the definition of the Dirac bracket, equation~(\ref{subsec:diracbracket:diracBracket}), at a fixed time the generator of error correction transformations, $E^{\left(0\right)}_{N}$, must have terms of the form
\begin{align}
\label{subsec:numerical:multipliersNumericalErrorFixed}
\epsilon^{\left(1\right)}_{\Gamma}\left(t_{0},x\right) = \int d^3x^\prime~
		D_{\Gamma\Phi}\left(x,x^{\prime}\right)
		\bar{\mathcal{C}}^{\Phi}\left(t_{0},x^{\prime}\right)
\end{align}
to first order, with $\bar{\mathcal{C}}^{\Phi}$ being the numerical value of the constraint $\mathcal{C}^{\Phi}$ at time $t_0$. Since barred coordinates express the current numerical state, they are fixed terms and commute with all canonical variables. Inserting these values into equation~(\ref{subsec:numerical:numErrTransGenerating}) and examining the transformation generated for a given constraint, $\mathcal{C}^{\Gamma}$, yields
\begin{align}
\label{subsec:numerical:numericalErrorRemovalOnConstraint}
\begin{split}
\hat{\delta} \mathcal{C}^{\Gamma} = \pb{\mathcal{C}^{\Gamma}}{E_{N}} &= -\pb{\mathcal{C}^{\Gamma}}{\bar{\mathcal{C}}^{\Phi}D_{\Phi\Theta}\mathcal{C}^{\Theta}}\\
& =  -\bar{\mathcal{C}}^{\Phi}D_{\Phi\Theta}\pb{\mathcal{C}^{\Gamma}}{\mathcal{C}^{\Theta}} = -\bar{\mathcal{C}}^{\Phi}D_{\Phi\Theta}D^{\Gamma\Theta} = -\bar{\mathcal{C}}^{\Phi}\delta^{\Gamma}_{\Phi}\\
& = -\bar{\mathcal{C}}^{\Gamma}
\end{split}
\end{align}
which shows that, the error correction generating function yields a transformation on the canonical space which exactly removes the numerical error, to first order. When the Dirac bracket has a dependency on the canonical variables or fields directly, higher order terms may arise that are on the order of the square of the current numerical error, calculated by replacing $H_{FC}$ with $E_{N}$ in equation~(\ref{subsec:numerical:multipliersNumericalErrorSecondOrder}) and then updating the constraint coefficients in equation~(\ref{subsec:numerical:numErrTransGenerating}). These subsequent terms will be of order $N$, for the $N^{th}$ update, in the constraints, so may be neglected for states close to the constraint surface. Far from the constraint surface, higher order terms may become important, but numerically it will be sufficient to compute first order updates iteratively until some constraint tolerance is reached, which is beneficial when the higher order terms become increasingly complicated.
%
\subsection{Synopsis of Stability}
\label{subsec:numerical:synopsis}
Hamiltonian formulations of gauge theories in which the gauge is not fixed cannot yield integrable systems precisely because the theory has been embedded in a canonical phase space. If it were possible to restrict to the first class constraint manifold, the system would become integrable, but this would require that all first class constraints strongly vanish. By introducing a set of second class constraints, uniquely fixing the gauge and thereby allowing the Dirac bracket to be constructed, all first class constraints are converted to second class constraints, which strongly vanish. Once the theory is restricted to a constraint manifold defined by a set of strongly vanishing constraints, the theory will become integrable. Restricting to an integrable system will always yield a formulation which is at least strongly hyperbolic, given a canonical formulation which is weakly hyperbolic, yielding a well-posed problem. Introducing second class constraints provides a canonical method to project out numerical error. Stability addresses the ability for differences in initial conditions to be bounded exponentially in time as the system evolves, so the presence of zero modes should not be disconcerting since the Dirac bracket ensures that the resulting system of evolution equations possesses a complete set of eigenvectors. This prohibits the propagation of non-physical modes, and through the methods we introduced here for the removal of numerical error, provides equations for evolution off-shell and a systematic prescription for returning an evolved system with numerical error back to the original constraint surface.

\section{Hamiltonian Formulation of Electrodynamics}
\label{subsec:examples}
Electrodynamics provides a particularly simple yet useful setting in which to apply the methods of section~(\ref{sec:GaugeHamiltonian})  as the gauge group for the theory is abelian, consisting of a set of commuting first class constraints, yet exhibits many features found in more complicated theories, making it an ideal example. We begin this section with a review of the canonical formulation of Electrodynamics, subsection~(\ref{subsubsec:em:canonEvolution}), deriving the first class constraints and first class constraint algebra along with the canonical evolution equations. In subsection~(\ref{subsubsec:em:sccAndDB}), we define the Coulomb gauge and impose it on the canonical formulation through second class constraints, which is shown (as expected) to completely fix the gauge freedom present in the theory.  A Dirac bracket is derived for this choice of second class constraints, and we construct the gauge fixed extended Hamiltonian
and the gauge fixed equations of motion. We discuss relation between the canonical and the gauge fixed formulations. We then discuss (in subsection~(\ref{subsubsec:em:shGFsys})) the hyperbolicity of both the canonical and gauge fixed formulations. The canonical Hamiltonian generates a weakly hyperbolic system, while the gauge fixed extended Hamiltonian generates a strongly hyperbolic well-posed system. These results are in agreement with the analysis presented in section~(\ref{subsec:numerical:hyperbolicity}) for general gauge theories.
%
\subsection{Action and Canonical Evolution}
\label{subsubsec:em:canonEvolution}
The Maxwell action in vacuum is
\begin{align}
\label{subsec:DBexamples:em:action}
S = \int vol^4~\left[A_{\alpha}J^{\alpha}-\frac{1}{4}F_{\alpha\beta}F^{\alpha\beta}\right] \equiv \int dt ~ \int dx^3~\mathcal{L}
\end{align}
where the $1-$form $A_{\alpha}dx^\alpha$ is the \emph{electromagnetic potential} and the field strength $2$ form, $\mathbf{F}^{(2)} \equiv F_{\alpha\beta}dx^\alpha \wedge dx^\beta$, is defined by
\begin{align}
\label{subsec:DBexamples:em:fieldStrength}
F_{\alpha\beta} &\equiv dA^{(1)} \equiv A_{\beta;\alpha} - A_{\alpha;\beta} = A_{\beta;\alpha}~dx^\alpha \wedge dx^\beta
\end{align}
The coefficient of the volume element has been absorbed into the definition of $\mathcal{L}$ in the action, equation~(\ref{subsec:DBexamples:em:action}). The electric and magnetic fields take their values from the field strength $2$ form, $\mathbf{F}^{(2)}$, and are defined by
\begin{align}
\label{subsec:DBexamples:em:fields1}
\mathcal{E}_{i}dx^i = -F_{0i}dx^i = F_{i0}~dx^i~&~\leftrightarrow~~ E^{j} = F^{0i}\\\nonumber
\mathcal{B}_{ij}dx^i\wedge dx^j = F_{ij}~dx^i\wedge dx^j ~&~\leftrightarrow~~ B^i = \star F^{0i}
\end{align}
Here $\star$ is the \emph{Hodge operator} mapping $p$ forms to their \emph{Hodge dual}. The hodge operator depends on the metric of the base manifold. Differential forms which are restricted to the spatial manifold, such as the spatial $1-$form $\mathcal{E}^{(1)}$ and spatial $2-$form $\mathcal{B}^{(2)}$ defined above, will be denoted in script throughout this section. It is also useful at this point to introduce the \emph{codifferential} operator defined by
 \begin{align}
\label{subsec:DBexamples:em:coBoundary}
d^\ast A^{(p)} \equiv \star d \star A^{(p)}
\end{align} 
which sends $p$ forms to $(p-1)$ forms. 

It is common to separate the temporal and spatial components of the electromagnetic potential to simplify the notation when explicitly dealing with a spacetime split, as is necessary in the Hamiltonian formulation. The standard definitions found in any textbook which deals with the electromagnetic potential are  
\begin{align}
\label{subsec:DBexamples:em:redefFields1}
A_0 ~dt \equiv \phi ~dt\\\nonumber
A_{i}~dx^i \equiv \mathcal{A}^{(1)}
\end{align}
The temporal term is commonly referred to as the \emph{scalar potential}, $\phi$, while the spatial vector, $\vec{A}$, which is the contravariant version of the covariant $\mathcal{A}^{(1)}$, is commonly known as the \emph{vector potential}. Again, as noted above, any form denoted in script, e.g.  $\mathcal{E}^{(1)},\mathcal{A}^{(1)},\mathcal{B}^{(2)}$, will correspond to a form defined on the spatial manifold. The exterior derivative restricted to the spatial slice will be denoted in bold, $\mathbf{d}$, along with the spatial codifferential operator, $\mathbf{d}^\ast \equiv \ast_S \mathbf{d} \ast_S$, where $\ast_{S}$ denotes the Hodge operator on the spatial slice. Using these definitions, Maxwell's equations can be written in terms of the physically familiar electric and magnetic fields
\begin{align}
\label{subsec:DBexamples:em:fields2}
\mathcal{E}^{(1)} =  \partial_{0}\mathcal{A}^{(1)}-\mathbf{d}\phi~&~\leftrightarrow ~~ \vec{E} = \partial_0 \vec{A} - \vec{\nabla} \phi \\\nonumber
\mathcal{B}^{(2)} = \mathbf{d}\mathcal{A}^{(1)} ~&~\leftrightarrow ~~ \vec{B} = \vec{\nabla} \times \vec{A}
\end{align}
The field strength $2$ form, $\mathbf{F}^{(2)}$, and its dual, $\star \mathbf{F}^{(2)}$, which reside on the 4 manifold become
\begin{align}
\label{subsec:DBexamples:em:fieldStrength2}
\mathbf{F}^{(2)} = \mathcal{E}^{(1)}\wedge dt + \mathcal{B}^{(2)} \\\nonumber
\star \mathbf{F}^{(2)} = \mathcal{H}^{(1)} \wedge dt - \mathcal{D}^{(2)}
\end{align}
When the 4 manifold is Minkowski, the $p$ forms $\mathcal{H}^{(1)}$ and $\mathcal{D}^{(2)}$ are related to the $p$ forms $\mathcal{E}^{(1)}$ and $\mathcal{B}^{(2)}$ through
\begin{align}
\label{subsec:DBexamples:em:minkDuals}
\ast_S \mathcal{E}^{(1)} & \equiv  \mathcal{D}^{(2)}\\\nonumber
\ast_S \mathcal{B}^{(2)} & \equiv  \mathcal{H}^{(1)}
\end{align} 
Whenever the manifold is not Minkowski these relation will be considerably more complicated. In order to avoid this complication, we use the Minkowski metric throughout this section. The spatial manifold will assumed to be asymptotically flat, with sufficient locality and fall off of source fields, so that any boundary terms arising from integrating by parts will identically vanish.

In the Hamiltonian formulation, the 4 elements of the electromagnetic potential, $(\phi, A_{i})$, define the configuration space. The momenta conjugate to the vector potential are
\begin{align}
\label{subsec:DBexamples:em:momenta}
\pi^{i} \equiv \frac{\delta ~\mathcal{L}}{\delta \left(\partial_{0}A_{i}\right)} = F^{i0} \equiv E^{i}
\end{align}
Because $\mathbf{F}^{(2)}$ is anti-symmetric, equation~(\ref{subsec:DBexamples:em:fieldStrength}), the momentum conjugate to the scalar potental, $\phi = A_{0}$, must vanish, yielding the primary constraint 
\begin{align}
\label{subsec:DBexamples:em:primaryFCC}
\pi^0 \approx 0
\end{align}
The canonically conjugate pairs are $\left(A_{i},\pi^{i}\right)$ and $\left(\phi,\pi^0\right)$, and the canonical Hamiltonian in vacuum is
\begin{align}
\label{subsec:DBexamples:em:canonHamiltonian}
H_{0} & \equiv \int d^3x~\mathcal{H}_0 \\\nonumber
~& =  \int d^3x~\left\{\frac{1}{2}\pi^i\pi_i + \frac{1}{4}F^{ij}F_{ij} - \phi_{;j}\pi^{j}\right\}\\\nonumber
~& = \int d^3x~\left\{\frac{1}{2}\vec{\pi}\cdot \vec{\pi} + \frac{1}{2}\left(\vec{\nabla}\times\vec{A}\right) \cdot \left(\vec{\nabla}\times\vec{A} \right)- \vec{\nabla}\phi\cdot \vec{\pi}\right\} \\\nonumber
~& =  \int d^3x~\left\{\frac{1}{2}\vec{\pi}\cdot \vec{\pi}- \vec{\nabla}\phi \cdot \vec{\pi}\right\} + \frac{1}{2} \int \left\{\mathbf{d}\mathcal{A}^{(1)}\wedge\ast_S\mathbf{d}\mathcal{A}^{(1)}\right\}
\end{align}
In vacuum, the time derivative of the primary first class constraint, $\pi^0 \approx 0$, generates the consistency constraint
\begin{align}
\label{subsec:DBexamples:em:secondFCC}
\dot{\pi}^0 \equiv \left[\pi^0,H_{0}\right] = -\pi^{j}_{;j} \approx 0
\end{align}
From section~(\ref{sec:GaugeHamiltonian}), this secondary constraint must also weakly vanish in order for the canonical Hamiltonian to generate consistent dynamics on the first class constraint manifold. The constraint, $\pi^{i}_{;i} \approx 0$ strongly commutes with the canonical Hamiltonian, due to the fact that $\mathbf{d}^{\ast}\mathbf{d}^{\ast} = \ast_{S}\mathbf{d}\mathbf{d}\ast_{S} = 0$, hence no further consistency constraints are generated by the canonical Hamiltonian. When charge is present, the secondary constraint of equation~(\ref{subsec:DBexamples:em:secondFCC}) corresponds to the familiar \define{Gauss's law}
\begin{align}
\label{subsec:DBexamples:em:gaussLaw}
\vec{\nabla}\cdot \vec{E} \equiv \pi^{i}_{;i} = \rho
\end{align}
Since secondary constraint $\vec{\nabla}\cdot\vec{\pi} \approx 0$ commutes with both the canonical Hamiltonian and primary constrant $\pi^{0} \approx 0$, the first class constraint algebra for electrodynamics is given by the two strongly commuting first class constraints
\begin{align}
\label{subsec:DBexamples:em:constraintAlgebra}
\pi^{0} \approx 0\\\nonumber
\pi^{i}_{;i} \approx 0
\end{align}
The canonical Hamiltonian generates the evolution equations
\begin{align}
\label{subsec:DBexamples:em:evoVectPot}
\dot{A}_{i} = \left[A_{i},H_{0}\right] & =  \pi_{i} - \phi_{,i}
\end{align}
for the vector potential and
\begin{align}
\label{subsec:DBexamples:em:evoElecField}
\dot{\vec{\pi}} = \left[\vec{\pi},H_{0}\right] & =  -\vec{\nabla}\times\vec{\nabla}\times\vec{A} 
\end{align}
for the canonically conjugate momenta. Expressing the canonical momenta $\vec{\pi}$ as a spatial $1-$form
\begin{align}
\label{subsec:DBexamples:em:momentaOneForm}
\tilde{\pi}^{(1)} \equiv \pi_{i}~dx^{i}
\end{align}
the canonical equations of motion become
\begin{align}
\label{subsec:DBexamples:em:canonicalEvolutionOneForms}
\dot{\mathcal{A}^{(1)}} = \pb{\mathcal{A}^{(1)}}{H_{0}} & =  \tilde{\pi}^{(1)} - \mathbf{d}\phi\\\nonumber
\dot{\tilde{\pi}}^{(1)} = \pb{\tilde{\pi}^{(1)}}{H_{0}} & = -\mathbf{d}^{\ast}\mathbf{d}\mathcal{A}^{(1)}
\end{align}
Note that the spatial Laplacian operator, $\Delta$, is defined as
\begin{align}
\label{subsec:DBexamples:em:spatialLaplacian}
\Delta \equiv \left(\mathbf{d}\mathbf{d}^{\ast} + \mathbf{d}^{\ast}\mathbf{d}\right) 
\end{align}
and since $\mathcal{A}^{(1)}$ is a spatial $1-$form, not a scalar, $\mathbf{d}^{\ast}\mathbf{d}\mathcal{A}^{(1)} \neq \Delta\mathcal{A}^{(1)}$. There is no canonical evolution equation for the scalar potential, $\phi$. From equation~(\ref{subsec:DBexamples:em:secondFCC}), the canonical evolution equation for the primary first class constraint $\pi^{0}\approx 0$ will vanish when the evolution remains on the constraint manifold.
%
\subsection{Second Class Constraints and the Dirac Bracket}
\label{subsubsec:em:sccAndDB}
In order to construct a Dirac bracket and fix the gauge, it is necessary to enlarge the constraint algebra, equation~(\ref{subsec:DBexamples:em:constraintAlgebra}), by introducing second class constraints. From section~(\ref{sec:GaugeHamiltonian}), the introduction of second class constraints will enable a coordinate system to be constructed in phase space throughout some neighborhood of the initial solution, allow a non-degenerate symplectic form, the Dirac bracket, residing on the space of physical observables to be constructed, and ensure that the Hamiltonian formulation be well-posed. Again, as noted in subsection~(\ref{subsubsec:em:canonEvolution}), in order to simplify the following analysis only flat Minkowski spacetime will be considered in this section. 

In the Lagrangian form of Electrodynamics a popular constraint to impose is the \emph{Lorentz gauge} choice 
\begin{align}
\label{subsec:DBexamples:em:lorentzgauge}
d^\ast A^{(1)} = A^{\alpha}_{;\alpha} = 0
\end{align}
Given that the field strength {2} form, $\mathbf{F}^{(2)}$, is defined by $dA^{(1)}$, so it seems natural to impose a constraint involving the codifferential operator. Since the Hamiltonian formulation must inherently split space and time, generating time evolution for field variables defined on a spatial slice at constant time, it will be slightly simpler to impose the \emph{Coulomb gauge} choice in vacuum
\begin{align}
\label{subsec:DBexamples:em:coulombGauge}
\phi & = 0\\\nonumber
\mathbf{d}^\ast \mathcal{A}^{(1)} = \vec{\nabla}\cdot \vec{A} & = 0
\end{align}
The non-vanishing commutation relations of the Coulomb gauge choice with the first class constraints are
\begin{align}
\label{subsec:DBexamples:em:coulombGaugeCommutations}
\pb{\phi\left(\mathbf{x}\right)}{\pi^{0}\left(\mathbf{x}^{\prime}\right)} & = \delta\left(\mathbf{x},\mathbf{x}^{\prime}\right)\\\nonumber
\pb{\vec{\nabla}\cdot\vec{A}\left(\mathbf{x}\right)}{\vec{\nabla}\cdot\vec{\pi}\left(\mathbf{x}^{\prime}\right)} & = -\nabla^{2}\delta\left(\mathbf{x},\mathbf{x}^{\prime}\right) = \Delta\delta\left(\mathbf{x},\mathbf{x}^{\prime}\right)
\end{align}
with $\delta\left(\mathbf{x},\mathbf{x}^{\prime}\right)$ denoting the Dirac delta function in three dimensions. The commutation relations of equation~(\ref{subsec:DBexamples:em:coulombGaugeCommutations}) are invertible throughout phase space, since the differential operator $\nabla^{2}$ can be inverted throughout Minkowski space. This shows that imposing the Coulomb gauge, equation~(\ref{subsec:DBexamples:em:coulombGauge}), as a set of second class constraints will completely fix the gauge freedom uniquely.

With the introduction of second class constraints the constraint algebra generated by the first class constraints has been expanded from two first class constraints, equation~(\ref{subsec:DBexamples:em:constraintAlgebra}), to four second class constraints
\begin{align}
\label{subsec:DBexamples:em:SCCcoulomb}
\pi^0 & = \phi = 0 \\\nonumber
\vec{\nabla}\cdot\vec{\pi} & = \vec{\nabla}\cdot\vec{A} = 0
\end{align}
satisfying the commutation relations derived in equation~(\ref{subsec:DBexamples:em:coulombGaugeCommutations}). 
The expanded constraint algebra defines the \define{second class constraint manifold}, which is a sub-manifold of the original first class constraint manifold. 
Denoting the collection of first and second class constraints as
\begin{align}
\label{subsec:DBexamples:em:constVec}
\mathcal{C}_{\Gamma} \equiv \left\{\pi^0,\vec{\nabla}\cdot\vec{\pi},\phi,\vec{\nabla}\cdot\vec{A}\right\}
\end{align}
the constraint commutation matrix can be expressed as
\begin{equation}
\label{subsec:DBexamples:em:scccmatrix1}
\pb{\mathcal{C}_\Gamma\left(\mathbf{x}\right)}{\mathcal{C}_\Psi\left(\mathbf{x}^\prime\right)} \equiv \left(\begin{array}{cccc}
0 & 0 & -1 & 0 \\
\\
0 & 0 & 0 & \nabla^{2} \\
\\
1 & 0 & 0 & 0 \\
\\
0 & -\nabla^{2} & 0 & 0\\
\end{array}\right)~\delta\left(\mathbf{x},\mathbf{x}^{\prime}\right)
\end{equation}
Since the constraints $\phi = \pi^0 = 0$ form a canonical conjugate pair in phase space, $\left(\phi,\pi^{0}\right)$, the phase space can be reduced by setting $\phi = \pi^{0} = 0$ in the canonical action and dropping the phase space coordinates $\phi$ and $\pi^{0}$ from consideration. This reduces the Poisson bracket to the phase space defined by the canonical pairs $\left(A_{i},\pi^{i}\right)$. After reducing the phase space, the set of all constraints, equation~(\ref{subsec:DBexamples:em:constVec}), reduces to
\begin{align}
\label{subsec:DBexamples:em:constVecRed}
\mathcal{C}_{A} = \left\{\vec{\nabla}\cdot\vec{\pi},\vec{\nabla}\cdot\vec{A}\right\}
\end{align}
With this simplification, the constraint commutation matrix of equation~(\ref{subsec:DBexamples:em:scccmatrix1}) reduces to the $2\times 2$ matrix 
\begin{equation}
\label{subsec:DBexamples:em:scccmatrix2}
\left[\mathcal{C}_{A}\left(\mathbf{x}\right),\mathcal{C}_{B}\left(\mathbf{x}^\prime\right)\right] \equiv \left(\begin{array}{cc}
0 & \nabla^{2} \\
\\
-\nabla^{2} & 0\\
\end{array}\right)~\delta\left(\mathbf{x},\mathbf{x}^{\prime}\right)
\end{equation}
The inverse operator of this constraint commutation matrix, satisfying equation~(\ref{subsec:diracbracket:fieldInverse}), is 
\begin{equation}
\label{subsec:DBexamples:em:sccInverse}
D^{AB}\left(\mathbf{x}^{\prime\prime},\mathbf{x}^\prime\right) \equiv \delta\left(\mathbf{x}^{\prime\prime},\mathbf{x}^\prime\right) \left(\begin{array}{cc}
0 & \frac{-1}{\nabla^{2}} \\
\\
\frac{1}{\nabla^{2}} & 0\\
\end{array}\right)
\end{equation}
The Dirac bracket then takes the form of equation~(\ref{subsec:diracbracket:diracBracketContinuum}) evaluated on the reduced set of canonical variables
\begin{align}
\label{subsec:DBexamples:em:diracBracket}
\begin{split}
\pb{F\left(\mathbf{x}\right)}{G\left(\mathbf{x}^\prime\right)}_{D}~\equiv~& 
	\pb{F\left(\mathbf{x}\right)}{G\left(\mathbf{x}^\prime \right)}\\
	&-\int d^3x^{\prime\prime\prime}~\int d^3x^{\prime\prime}~
	\left\{
		\pb{F\left(\mathbf{x}\right)}{\mathcal{C}_A\left(\mathbf{x}^{\prime\prime}\right)}
		D^{AB}\left(\mathbf{x}^{\prime\prime},\mathbf{x}^{\prime\prime\prime}\right)
		\pb{\mathcal{C}_{B}\left(\mathbf{x}^{\prime\prime\prime}\right)}{G\left(\mathbf{x}^\prime\right)}
	\right\}
\end{split}
\end{align}
Inserting the constraints, $\mathcal{C}_{A}$, equation~(\ref{subsec:DBexamples:em:constVecRed}), and the operator inverting the constraint commutation matrix, $D^{AB}$, equation~(\ref{subsec:DBexamples:em:sccInverse}), into the Dirac bracket, equation~(\ref{subsec:DBexamples:em:diracBracket}), the commutation relations amongst the canonical phase space variables $\left(A_{i},\pi^{i}\right)$ become
\begin{align}
\label{subsec:DBexamples:em:DBcommutation}
\pb{A_{i}\left(\mathbf{x}\right)}{\pi^{j}\left(\mathbf{x}^{\prime}\right)}_{D} & = \left[\delta^{j}_{i}-\frac{1}{2}\left(\nabla_{i}\nabla^{j} + \nabla^j\nabla_i\right)\left(\frac{1}{\nabla^2}\right)\right]\delta\left(\mathbf{x},\mathbf{x}^\prime\right)\\\nonumber
\left[A_i,A_j\right]_{D} & = \left[\pi^{i},\pi^{j}\right]_{D} = 0
\end{align}
manifestly satisfying 
\begin{align}
\label{subsec:DBexamples:em:DBcommutationNullSpace}
\pb{A_{i}}{\mathcal{C}_{A}}_{D} & = 0\\
\pb{\pi^{j}}{\mathcal{C}_{B}}_{D} &= 0
\end{align}
for all constraints present. 

Using the derived Dirac bracket, equation~(\ref{subsec:DBexamples:em:diracBracket}) along with the canonical Hamiltonian, equation~(\ref{subsec:DBexamples:em:canonHamiltonian}), the evolution equations for any phase space function $F$ restricted to the second class constraint manifold will be
\begin{align}
\label{subsec:DBexamples:em:DBevolution}
\begin{split}
\pb{F\left(\mathbf{x}\right)}{H_{0}}_{D}~\equiv~& 
	\pb{F\left(\mathbf{x}\right)}{H_{0}}\\
	&-\int d^3x^{\prime\prime}~\int d^3x^{\prime}~
	\left\{
		\pb{F\left(\mathbf{x}\right)}{\mathcal{C}_A\left(\mathbf{x}^{\prime}\right)}
		D^{AB}\left(\mathbf{x}^{\prime},\mathbf{x}^{\prime\prime}\right)
		\pb{\mathcal{C}_{B}\left(\mathbf{x}^{\prime\prime}\right)}{H_{0}}
	\right\}
\end{split}
\end{align}
From equation~(\ref{subsec:diracbracket:LagrangeMultipliersField}), the Dirac bracket fixes the form of the Lagrange multipliers of each of the constraints as
\begin{align}
\label{subsec:DBexamples:em:extHamUndMul}
\lambda^{A}\left(\mathbf{x}\right) = \int d^3x^\prime~\left\{
		D^{AB}\left(\mathbf{x},\mathbf{x}^{\prime}\right)
		\pb{\mathcal{C}_{B}\left(\mathbf{x}^{\prime}\right)}{H_{0}}~\right\}
\end{align} 
Because the original first class constraints commute with the canonical Hamiltonian, $H_{0}$, half of the Lagrange multipliers will weakly vanish on the original first class constraint manifold, allowing the gauge fixed extended Hamiltonian to weakly take the same form as the canonical extended Hamiltonian
\begin{align}
\label{subsec:DBexamples:em:extHamGF}
H_{E} & = H_{0} + \int d^{3}x~\int d^{3}x^{\prime}~\left\{
	\delta\left(\mathbf{x},\mathbf{x}^{\prime}\right)
	\lambda^{A}\left(\mathbf{x}\right)
	\mathcal{C}_{A}\left(\mathbf{x}^{\prime}\right)
	\right\}
\end{align}
On the reduced phase space with constraints
\begin{align}
\label{subsec:DBexamples:em:constraintVectorReduced}
\mathcal{C}_{0} & = \vec{\nabla}\cdot\vec{\pi} = 0\\\nonumber
\mathcal{C}_{1} & = \vec{\nabla}\cdot\vec{A} = 0
\end{align}
the Lagrange multipliers defined in equation~(\ref{subsec:DBexamples:em:extHamUndMul}) take the form
\begin{align}
\label{subsec:DBexamples:em:lagMultipliersGaugeFixed}
\lambda^{0}\left(\mathbf{x}\right) & =  \int d^{3}x^{\prime}~\left\{\frac{1}{\nabla^{2}\left(\mathbf{x},\mathbf{x}^{\prime}\right)}\pb{\mathcal{C}_{1}\left(\mathbf{x}^{\prime}\right)}{H_{0}}\right\} \\\nonumber
\lambda^{1}\left(\mathbf{x}\right) & = \int d^{3}x^{\prime}~\left\{\frac{1}{\nabla^{2}\left(\mathbf{x},\mathbf{x}^{\prime}\right)}\pb{\mathcal{C}_{0}\left(\mathbf{x}^{\prime}\right)}{H_{0}}\right\} 
\end{align} 
These expressions can be simplified using the canonical evolution equations, equation~(\ref{subsec:DBexamples:em:canonicalEvolutionOneForms}), along with the properties
\begin{align}
\label{subsec:DBexamples:em:laplacianCommute}
\mathbf{d}^{\ast}\Delta = \Delta\mathbf{d}^{\ast}
\end{align}
and
\begin{align}
\label{subsec:DBexamples:em:laplacianFlatCovariantSqrd}
\Delta = -\nabla^{2}
\end{align}
which is true for $\Delta$ acting on divergenceless $p$-form fields whenever the Riemann tensor on the spatial manifold vanishes, such as it does in Minkowski space. Using the definition for the canonical Hamiltonian, $H_0$, equation~(\ref{subsec:DBexamples:em:canonHamiltonian}), as a volume integral over the spatial domain to integrate $\pb{\mathcal{C}_{0}}{H_{0}}$ by parts and applying the integral operator $\frac{1}{\nabla^{2}}$, the equations of motion generated by $H_{0}$, equation~(\ref{subsec:DBexamples:em:canonicalEvolutionOneForms}), yield
\begin{align}
\label{subsec:DBexamples:em:lagMultipliersGaugeFixedSecondClass}
\lambda^{0} & = \frac{\mathcal{C}_{0}}{\nabla^{2}}\\\nonumber
\lambda^{1} & = \frac{1}{\nabla^{2}}\mathbf{d}^{\ast}\Delta\mathcal{A}^{(1)} = -\mathbf{d}^{\ast}\mathcal{A}^{(1)} = -\mathcal{C}_{1}
\end{align} 
Inserting equations~(\ref{subsec:DBexamples:em:lagMultipliersGaugeFixedSecondClass}) into the extended Hamiltonian, equation~(\ref{subsec:DBexamples:em:extHamGF}), reveals that in the Coulomb gauge all Lagrange multiplier terms in the extended Hamiltonian will be quadratic in the constraints.

Even though imposing a complete set of second class constraints will uniquely fix the Lagrange multipliers found in a general gauge theory, allowing each to be expressed in terms of the canonical variables, as in equations~(\ref{subsec:DBexamples:em:lagMultipliersGaugeFixed})~and~(\ref{subsec:DBexamples:em:lagMultipliersGaugeFixedSecondClass}), when varying the canonical action the Lagrange multipliers themselves must not be varied directly since the actual values assigned to each will be fixed by the variation itself.
 In the Coulomb gauge, the situation simplifies considerably from the general case. Since each term in $\lambda^{A}\mathcal{C}_{A}$ will be quadratic in a single scalar constraint, multiplying each term by $\frac{1}{2}$ will yield the same result as varying only $\mathcal{C}_{A}$ alone, leaving $\lambda^{A}$ to be determined by the dynamics. With this simplification, the gauge fixed extended Hamiltonian becomes
\begin{align}
\label{subsec:DBexamples:em:gfeHamRed}
H_{E} = \frac{1}{2}\int \left\{\tilde{\pi}^{(1)}\wedge\ast_S\tilde{\pi}^{(1)} + \frac{1}{\nabla^{2}}\mathbf{d}^{\ast}\tilde{\pi}^{(1)}\wedge \ast_{S} \mathbf{d}^{\ast}\tilde{\pi}^{(1)}  + \mathbf{d}\mathcal{A}^{(1)}\wedge\ast_S\mathbf{d}\mathcal{A}^{(1)} - \mathbf{d}^{\ast}\mathcal{A}^{(1)}\wedge \ast_{S} \mathbf{d}^{\ast}\mathcal{A}^{(1)}  \right\} 
\end{align}
generating equations of motion
\begin{align}
\label{subsec:DBexamples:em:gfeHamEOM}
\dot{A}_{i} \equiv \pb{A_{i}}{H_{E}} &= \left[\delta_{ij} - \frac{1}{\nabla^{2}}\nabla_{(i}\nabla_{j)}\right]\pi^{j} = \pi_{i} - \frac{1}{\nabla^{2}}\nabla_{i}\mathcal{C}_{0}\\\nonumber
\dot{\pi}^{i} \equiv\pb{\pi^{i}}{H_{E}} &= -\left[\delta^{ij} - \nabla^{(i}\nabla^{j)}\frac{1}{\nabla^{2}}\right]\Delta A_{j} = \nabla^{2} A^{i} - \nabla^{i}\mathcal{C}_{1}
\end{align}
which are in agreement with the equations of motion generated by canonical Hamiltonian on the second class constraint manifold through the Dirac bracket
\begin{align}
\label{subsec:DBexamples:em:gfeHamEOMEquiv}
\pb{A_{i}}{H_{E}} &= \pb{A_{i}}{H_{0}}_{D} \\\nonumber
\pb{\pi^{i}}{H_{E}} &=\pb{\pi^{i}}{H_{0}}_{D} 
\end{align}
validating the form of the gauge fixed extended Hamiltonian, equation~(\ref{subsec:DBexamples:em:gfeHamRed}). On the constraint manifold, $\mathcal{C}_{A} = 0$, the evolution of the spatial vector potential, $\dot{A}_{i} \approx \pi_{i}$, can be inserted into the equation of motion for $\pi^{i}$, yielding $\ddot{A}_{i}$. From the equations of motion generated by the Dirac bracket, equation~(\ref{subsec:DBexamples:em:gfeHamEOM}), the result
is
\begin{align}
\label{subsec:DBexamples:em:waveEquationCG}
\left(\frac{\partial}{\partial t}\right)^2A_{i}dx^i - \nabla^2A_idx^i \equiv \square A_i \approx 0
\end{align} 
revealing that, in vacuum, all physical fields propagate as waves traveling at the single constant speed $\pm 1$. Since the Dirac bracket removes one degree of freedom from the three degrees of freedom corresponding to the six canonical phase space coordinates $\left(A_{i},\pi^{i}\right)$, once the Dirac bracket has been constructed, only two degrees of freedom remain. As will be seen in section~(\ref{subsubsec:em:shGFsys}), perturbations to the two independent physical degrees of freedom will propagate as waves traveling at some constant speed $\pm c \in \mathbb{R}$, bounding the domain of dependence of perturbations as an initial solution is evolved. As expected, these two remaining degrees of freedom are physical observables which correspond to the two independent helicity states of electromagnetic radiation, \cite{weinberg1}. 
%
\subsection{Hyperbolicity}
\label{subsubsec:em:shGFsys}
This subsection assumes a knowledge of pseudo-differential methods used to convert second order variable coefficient partial differential systems into first order constant coefficient pseudo-differential systems. Concise introductions to the pseudo-differential methods used here as well as proofs of strongly hyperbolic formulations yielding well-posed problems can be found elsewhere \cite{nor},\cite{brown3}.\cite{reula1}.
 
Consider the canonical Hamiltonian, $H_{0}$, on the reduced phase space in which the canonical pair $\left(\phi,\pi^{0}\right)$ have been dropped, and the resulting canonical action on the reduced phase space generated by setting $\phi=\pi^{0} = 0$ throughout the canonical action defined on the initial phase space which includes the canonical pair $\left(\phi,\pi^{0}\right)$. The reduced phase space, defined by the three canonical coordinate pairs $\left(A_{i},\pi^{i}\right)$, will be used throughout the remainder of this subsection. 

\subsubsection{Canonical Formulation}
\label{subsubsec:em:shGFsys:canonical}
Varying the extremal path generated by the canonical Hamiltonian, equation~(\ref{subsec:DBexamples:em:canonicalEvolutionOneForms}), yields
\begin{align}
\label{subsec:DBexamples:em:secondVariation}
\dot{\delta A}_{i} &= \delta \pi_{i}\\\nonumber
\dot{\delta \pi_{i}} &= \nabla^{2}\delta A_{i}
\end{align}
In order to convert to a pseudo-differential system, define 
\begin{align}
\label{subsec:DBexamples:em:pseudoDiffNorm}
\left|k\right| \equiv \sqrt{k_{i}k_{j}\delta^{ij}} = \sqrt{k_{i}k^{i}}
\end{align}
and insert the variation
\begin{align}
\label{subsec:DBexamples:em:pseudoDiffVar}
\delta A_{i} & = -i\frac{\hat{A}_{i}}{\left|k\right|}e^{i\left(\omega t + k_{i}x^{i}\right)}\\\nonumber
\delta \pi^{i} & = \hat{\pi}^{i}e^{i\left(\omega t + k_{i}x^{i}\right)}
\end{align}
into equation~(\ref{subsec:DBexamples:em:secondVariation}). Setting $\hat{A}_{i},\hat{\pi}^{j}$ equal to constants and dropping all terms lower than first order in $|k|$ converts the second order canonical formulation, given by equation~(\ref{subsec:DBexamples:em:canonicalEvolutionOneForms}), into a first order pseudo-differential system
\begin{align}
\label{subsec:DBexamples:em:pseudoDiffSecondVariation}
\omega \hat{A}_{i} &= \left|k\right| \hat{\pi}_{i}\\\nonumber
\omega \hat{\pi}_{i} &= \left|k\right| \hat{A}_{i}
\end{align}
In this pseudo-differential formulation, the first class constraint $\vec{\nabla}\cdot\vec\pi \approx 0$, which imposes a relation between phase space coordinates, becomes
\begin{align}
\label{subsec:DBexamples:em:pseudoDiffFCC}
ik_{i}\hat{\pi}^{i} \equiv i\left|k\right|n_{i}\hat{\pi}^{i} = 0
\end{align}
imposing a relation between the constants $\hat{\pi}^{i}$ and the permissible spatial directions of propagation, given by the spatial unit vector
\begin{align}
\label{subsec:DBexamples:em:pseudoDiffSpatialPropDir}
\vec{n} \equiv \frac{\vec{k}}{\left|k\right|} = \frac{k^{i}}{\left|k\right|}\p_{i}
\end{align}
satisfying $\left|n\right| = \sqrt{n_{i}n^{i}} = 1$. The form of equation~(\ref{subsec:DBexamples:em:pseudoDiffFCC}) suggests that the variational constants, $\hat{A}_{i},\hat{\pi}^{j}$, be projected onto terms which are tangent to $\vec{n}$, defining the \define{longitudinal} components
\begin{align}
\label{subsec:DBexamples:em:pseudoDiffLongitudinal}
\hat{A}_{L} \equiv \hat{A}_{i}n^{i}\\\nonumber
\hat{\pi}_{L} \equiv \hat{\pi}_{i}n^{i}
\end{align}
and terms which are orthogonal to $\vec{n}$, defining the \define{transverse} components
\begin{align}
\label{subsec:DBexamples:em:pseudoDiffTransverse}
\hat{A}^{T}_{i} \equiv \hat{A}_{i} - n_{i}\hat{A}_{L}\\\nonumber
\hat{\pi}^{T}_{i} \equiv \hat{\pi}_{i} - n_{i}\hat{\pi}_{L}
\end{align}
Defining the normalized frequency as
\begin{align}
\label{subsec:DBexamples:em:pseudoDiffNormFreq}
\kappa \equiv \frac{\omega}{\left|k\right|}
\end{align}
and inserting the first class constraint condition $\hat{\pi}_{L} = 0$, the pseduo-differential system of equation~(\ref{subsec:DBexamples:em:pseudoDiffSecondVariation}) has the longitudinal sub-block 
\begin{align}
\label{subsec:DBexamples:em:pseudoDiffSecondVariationLongitudinal}
\kappa_{L} \hat{A}_{L} &= \hat{\pi}_{L}\\\nonumber
\kappa_{L} \hat{\pi}_{L} &= 0
\end{align}
which has the single eigenvalue $\kappa_{L} = 0$ of multiplicity two, and a single eigenvector, $\hat{\pi}_{L}$. This can be easily seen by expressing the evolution equations as a pseudodifferential system with principal symbol corresponding to the evolution generated by the Hamiltonian acting on the longitudinal subblock
\begin{align}
\label{subsec:DBexamples:em:pseudoDiffSecondVariationLongitudinalMatrixForm}
\kappa \begin{bmatrix}
\hat{A}_{L} \\
\hat{\pi}_{L}
\end{bmatrix} 
& \equiv \hat{\mathbf{H}} \begin{bmatrix}
\hat{A}_{L} \\
\hat{\pi}_{L}
\end{bmatrix} 
=
\begin{bmatrix}
0 & 1 \\
0 & 0 \\
\end{bmatrix}~
\begin{bmatrix}
\hat{A}_{L} \\
\hat{\pi}_{L}
\end{bmatrix} 
= 
\begin{bmatrix}
\hat{\pi}_{L} \\
0
\end{bmatrix}
\end{align}
with $\hat{\mathbf{H}}$ denoting the principal symbol, which takes the above form for the longitudinal subblock of phase space $\left(\hat{A}_{L},\hat{\pi}^{L}\right)$. Using equation~(\ref{subsec:DBexamples:em:pseudoDiffSecondVariationLongitudinalMatrixForm}), the eigenvalues of the pseudodifferential evolution operator on the longitudinal subblock solve the equation $\kappa_{L}^{2} = 0$, hence the evolution of the longitudinal subblock yields an eigenvalue of zero with multiplicity of two. Although the eigenvalue has multiplicity of two, there is only one eigenvector for the longitudinal subblock, namely the vector
\begin{align}
\label{subsec:DBexamples:em:pseudoDiffLongitudinalCanonicalEigenvec}
\begin{bmatrix}
\hat{\pi}_{L} \\
0
\end{bmatrix}
\end{align}
Since there is a single eigenvalue of multiplicity two for the longitudinal subblock yet only a single eigenvector, the pseduodifferential evolution generated by the canonical Hamiltonian cannot possess a complete set of eigenvectors.

The canonical pair along the longitudinal satisfies $[A_L,\pi^L] = 1$ and zero for all other phase space coordinates, $\left(A^{T}_{i},\pi^{i}_{T}\right)$. The pseudodifferential form of the evolution equations generated by the Hamiltonian yields a pair with one explicit equation yielding $\kappa_{L} = 0$, namely, the evolution for $\hat{\pi}^L$ which is zero in the pseudodifferential space because $\pi^{i}_{;i} \approx 0$ commutes with the canonical Hamiltonian, and one equation that implicitly sets $\kappa_{L} = 0$, namely the equation $\kappa_{L} A_L = \hat{\pi}^{L}$ in the pseudodifferential reduction. Canonical pairs propagate at the same speed since the commutation relations of the Poisson bracket commute strongly with the Hamiltonian.The fact that these two variables form a canonical pair that commutes with all other variables means that the principal symbol of the Hamiltonian evolution must have an eigenvalue with multiplicity two. The problem is, without an explicit constraint to fix the evolution of $A_L$, there is only one eigenvector of the principal symbol in the longitudinal subblock. Stated differently, the Hamiltonian commutes with the first class constraint $\pi^{i}_{;i} \approx 0$ and thereby limits the perturbations $\delta \pi^{i}$ which are permissible under the evolution. There is no such canonical constraint on the longitudinal field $A^{L}$ which could evolve differently using an extended Hamiltonian, $H_E$, that has an arbitrary, coordinate independent, Lagrange multiplier for $\pi^{i}_{;i} \approx 0$, as discussed in section~(\ref{subsec:constraints}). This makes it impossible for the system to be strongly hyperbolic, precisely because the gauge freedom is not fixed, as shown in section~(\ref{subsec:numerical:hyperbolicity}). By creating a second class constraint which controls $\hat{A}_{L}$, the principal symbol for the longitudinal subblock decouples the canonical pair, since the Dirac bracket will project out all non-physical terms, yielding an principal symbol with two distinct eigenvectors for the longitudinal subblock, both with eigenvalue equal to zero.

Denoting the two independent transverse components with the index $a$, where $a = \left\{-1,1\right\}$, the transverse sub-block of the pseudo-differential system, equation~(\ref{subsec:DBexamples:em:pseudoDiffSecondVariation}), is
\begin{align}
\label{subsec:DBexamples:em:pseudoDiffSecondVariationTransverse}
\kappa \hat{A}^{T}_{a} &= \hat{\pi}^{T}_{a}\\\nonumber
\kappa \hat{\pi}^{T}_{a} &= \hat{A}^{T}_{a}
\end{align}
and has four linearly independent eigenvectors
\begin{align}
\label{subsec:DBexamples:em:pseudoDiffSecondVariationTransverseEigenvectors}
\hat{z}^{\pm}_{a} \equiv \hat{A}^{T}_{a} \pm \hat{\pi}^{T}_{a} 
\end{align}
with eigenvalues $\kappa^{\pm}_{a} = \pm 1$. 

Since the longitudinal sub-block does not have a complete set of eigenvectors, the canonical Hamiltonian will generate a weakly hyperbolic system.

\subsubsection{Gauge Fixed Formulation}
\label{subsubsec:em:shGFsys:gaugeFixed}
Varying the extremal path generated by the gauge fixed extended Hamiltonian, equation~(\ref{subsec:DBexamples:em:gfeHamEOM}), yielding
\begin{align}
\label{subsec:DBexamples:em:GF:secondVariation}
\dot{\delta{A}}_{i} &=  \delta \pi_{i} - \frac{1}{\nabla^{2}}\nabla_{i}\nabla^{j}\delta\pi_{j}\\\nonumber
\dot{\delta \pi}_{i} &= \nabla^{2} \delta A_{i} - \nabla^{i}\nabla^{j}\delta A_{j}
\end{align}
Following the same procedure implemented for the canonical formulation, to convert from a second order differential system to a first order pseduo-differential system, insert the variation defined in equation~(\ref{subsec:DBexamples:em:pseudoDiffVar}) into equation~(\ref{subsec:DBexamples:em:GF:secondVariation}), again setting $\hat{A}_{i},\hat{\pi}^{j}$ equal to constants and dropping all terms lower than first order in $|k|$. Projecting once again onto the transverse components yields
\begin{align}
\label{subsec:DBexamples:em:pseudoDiffSecondVariationTransverse}
\kappa \hat{A}^{T}_{a} &= \hat{\pi}^{T}_{a}\\\nonumber
\kappa \hat{\pi}^{T}_{a} &= \hat{A}^{T}_{a}
\end{align}
revealing that the Coulomb gauge has not disturbed the transverse sub-block, whence the eigenvectors and eigenvalues of the transverse sub-block will be the same as for the canonical formulation. The pseudo-differential system derived from the evolution equations generated by the gauge fixed extended Hamiltonian has the longitudinal sub-block 
\begin{align}
\label{subsec:DBexamples:em:pseudoDiffSecondVariationLongitudinal}
\kappa \hat{A}_{L} &= 0\\\nonumber
\kappa \hat{\pi}_{L} &= 0
\end{align}
which has the two linearly independent eigenvectors $\hat{A}_{L}$ and $\hat{\pi}_{L}$, each with eigenvalue $\kappa = 0$.  The longitudinal sub-block now contains a complete set of eigenvectors, whence the gauge fixed extended Hamiltonian yields a strongly hyperbolic, and therefore well-posed, system. 

Despite the presence of zero eigenvalues, the system is strongly hyperbolic due to the fact that the gauge fixed extended Hamiltonian yields a principal symbol corresponding to a pseduo-differential evolution operator, $\hat{H}$, on the Sobolev space of variations, defined in equation~(\ref{subsec:DBexamples:em:pseudoDiffSecondVariation}), having \define{index} equal to zero, meaning that the dimension of the null space of the operator is equal to the dimension of the null space of the formal adjoint of the operator in the space. 

In practical terms, in order for an operator to have an index of zero on a particular space, the operator must be in some sense \define{self-adjoint} on that space, hence operators having an index equal to zero have minima which are guaranteed to correspond to true local minima in the space as opposed to saddle points, which would require the dimension of the null space of the evolution operator $\hat{H}$ to differ from that of the adjoint operator $\hat{H}^{\dagger}$. Since the symplectic form on canonical phase space remains constant, the evolution in the canonical phase space generated by the gauge fixed extended Hamiltonian through the canonical Poisson bracket is guaranteed to be an anti-symmetric operator, hence for the analytic index of $\hat{H}$ to be zero, the null space of $\hat{H}$ must be the canonical conjugate of the null space of $\hat{H}^{\dagger}$. As a consequence, if the operator $\hat{H}$ possesses a complete set of eigenvectors, the operator can be diagonalized with values along the diagonal coming in pairs of real roots of quadratic equations, hence a complete set of eigenvectors guarantees that the null space of the operator $\hat{H}$ will be isomorphic, in some neighborhood of a given solution satisfying the Euler-Lagrange equations, to the null space of the adjoint, $\hat{H}^{\dagger}$, thereby ensuring that the index remains zero as the system evolves. 

In essence, what this means for the well-posedness of a system is that zero eigenvalues correspond to a null space which remains fixed as the system evolves so long as the operator $\hat{H}$ possesses a complete set of eigenvectors, in which case there will be a well-defined domain of dependency for variations of initial solutions which remains exponentially bounded as the system evolves, namely variations which project onto the space of physical solutions propagate with a finite speed while variations which project onto the null space remain constant as the system evolves, thereby ensuring that a finite difference between two initial solutions results in a difference which is exponentially bounded as the two system evolves. The important point is that, no matter whether the evolution operator possesses zero eigenvalues, as long as the eigenvectors of the evolution operator are complete with eigenvalues corresponding to pairs of real roots of quadratic equations, the formulation will be strongly hyperbolic and therefore well-posed.
%
\subsection{Initial Data}
\label{subsubsec:em:initialDataError}
From section~(\ref{subsec:diracbracket}), all gauge freedom in a Hamiltonian system is completely fixed by the introduction of second class constraints that form an invertible Dirac bracket amongst the full set of constraints. The physical degrees of freedom belong to the space conjugate to the manifold defined by the full set of $2N$ constraints. In this space, there are $N$ first class constraints and $N$ gauge fixing second class constraints. The $N$ first class constraint solutions form a manifold of all equivalent physical solutions. The $N$ second class constraints define a unique coordinate choice on the first class constraint manifold of equivalent physical solutions. Second class constraints contain no physical degrees of freedom and serve only to distinguish a unique physical solution out of all equivalent physical solutions. There is no physical difference between two choices for second class constraints, so long as they form an invertible set with the first class constraints.

For the constrained Electrodynamics example considered here, the primary first class constraint $\pi^0 \approx 0$ and conjugate momenta, $A_0$ have been dropped entirely from the canonical coordinate space. This leaves three pairs of canonical variables in the space, constrained by one canonical first class constraint, $\pi^{i}_{;i} \approx 0$, and one second class constraint which fixes the gradient of the spatial vector $A^{i}_{;i}$. Once all constraints are imposed, only two physical degrees of freedom remain. Initial data requires uniquely fixing the physical state.

Consider then a physical packet of Electrodynamic radiation at time $t=0$, spatially bound in a region of a flat spatial manifold $\Sigma$. No boundary conditions are necessary for this problem and there are no matter terms to interact with. Physical solutions of this form belong to the space of infinitely differentiable functions bound within some finite region of space, denoted $\mathbb{C}^\infty\left(\Sigma\right)$. Any initial data, and any evolved data on other spatial slices of constant coordinate time, will belong to this function space.

To construct initial data corresponding to a physical solution, start by supplying an initial value for all the fields that are to be evolved; the three arbitrary functions for the vector potential
\begin{align}
\label{subsec:DBexamples:em:barredA}
\bar{A}_{i}\left(t=0,x,y,z\right) \in \mathbb{C}^\infty\left(\Sigma\right)
\end{align}
and the three conjugate momenta
\begin{align}
\label{subsec:DBexamples:em:barredPi}
\bar{\pi}^{i}\left(t=0,x,y,z\right) \in \mathbb{C}^\infty\left(\Sigma\right)
\end{align}
The barred values, $\bar{A}_{i},\bar{\pi}^{j}$, were chosen from the full space of permissible functions without regard to any set of constraint equations, and are to be interpreted as just numerical values, distinct from the canonical variables $\left(A_{i},\pi^{i}\right)$ which will define a solution. Since the barred values were not restricted within the function space $\mathbb{C}^\infty\left(\Sigma\right)$, there is no expectation that the set $\bar{A}_{i},\bar{\pi}^{j}$ will  provide a solution to either $A^{i}_{;i} = 0$ or $\pi^{i}_{;i} \approx 0$, so will be off-shell. 

Recall that the second class constraint is arbitrary, so replacing the constraint $A^{i}_{;i} = 0$ with the constraint 
\begin{align}
\label{subsec:DBexamples:em:barredGaugeConstraint}
A^{i}_{;i} - \bar{A}^{i}_{;i}\left(t=0,x,y,z\right) = 0
\end{align} 
is entirely valid, and corresponds to a different coordinate choice on the manifold of physically equivalent solutions. From section~(\ref{subsec:stability:gaugeIntegrable}), the flow generated by the second class constraints is, primarily, orthogonal to the gauge constraint surface. Consequently, the second class constraint can be used to generate transformations between manifolds of equivalent physical solutions to arrive at $\pi^{i}_{;i} = 0$. For a well-posed problem, along with unambiguous initial data, the evolution generated by the extended Hamiltonian for all fields is causal and is thus completely fixed on any future or past time slice. This means any violations of the constraints will be preserved by the evolution, so the barred values $\bar{A}_{i}, \bar{\pi}^{j}$, which provide a state in phase space off-shell, can be used to construct an error correction generating function, equation~(\ref{subsec:numerical:numErrTransGenerating}), that maintains the form of the constraints preserved by the Hamiltonian.

Concretely then, to construct initial data which provide a physical solution solving all constraints, take an initial set of barred values and calculate the coefficients in equation~(\ref{subsec:numerical:multipliersNumericalErrorFixed}). Assume that the second class constraint is provided by equation~(\ref{subsec:DBexamples:em:barredGaugeConstraint}). Since $\pi^{i}$ commutes strongly with the constraint $\pi^{j}_{;j} \approx 0$, and $A^{i}$ is assumed to already satisfy a chosen second class constraint, a set of initial data for a physical state on the constraint manifold can be found updating only the momenta. The error correction generating function, equation~(\ref{subsec:numerical:numErrTransGenerating}), has constraint coefficients which are calculated using the selected barred values, equation~(\ref{subsec:numerical:multipliersNumericalErrorFixed}), yielding
\begin{align}
\label{subsec:DBexamples:em:lambdaUpdates}
\bar{\Lambda}^{A} = \mathcal{D}_{\pi A}\bar{\pi}^{j}_{;j}
\end{align}
Substituting in the results from equation~(\ref{subsec:DBexamples:em:sccInverse}), the coefficients are
\begin{align}
\label{subsec:DBexamples:em:lambdaExplicitUpdates}
\bar{\Lambda}^{A} = \left(\frac{1}{\nabla^2}\bar{\pi}^{k}_{;k}\right)
\end{align} 
where the inverse Laplacian operator, $\frac{1}{\nabla^2}$, on a flat spatial manifold with no boundary conditions, is an integral equation with kernel given by the familiar Green's function
\begin{align}
\label{subsec:DBexamples:em:greenFunction}
G = -4\pi\frac{1}{\left|x-y\right|}
\end{align}
All other terms in the error correction generating function vanish when equation~(\ref{subsec:DBexamples:em:barredGaugeConstraint}) is used for the second class constraint definition. The error correction transformation for $\pi^{i}$, calculated from the barred off-shell position of $\bar{\pi}^{i}\left(x\right)$ is 
\begin{align}
\label{subsec:DBexamples:em:updatePi}
-\delta^{0}\pi^{i} = -\delta^{ij}\nabla_{j}\int dy G\left(x,y\right) \bar{\pi}^{k}\left(y\right)_{;k} = 4\pi \int dy \frac{\nabla^{i}\nabla_k\bar{\pi}^{k}\left(y\right)}{\left|x-y\right|}
\end{align} 
Since there are only two constraints, the anti-symmetry of the Dirac bracket ensures that the error correction generating function, $E_{N}$, can be used to independently update the two constraints. Hence, if a different choice of second class constraint is desired, updating the fields to satisfy the different constraint choice will leave the other constraint, and dependent fields, unaffected. In particular, for the constraint choice $A^{i}_{;i} = 0$, equation~(\ref{subsec:DBexamples:em:coulombGauge}), the error correction transformation for $A^{i}$ is the exact same form as for $\pi^i$, equation~(\ref{subsec:DBexamples:em:updatePi}), with $\bar{\pi}^{i}$ terms replaced with $\bar{A}^{i}$ values
\begin{align}
\label{subsec:DBexamples:em:updateVector}
-\delta^{0} A_{i} = -\delta^{j}_{i}\nabla_{j}\int dy G\left(x,y\right) \bar{A}^{k}\left(y\right)_{;k} = 4\pi \int dy \frac{\nabla^{i}\nabla_k\bar{A}^{k}\left(y\right)}{\left|x-y\right|}
\end{align} 
up to a sign, dependent upon the sign of the constraints imposed. The initial data on the constrained, physical solution manifold, are
\begin{align}
\label{subsec:DBexamples:em:initialDataFinal}
A_i & = \bar{A}_i -\delta^{0} A_{i}\\\nonumber
\pi^j & = \bar{\pi}^j - \delta^{0}\pi^{j}
\end{align} 
These equations define initial data for a physical solution. The complete set of constraints forms a consistent transformation method between, and within, a space of physically equivalent solutions. The non-locality of the Green's function, equation~(\ref{subsec:DBexamples:em:greenFunction}), ensures that the solution supplied by equation~(\ref{subsec:DBexamples:em:initialDataFinal}) is global and that the same method extends to more complicated arrangements which include matter and non-trivial boundary conditions. The constraint commutation relations in the Coulomb gauge are independent of the field variables, so only first order terms are present for both $E_{N}$ and $H_{N}$, so that $H_{N} = H_{E}$ strongly. This process of error removal to construct an initial physical solution can be applied on any constant time slice as the system evolves and numerical error enters.
%
\subsection{Synopsis of Electrodynamics as a Gauge System}
\label{subsec:em:synopsis}
Expressing Electrodynamics as a gauge theory, the canonical formulation is seen to yield a system of equations which are only weakly hyperbolic. Fixing the gauge freedom through imposing the Coulomb gauge as a set of second class constraints, defined by the equations $\phi = 0$ and $\nabla_{i}A^{i} = 0$, yields a system of evolution equations which are strongly hyperbolic, in keeping with section~(\ref{subsec:numerical:hyperbolicity}). Though the Coulomb gauge was chosen for simplicity, more elaborate choices for second class constraints will yield the same result, so long as the second class constraints results in all gauge freedom being uniquely fixed. This is ensured by the construction of the Dirac bracket over the first class constraint manifold, accomplished by removing all gauge freedom, which results in a formulation evolving only physical quantities. Since perturbations in physical quantities cannot propagate faster than the speed of light, when the Dirac bracket is used to define evolution equations only, the system will evolve perturbations affecting physical quantities, whence perturbations will propagate at a finite speed as the system evolves. This means that perturbations of a solution will have a finite domain of dependence, bounded by the future light cone, and hence the magnitude of such perturbation cannot grow arbitrarily in time. In this way, removing gauge freedom from the theory ensures that the resulting gauge fixed system exhibits stability during evolution. 


\section{Conclusion and Future Work}
\label{sec:conclusion}
As shown in section~(\ref{subsec:stability:gaugeIntegrable}), gauge theories cannot be numerically stable unless all gauge freedom is sufficiently fixed in system of evolution equations. This result is new, but intuitively obvious: in gauge theories, the gauge freedom could reasonably propagate at any non-physical speed, so it cannot be expected that perturbations in an initial solution can be bounded as the system evolves without the inclusion of additional constraints. As we showed in section~(\ref{subsec:stability:removeErr}), the Dirac formalism can also be used to provide a consistent method to remove the error directly and to control the propagation of error off-shell.

It is interesting to note the close relationship between quantization and computational stability. Canonical quantization alleviates any necessity, or even ability, to define both the position and momenta of a system with infinite precession simultaneously. As a result, the quantized theory must yield dynamics effectively bounding the magnitude of effects due to the unavoidable ambiguity in the description of a physical system at some initial time on permissible observable physical states at later times, otherwise the theory would not be able to make meaningful predictions and would therefore be of little use. Dirac introduced the formalism initially to provide methods for quantization of classical systems but, as we have shown in section~(\ref{sec:stability}), these methods also provide a consistent framework to handle numerical evolution of classical gauge theories and the identification of error. Intuitively, it makes complete sense that a formalism for quantization would necessarily control small perturbations, so would be a natural framework for understanding stability of classical systems. After it was shown that the General Relativity is a non-renormalizable field theory, and so Dirac's canonical quantization procedure would not work, the Dirac bracket formalism for constrained systems fell into disfavor and was not so actively pursued. As a result, within G.R. much of the subsequent research into numerical computation relied on more geometric motivations, retreating from a strict canonical formalism, which at times has obfuscated the availability of machinery and techniques used elsewhere for Hamiltonian systems.

The formalism presented here for constrained Hamiltonian dynamics allows for a consistent treatment of any physical theory with gauge freedom, regardless of the gauge group. This is incredibly important when considering General Relativity, since the gauge group of all diffeomorphisms of the solution space does not easily lend itself to the same machinery as gauge theories with compact gauge groups. The results in this work, showing that gauge freedom must be fixed for a resulting system to be strongly hyperbolic, extends to G.R. when treated in a constrained Hamiltonian formalism. As discussed briefly in section~(\ref{subsubsec:hamiltonian_constraints}), G.R. requires more consideration because its canonical Hamiltonian locally vanishes everywhere. Our subsequent work will address these considerations by extending the formalism presented here to generally covariant theories, invariant under re-parameterizations of time, with G.R. as the focus. Specifically, initial data, error correction, and a generalized method for deriving stable evolution equations in G.R. will be recovered.

\section*{Acknowledgements} 
We thank Claudio Bunster, Mark Henneaux, David Brown, and Phil Morrison for their interesting early discussions.


\bibliographystyle{apsrev}
\bibliography{gaugeFixing}

\begin{thebibliography}{26}
\expandafter\ifx\csname natexlab\endcsname\relax\def\natexlab#1{#1}\fi
\expandafter\ifx\csname bibnamefont\endcsname\relax
  \def\bibnamefont#1{#1}\fi
\expandafter\ifx\csname bibfnamefont\endcsname\relax
  \def\bibfnamefont#1{#1}\fi
\expandafter\ifx\csname citenamefont\endcsname\relax
  \def\citenamefont#1{#1}\fi
\expandafter\ifx\csname url\endcsname\relax
  \def\url#1{\texttt{#1}}\fi
\expandafter\ifx\csname urlprefix\endcsname\relax\def\urlprefix{URL }\fi
\providecommand{\bibinfo}[2]{#2}
\providecommand{\eprint}[2][]{\url{#2}}

\bibitem[{\citenamefont{{Dirac}}(1958{\natexlab{a}})}]{dirac1}
\bibinfo{author}{\bibfnamefont{P.~A.~M.} \bibnamefont{{Dirac}}},
  \bibinfo{journal}{Royal Society of London Proceedings Series A}
  \textbf{\bibinfo{volume}{246}}, \bibinfo{pages}{333}
  (\bibinfo{year}{1958}{\natexlab{a}}).

\bibitem[{\citenamefont{Arnowitt et~al.}(1962)\citenamefont{Arnowitt, Deser,
  and Misner}}]{adm0}
\bibinfo{author}{\bibfnamefont{R.}~\bibnamefont{Arnowitt}},
  \bibinfo{author}{\bibfnamefont{S.}~\bibnamefont{Deser}}, \bibnamefont{and}
  \bibinfo{author}{\bibfnamefont{C.}~\bibnamefont{Misner}}, in
  \emph{\bibinfo{booktitle}{Gravitation: An Introduction to Current Research}},
  edited by \bibinfo{editor}{\bibfnamefont{L.}~\bibnamefont{Witten}}
  (\bibinfo{publisher}{Wiley}, \bibinfo{address}{New York; London},
  \bibinfo{year}{1962}), pp. \bibinfo{pages}{227--265}.

\bibitem[{\citenamefont{Hahn and Lindquist}(1964)}]{hahn_lindquist}
\bibinfo{author}{\bibfnamefont{S.~G.} \bibnamefont{Hahn}} \bibnamefont{and}
  \bibinfo{author}{\bibfnamefont{R.~W.} \bibnamefont{Lindquist}},
  \bibinfo{journal}{Annals of Physics} \textbf{\bibinfo{volume}{29}},
  \bibinfo{pages}{304 } (\bibinfo{year}{1964}), ISSN \bibinfo{issn}{0003-4916}.

\bibitem[{\citenamefont{Abbott et~al.}(2016)\citenamefont{Abbott, Abbott,
  Abbott, Abernathy, Acernese, Ackley, Adams, Adams, Addesso, Adhikari
  et~al.}}]{grligo}
\bibinfo{author}{\bibfnamefont{B.~P.} \bibnamefont{Abbott}},
  \bibinfo{author}{\bibfnamefont{R.}~\bibnamefont{Abbott}},
  \bibinfo{author}{\bibfnamefont{T.~D.} \bibnamefont{Abbott}},
  \bibinfo{author}{\bibfnamefont{M.~R.} \bibnamefont{Abernathy}},
  \bibinfo{author}{\bibfnamefont{F.}~\bibnamefont{Acernese}},
  \bibinfo{author}{\bibfnamefont{K.}~\bibnamefont{Ackley}},
  \bibinfo{author}{\bibfnamefont{C.}~\bibnamefont{Adams}},
  \bibinfo{author}{\bibfnamefont{T.}~\bibnamefont{Adams}},
  \bibinfo{author}{\bibfnamefont{P.}~\bibnamefont{Addesso}},
  \bibinfo{author}{\bibfnamefont{R.~X.} \bibnamefont{Adhikari}},
  \bibnamefont{et~al.} (\bibinfo{collaboration}{LIGO Scientific Collaboration
  and Virgo Collaboration}), \bibinfo{journal}{Phys. Rev. Lett.}
  \textbf{\bibinfo{volume}{116}}, \bibinfo{pages}{061102}
  (\bibinfo{year}{2016}).

\bibitem[{\citenamefont{{Baumgarte} and {Shapiro}}(1999)}]{bssn2}
\bibinfo{author}{\bibfnamefont{T.~W.} \bibnamefont{{Baumgarte}}}
  \bibnamefont{and} \bibinfo{author}{\bibfnamefont{S.~L.}
  \bibnamefont{{Shapiro}}}, \bibinfo{journal}{\prd}
  \textbf{\bibinfo{volume}{59}}, \bibinfo{pages}{024007}
  (\bibinfo{year}{1999}), \eprint{arXiv:gr-qc/9810065}.

\bibitem[{\citenamefont{{Nagy} et~al.}(2004)\citenamefont{{Nagy}, {Ortiz}, and
  {Reula}}}]{nor}
\bibinfo{author}{\bibfnamefont{G.}~\bibnamefont{{Nagy}}},
  \bibinfo{author}{\bibfnamefont{O.~E.} \bibnamefont{{Ortiz}}},
  \bibnamefont{and} \bibinfo{author}{\bibfnamefont{O.~A.}
  \bibnamefont{{Reula}}}, \bibinfo{journal}{\prd}
  \textbf{\bibinfo{volume}{70}}, \bibinfo{pages}{044012}
  (\bibinfo{year}{2004}), \eprint{arXiv:gr-qc/0402123}.

\bibitem[{\citenamefont{{Frankel}}(2003)}]{frankel1}
\bibinfo{author}{\bibfnamefont{T.}~\bibnamefont{{Frankel}}},
  \emph{\bibinfo{title}{{The Geometry of Physics}}} (\bibinfo{publisher}{The
  Geometry of Physics, by Theodore Frankel, pp.~720.~ISBN
  0521833302.~Cambridge, UK: Cambridge University Press, November 2003.},
  \bibinfo{year}{2003}).

\bibitem[{\citenamefont{Goldstein}(2002)}]{goldstein}
\bibinfo{author}{\bibfnamefont{H.}~\bibnamefont{Goldstein}},
  \emph{\bibinfo{title}{Classical Mechanics}} (\bibinfo{publisher}{Addison
  Wesley}, \bibinfo{address}{Harlow}, \bibinfo{year}{2002}), ISBN
  \bibinfo{isbn}{0201657023}.

\bibitem[{\citenamefont{Nakahara}(2003)}]{nakahara}
\bibinfo{author}{\bibfnamefont{M.}~\bibnamefont{Nakahara}},
  \emph{\bibinfo{title}{Geometry, Topology, and Physics}}
  (\bibinfo{publisher}{Institute of Physics Publishing},
  \bibinfo{address}{Bristol}, \bibinfo{year}{2003}), ISBN
  \bibinfo{isbn}{0750306068}.

\bibitem[{\citenamefont{{Morrison}}(1998)}]{morrison1}
\bibinfo{author}{\bibfnamefont{P.~J.} \bibnamefont{{Morrison}}},
  \bibinfo{journal}{Reviews of Modern Physics} \textbf{\bibinfo{volume}{70}},
  \bibinfo{pages}{467} (\bibinfo{year}{1998}).

\bibitem[{\citenamefont{{Rovelli}}(2004)}]{rovelli1}
\bibinfo{author}{\bibfnamefont{C.}~\bibnamefont{{Rovelli}}},
  \emph{\bibinfo{title}{{Quantum Gravity}}} (\bibinfo{publisher}{Quantum
  Gravity, by Carlo Rovelli, pp.~480.~ISBN 0521837332.~Cambridge, UK: Cambridge
  University Press, November 2004.}, \bibinfo{year}{2004}).

\bibitem[{\citenamefont{{Jorge} and {Saletan}}(1998)}]{jorge1}
\bibinfo{author}{\bibfnamefont{J.}~\bibnamefont{{Jorge}}} \bibnamefont{and}
  \bibinfo{author}{\bibfnamefont{E.}~\bibnamefont{{Saletan}}},
  \emph{\bibinfo{title}{Classical Dynamics}} (\bibinfo{publisher}{Cambridge
  University Press}, \bibinfo{address}{Cambridge}, \bibinfo{year}{1998}), ISBN
  \bibinfo{isbn}{9780521636360}.

\bibitem[{\citenamefont{{Henneaux} and {Teitelboim}}(1992)}]{teitel1}
\bibinfo{author}{\bibfnamefont{M.}~\bibnamefont{{Henneaux}}} \bibnamefont{and}
  \bibinfo{author}{\bibfnamefont{C.}~\bibnamefont{{Teitelboim}}},
  \emph{\bibinfo{title}{{Quantization of Gauge Systems}}}
  (\bibinfo{publisher}{Princeton, USA: Univ. Pr., 1992, 520 p.},
  \bibinfo{year}{1992}).

\bibitem[{\citenamefont{Noether and Tavel}(2005)}]{noether1}
\bibinfo{author}{\bibfnamefont{E.}~\bibnamefont{Noether}} \bibnamefont{and}
  \bibinfo{author}{\bibfnamefont{M.~A.} \bibnamefont{Tavel}}
  (\bibinfo{year}{2005}), \eprint{arXiv:physics/0503066}.

\bibitem[{\citenamefont{{Dirac}}(1958{\natexlab{b}})}]{dirac2}
\bibinfo{author}{\bibfnamefont{P.~A.~M.} \bibnamefont{{Dirac}}},
  \bibinfo{journal}{Royal Society of London Proceedings Series A}
  \textbf{\bibinfo{volume}{246}}, \bibinfo{pages}{326}
  (\bibinfo{year}{1958}{\natexlab{b}}).

\bibitem[{\citenamefont{{Bergmann} and {Goldberg}}(1955)}]{bergmann1}
\bibinfo{author}{\bibfnamefont{P.~G.} \bibnamefont{{Bergmann}}}
  \bibnamefont{and}
  \bibinfo{author}{\bibfnamefont{I.}~\bibnamefont{{Goldberg}}},
  \bibinfo{journal}{Physical Review} \textbf{\bibinfo{volume}{98}},
  \bibinfo{pages}{531} (\bibinfo{year}{1955}).

\bibitem[{\citenamefont{{Anderson} and {Bergmann}}(1951)}]{bergmann2}
\bibinfo{author}{\bibfnamefont{J.~L.} \bibnamefont{{Anderson}}}
  \bibnamefont{and} \bibinfo{author}{\bibfnamefont{P.~G.}
  \bibnamefont{{Bergmann}}}, \bibinfo{journal}{Physical Review}
  \textbf{\bibinfo{volume}{83}}, \bibinfo{pages}{1018} (\bibinfo{year}{1951}).

\bibitem[{\citenamefont{{Dewitt} and {Dewitt}}(1952)}]{dewitt1}
\bibinfo{author}{\bibfnamefont{B.~S.} \bibnamefont{{Dewitt}}} \bibnamefont{and}
  \bibinfo{author}{\bibfnamefont{C.~M.} \bibnamefont{{Dewitt}}},
  \bibinfo{journal}{Physical Review} \textbf{\bibinfo{volume}{87}},
  \bibinfo{pages}{116} (\bibinfo{year}{1952}).

\bibitem[{\citenamefont{{Salisbury}}(2008)}]{salisbury1}
\bibinfo{author}{\bibfnamefont{D.~C.} \bibnamefont{{Salisbury}}}, in
  \emph{\bibinfo{booktitle}{The Eleventh Marcel Grossmann Meeting On Recent
  Developments in Theoretical and Experimental General Relativity, Gravitation
  and Relativistic Field Theories}}, edited by
  \bibinfo{editor}{\bibnamefont{{H.~Kleinert, R.~T.~Jantzen, \& R.~Ruffini}}}
  (\bibinfo{year}{2008}), pp. \bibinfo{pages}{2467--2469},
  \eprint{arXiv:physics/0701299}.

\bibitem[{\citenamefont{{G{\"o}ckeler} and {Sch{\"u}cker}}(1989)}]{diffgeo1}
\bibinfo{author}{\bibfnamefont{M.}~\bibnamefont{{G{\"o}ckeler}}}
  \bibnamefont{and}
  \bibinfo{author}{\bibfnamefont{T.}~\bibnamefont{{Sch{\"u}cker}}},
  \emph{\bibinfo{title}{{Differential Geometry, Gauge Theories, and Gravity}}}
  (\bibinfo{publisher}{Differential Geometry, Gauge Theories, and Gravity, by
  M.~G{\"o}ckeler and T.~Sch{\"u}cker, pp.~248.~ISBN 0521378214.~Cambridge, UK:
  Cambridge University Press, July 1989.}, \bibinfo{year}{1989}).

\bibitem[{\citenamefont{{Ivey} and {Landsberg}}(2003)}]{cfb}
\bibinfo{author}{\bibfnamefont{T.}~\bibnamefont{{Ivey}}} \bibnamefont{and}
  \bibinfo{author}{\bibfnamefont{J.}~\bibnamefont{{Landsberg}}},
  \emph{\bibinfo{title}{Cartan for Beginners}} (\bibinfo{publisher}{American
  Mathematical Society}, \bibinfo{address}{Providence}, \bibinfo{year}{2003}),
  ISBN \bibinfo{isbn}{0821833758}.

\bibitem[{\citenamefont{{Gustafsson} and {Oliger}}(1995)}]{kgo}
\bibinfo{author}{\bibfnamefont{H.}~\bibnamefont{{Gustafsson}},
  \bibfnamefont{B.~{Kreiss}}} \bibnamefont{and}
  \bibinfo{author}{\bibfnamefont{J.}~\bibnamefont{{Oliger}}},
  \emph{\bibinfo{title}{Time Dependent Problems and Difference Methods}}
  (\bibinfo{publisher}{Wiley}, \bibinfo{address}{New York},
  \bibinfo{year}{1995}), ISBN \bibinfo{isbn}{0471507342}.

\bibitem[{\citenamefont{Brown}(2008)}]{brown3}
\bibinfo{author}{\bibfnamefont{J.~D.} \bibnamefont{Brown}}
  (\bibinfo{year}{2008}), \eprint{0803.0334}.

\bibitem[{\citenamefont{{Brown} and {Lowe}}(2006)}]{brown2}
\bibinfo{author}{\bibfnamefont{J.~D.} \bibnamefont{{Brown}}} \bibnamefont{and}
  \bibinfo{author}{\bibfnamefont{L.~L.} \bibnamefont{{Lowe}}},
  \bibinfo{journal}{\prd} \textbf{\bibinfo{volume}{74}},
  \bibinfo{pages}{104023} (\bibinfo{year}{2006}), \eprint{arXiv:gr-qc/0606008}.

\bibitem[{\citenamefont{{Weinberg}}(1995)}]{weinberg1}
\bibinfo{author}{\bibfnamefont{S.}~\bibnamefont{{Weinberg}}},
  \emph{\bibinfo{title}{{The quantum theory of fields. Vol.1: Foundations}}}
  (\bibinfo{publisher}{Cambridge, New York: Cambridge University Press,
  |c1995}, \bibinfo{year}{1995}).

\bibitem[{\citenamefont{Reula}(2004)}]{reula1}
\bibinfo{author}{\bibfnamefont{O.~A.} \bibnamefont{Reula}}
  (\bibinfo{year}{2004}), \eprint{gr-qc/0403007}.

\end{thebibliography}


\end{document}